\documentclass{article}
\usepackage{amsmath,amsfonts,amssymb,amsthm,graphicx}
 \usepackage{dsfont}    
 \usepackage{bm} 
\usepackage{enumerate}
\usepackage{color}
\usepackage{natbib}
\usepackage{multirow}
\usepackage{comment}
\usepackage{adjustbox}

\newcommand{\cP}{\mathcal{P}}
\newcommand{\cK}{\mathcal{K}} 
 
\renewcommand{\epsilon}{\varepsilon}

\usepackage{booktabs,subcaption,dcolumn}
\usepackage{tabularx}

\usepackage{algorithm}
\usepackage[noend]{algpseudocode}  
\usepackage[colorlinks=true,citecolor=blue,pdfpagemode=UseNone,pdfstartview=FitH]{hyperref}

\usepackage{tikz}
\usetikzlibrary{arrows.meta,positioning,calc}

\newif\ifFULL
\ifFULL

\fi

\newif\ifRuodu
\ifRuodu
  \renewcommand{\ge}{\geqslant}
  \renewcommand{\le}{\leqslant}
  \renewcommand{\epsilon}{\varepsilon}
  
\fi

\newcommand{\p}{\mathbb{P}}

\newcommand{\cebh}{\overline{\mathrm{eBH}}}

\newcommand{\E}{\mathbb{E}}    
\newcommand{\N}{\mathbb{N}}    

\newcommand{\FDP}{\mathrm{FDP}}
\newcommand{\FDR}{\mathrm{FDR}}

\usepackage{setspace}

\theoremstyle{plain}
\newtheorem{theorem}{Theorem}[section]

\newtheorem{lemma}[theorem]{Lemma}
\newtheorem{proposition}[theorem]{Proposition}

\theoremstyle{definition}
\newtheorem{definition}[theorem]{Definition}
\newtheorem{example}[theorem]{Example}

\theoremstyle{remark}
\newtheorem{remark}[theorem]{Remark}

\newenvironment{proofidea}
{\par\smallskip\noindent\textit{Proof sketch.}\ }
{\par\smallskip}

 \renewcommand{\cite}{\citet}  

\title{Admissibility and Complete Classes for False Discovery Rate Control with E-values}
\author{
Liulei Sun\thanks%
 {Department of Statistics and Actuarial Science,
 University of Waterloo,
 Waterloo, Ontario, Canada.
 E-mail: \href{mailto:l229sun@uwaterloo.ca}{l229sun@uwaterloo.ca}.} \and
 Ruodu Wang\thanks%
 {Department of Statistics and Actuarial Science,
 University of Waterloo,
 Waterloo, Ontario, Canada.
 E-mail: \href{mailto:wang@uwaterloo.ca}{wang@uwaterloo.ca}.}}
\date{September 2026}

\begin{document} 
\maketitle  

\begin{abstract}

The false discovery rate (FDR) is the most widely used error metric in modern multiple testing. We provide a comprehensive analysis of the admissibility of  e-value-based procedures with FDR control. We consider both simultaneous and point procedures and introduce strong and weak notions of dominance. 
Every simultaneous procedure is strongly, and hence weakly, dominated by an admissible weighted-mean closed e-Benjamini--Hochberg ($\overline{\mathrm{eBH}}$) procedure, and thus weighted-mean $\overline{\mathrm{eBH}}$ procedures form a complete class. Every constant-free weighted-mean $\overline{\mathrm{eBH}}$ procedure is admissible at every level, and we propose two ways for choosing constant-free weights  based on side information, such as pre-screening data.  
Within the symmetric class,  the usual mean $\overline{\mathrm{eBH}}$ procedure is the largest element if and only if the FDR level is small enough; otherwise this class has no largest element. We also obtain results on the admissibility of symmetric weighted-mean $\overline{\mathrm{eBH}}$ procedures with non-zero constant terms. Point e-testing procedures have a parallel theory of admissibility.
These results highlight the central role of weighted-mean $\overline{\mathrm{eBH}}$ procedures in multiple testing.\bigskip

\noindent \textbf{Keywords:} Multiple testing, eBH procedure, BY procedure, arbitrary dependence, merging e-values
\end{abstract}

\section{Introduction}

The false discovery rate (FDR), introduced by \citet{BH95}, is one of the central error criteria in modern multiple testing. The literature now contains many ways to control FDR with p-values: the Benjamini--Hochberg (BH) procedure under independence \citep{BH95}, the Benjamini--Yekutieli procedure under arbitrary dependence \citep{BY01}, adaptive and q-value methods \citep{S02}, self-consistent procedures \citep{BR08}, and grouped or multilayer procedures \citep{BR17}. 
These procedures may be compared based on their performance, commonly referred to as power. 
For instance, if two different procedures both control FDR at the same level, but one procedure always rejects at least as many hypotheses as the other and sometimes rejects more, then the first  procedure dominates the latter in terms of power, and should be preferred  in practice.

The procedures that are not dominated by others are called admissible. Admissibility is a minimal optimality requirement imposed after error control has been established. 
Recent studies have shown several existing procedures to be inadmissible by constructing uniform improvements.
For instance, \cite{SG17} provided an improvement of the BH procedure.  \cite{IWR26} used compound e-values to obtain tiny but uniform improvements of several adaptive BH procedures, and  \cite{Goe26} presented a closed BH procedure that uniformly improves the BH procedure. 
To our knowledge, despite the central role of FDR for multiple testing, there is no general complete class theory, nor a general admissibility characterization, for FDR-controlling procedures in the literature. 
Admissibility results are available for other statistical objects, such as the tail probabilities of the false discovery proportion \citep{GHS21}, merging p-values under arbitrary dependence \citep{VWW22}, and combination of e-values \citep{W25,C26}.

In this paper, we provide the first  results on admissibility and characterization of complete classes for FDR-controlling procedures, in the framework of testing procedures based on e-values. In other words, we identify a special class of procedures such that every valid procedure can be uniformly improved to an admissible member of this class; consequently, no admissible procedure lies outside it. Our construction relies on the procedures recently developed by \cite{XSF25}, which we explain below.

E-values \citep{VW21, S21, GDK24} have been shown to be central in the study of FDR-controlling procedures, either through explicit procedures, or serving as intermediate tools.    
\citet{WR22} proposed the  e-Benjamini--Hochberg (eBH) procedure, which controls FDR for arbitrary input e-values.
\cite{IWR24,IWR26} demonstrated that the eBH procedure for some implicit compound e-values is behind every FDR-controlling procedure (not necessarily based on input e-values). \citet{LR24} proposed conditional calibration to boost eBH,  and \cite{RB24} used eBH to derandomize knockoff procedures.  We refer to the  recent monograph \cite{RW25} for a comprehensive treatment.

\cite{XSF25} proposed the closed eBH procedures, which offer a powerful class of procedures that dominate eBH and improve upon several other classes. Using closed eBH procedures, they provided a general representation of FDR-controlling procedures, while leaving a precise characterization open. We take the next step by introducing a special class, called the \emph{weighted-mean closed eBH procedure}, that plays a central role in the admissibility of e-testing procedures:
\begin{enumerate}[(a)]
    \item Many explicit members of this class are admissible for practical FDR levels;
    \item it forms a complete class for both simultaneous and point procedures. 
\end{enumerate} 
Our paper offers many more results that complement these two main conclusions. In particular,  specific choices of the parameters in the weighted-mean closed eBH procedures matter for both admissibility and empirical performance of the procedures, which have not been explored in the existing literature.

Our results are formulated under the setting that the input statistics are arbitrary e-values. These results do not give conclusions on the admissibility of FDR-controlling procedures based on p-values  (such as the BH procedure)
or procedures based on independent e-values (e.g.,  \citealp{VW24}). Studying admissibility for these procedures would likely require completely different techniques, and it is outside the scope of this paper.

\subsection{Contributions}

Our main conceptual contribution is to establish a mathematical framework
for comparing (both simultaneous and point) e-testing procedures in Section~\ref{sec:setting}, and to highlight the theoretical importance of admissibility for e-testing procedures.

On the methodological side, 
we propose the weighted-mean $\overline{\mathrm{eBH}}$ procedures as an analytically tractable  class of simultaneous FDR-controlling procedures in Section~\ref{sec:close-eBH}. We  identify conditions for such procedures to dominate the base eBH procedure (Proposition~\ref{prop:esc-dominated-by-weighted-ebh}). 
Further, we propose two practical ways of choosing the weights in the weighted-mean $\overline{\mathrm{eBH}}$ procedures based on side information in Section~\ref{sec:prior-score-weights},
as well as a rule of thumb in the absence of side information in Section~\ref{sec:constant-term}.
Simulation studies in Section~\ref{sec:simulations} illustrate the advantages of these methods. 

Our technical contributions can be split into three major categories, depending on the procedures that we study.
\begin{enumerate}
    \item \textbf{On simultaneous e-testing procedures.} We offer two major results that clarify the 
    structure of e-testing procedures that simultaneously produce multiple rejection sets.
    
\begin{enumerate}
    \item Every constant-free weighted-mean $\overline{\mathrm{eBH}}$ procedure is admissible at every level (Theorem~\ref{thm:constant-free-weighted-ebh-admissible}).  The two ways for choosing constant-free weights based on side information in Section~\ref{sec:prior-score-weights} both yield admissible simultaneous procedures.
    
    \item At any given FDR level, every simultaneous FDR-controlling e-testing procedure is strongly dominated, and hence weakly dominated, by an admissible weighted-mean $\overline{\mathrm{eBH}}$ procedure (Theorem~\ref{thm:admissible-weighted-ebh-dominator}). Thus weighted-mean $\overline{\mathrm{eBH}}$ procedures form a complete class of simultaneous procedures.   
\end{enumerate}

    \item \textbf{On point e-testing procedures.} We develop a parallel theory for e-testing procedures that output a single rejection set.

\begin{enumerate}
    \item Every constant-free and strictly positive point weighted-mean $\overline{\mathrm{eBH}}$ procedure is admissible when the level is less than $1/2$ (Theorem~\ref{thm:positive-point-weighted-ebh-admissible}). This threshold is sharp for the point mean $\overline{\mathrm{eBH}}$ procedure (Proposition~\ref{prop:point-mean-admissibility-threshold}).
    
    \item At any given FDR level, every point FDR-controlling e-testing procedure is dominated by an admissible point weighted-mean $\overline{\mathrm{eBH}}$ procedure (Theorem~\ref{thm:point-weighted-ebh-complete-class}). Thus, point weighted-mean $\overline{\mathrm{eBH}}$ procedures form a complete class of point procedures. 
\end{enumerate}
    
    \item \textbf{On symmetric e-testing procedures.} The additional results for symmetric e-testing procedures are threefold.
    
\begin{enumerate}
    \item We characterize the symmetry of weighted-mean $\overline{\mathrm{eBH}}$ procedures and study their related properties (Section~\ref{sec:symmetric_characterization}).

    \item We identify exactly when the mean $\overline{\mathrm{eBH}}$ procedure and the point mean $\overline{\mathrm{eBH}}$ procedure are the largest elements of their respective symmetric classes (Section~\ref{sec:mean-procedures}).
  
    \item For symmetric weighted-mean $\overline{\mathrm{eBH}}$ procedures with constant terms, we give a sufficient condition for admissibility and a necessary and sufficient condition for admissibility of an explicit one-parameter family (Section~\ref{sec:constant-term}). 
\end{enumerate}    

\end{enumerate} 
 In the main paper, we provide the most important proofs,  some short proofs, and  sketches for long  technical proofs.  
All  omitted proofs are presented in Appendix~\ref{app:proofs}.


\section{E-testing procedures and their dominance}\label{sec:setting}

This section sets up the decision-theoretic language used in the paper. 
Throughout, we assume the axiom of choice.\footnote{Two results rely on it: Theorem~\ref{thm:admissible-weighted-ebh-dominator} and  Theorem~\ref{thm:point-weighted-ebh-complete-class}. All other results hold without the axiom of choice.}

We first define e-variables and e-values, 
and then e-testing procedures and their dominance relations.
Let $\cP$ be a set of probability measures representing the null hypothesis.

\begin{definition}
    An \emph{e-variable} $E$ for $\cP$ is a $[0,\infty]$-valued random variable satisfying $\E^{\mathbb P}[E]\le 1$ for all $\mathbb P\in \cP$.
\end{definition}

 We consider the problem of testing $K\ge 2$ hypotheses, represented by sets $\cP_1,\dots,\cP_K$ of probability measures. Write $\cK=\{1,\dots,K\}$ and  $\mathbf E=(E_1,\dots,E_K)$ for a vector of e-values for $(\cP_1,\dots,\cP_K)$. Precisely, for each $k$, $E_k$ is an e-variable for $\cP_k$, and the variables $E_k$ are also called e-values. 
A hypothesis $\cP_k$ is a true null if $\p\in \cP_k$, where $\p$ is the unknown data-generating probability. 
The set of true null hypotheses will be denoted by $N$.
For any $A\subseteq \cK$, we write $\mathbf E^A=(E_i)_{i\in A}$.
The vector $\mathbf E^N$ is called the vector of null e-values. 
Note that $\E[E_i]\le1$ for all $i\in N$, where the expectation is computed under $\p$. 
No additional independence or prior model is assumed for the vector of null e-values.
We write $\mathbf e=(e_1,\dots,e_K)$ whenever there is no confusion of the dimension; also denote by the
sub-vector $\mathbf e^A=(e_i)_{i\in A}$ and
its average $\bar e_A=\sum_{i\in A}e_i / |A|$ for any nonempty $A\subseteq \cK$. The same notation will be applied to other vectors like $\mathbf x$ and $\mathbf y$.


We take the framework of \cite{XSF25} and allow testing procedures to output multiple sets of rejected hypotheses. 
A \emph{simultaneous (e-testing) procedure} $\mathcal D$ is a measurable mapping from $[0,\infty]^K$ to $2^{(2^{\cK})}$. 
The input of $\mathcal D$ is a vector of e-values, and the output of $\mathcal D$ is a collection of sets of hypotheses that are rejected (also called certified) based on the input. A \emph{point (e-testing) procedure} $\mathfrak D$ is a measurable mapping from $[0,\infty]^K$ to $2^{\cK}$. The input of $\mathfrak D$ is a vector of e-values, and the output of $\mathfrak D$ is one set of hypotheses that is rejected based on the e-values.  
We omit ``e-testing'' from the terminology for simplicity.  
Moreover, we say that $\mathcal D$ (or $\mathfrak D$) is \emph{symmetric} if $\mathcal D(\mathbf e_\sigma)=\sigma^{-1} (\mathcal D(\mathbf e))$ for all $\mathbf e\in [0,\infty]^K$ and all permutations $\sigma$ on $\cK$, where $\mathbf e_\sigma=(e_{\sigma(1)},\dots,e_{\sigma(K)})$.

The distinction between simultaneous and point procedures is not only notational. A simultaneous procedure describes all rejection sets that are certified by the method, while a point procedure must choose one of them for reporting. To better illustrate these two kinds of procedures, we introduce the self-consistent procedure and the eBH procedure \citep{WR22}, which serve as the first examples of simultaneous and point procedures.

\begin{definition}
Given a level $\alpha\in (0,1)$, the \emph{self-consistent procedure} is defined to reject
$$\mathrm{eSC}_\alpha(\mathbf e)=\{\varnothing\}\cup\left\{\varnothing\neq R\subseteq \mathcal K: \min_{k\in R} e_k\ge \frac{K}{\alpha |R|}\right\}.$$
The \emph{eBH procedure}  is a point procedure that rejects the unique largest set in $\mathrm{eSC}_\alpha$, that is, 
$$\mathrm{eBH}_\alpha(\mathbf e)=\bigcup_{R\in \mathrm{eSC}_\alpha(\mathbf e)} R.$$
\end{definition}

It is clear that eSC and eBH are symmetric procedures for any $\alpha$.
 The eBH procedure is often alternatively described as rejecting
all hypotheses with the largest $k^*$ e-values, where
$$k^*=\max\left\{r\in \cK: \frac{r e_{(r)}}{K}\ge \frac{1}{\alpha}\right\},$$ 
with the convention $\max(\varnothing)=0$ and $e_{(1)}\ge\cdots\ge e_{(K)}$ being the decreasing order statistics of $\mathbf e$. 


Note that we can uniquely associate $\mathrm{eBH}_\alpha$ with the simultaneous procedure $\mathbf e\mapsto \{\mathrm{eBH}_\alpha(\mathbf e)\}$. 
In general, a point procedure can be identified with a simultaneous procedure that outputs singletons, and therefore the terminology for simultaneous procedures carries over to point procedures naturally. 
Nevertheless, the admissibility structures of point procedures and simultaneous procedures are quite different, as we will see in Sections~\ref{sec:simultaneous-fdr}--\ref{sec:point-procedures}.

\begin{remark}
The domains of e-testing procedures and point procedures are taken to be $[0,\infty]^K$ rather than $[0,\infty)^K$. This is consistent with the usual definition of an e-variable as a $[0,\infty]$-valued random variable. Under a null hypothesis, an e-value is finite almost surely. Allowing $\infty$ in the input is convenient in several proofs and examples. There is no additional difficulty to restrict attention to finite e-values, and we omit that analysis.
\end{remark}

For two sets $A,R\subseteq \cK$, the false discovery proportion (FDP) is defined as 
$$
\FDP_A(R) = \frac{|A \cap R|}{|R|\vee 1},
$$
where $A$ represents any set of hypotheses and $R$ represents the set of rejected hypotheses.
The false discovery rate (FDR) of an e-testing procedure $\mathcal D$ for the input random vector $\mathbf E$ is defined as
$$
\FDR_\mathcal D(\mathbf E)= \E\left[\max_{R\in \mathcal D(\mathbf E)}\FDP_N(R)\right],
$$ 
where the maximum over an empty collection is defined to be $0$, and $\mathcal D(\mathbf E)$ is a set of possible rejection sets. 
An e-testing procedure has FDR at level $\alpha \in (0,1)$ if
$$
\FDR_\mathcal D(\mathbf E)\le \alpha,
$$
for all $\p$, $(\cP_1,\dots,\cP_K)$, and all input e-values $\mathbf E$; note that the above definition of $\FDR_\mathcal D(\mathbf E)$ depends on all of them. The FDR of a point procedure $\mathfrak D$ for the input random vector $\mathbf E$ is defined as
$$
\FDR_\mathfrak D(\mathbf E)= \E\left[ \FDP_N(
\mathfrak D(\mathbf E))\right].
$$ 
A point procedure has FDR at level $\alpha \in (0,1)$ if
$$
\FDR_\mathfrak D(\mathbf E)\le \alpha,
$$
for all $\p$, $(\cP_1,\dots,\cP_K)$, and all input e-values $\mathbf E$. 

We denote by $\mathrm{SP}_\alpha$  
(resp.~$\mathrm{PP}_\alpha$)
the class of all simultaneous procedures 
(resp.~point procedures) 
that control FDR at level $\alpha$. Let $\mathrm{SP}_\alpha^{\mathrm{sym}}$ and $\mathrm{PP}_\alpha^{\mathrm{sym}}$ denote all the symmetric elements in $\mathrm{SP}_\alpha$ and $\mathrm{PP}_\alpha$, respectively.
For instance, $$\mathrm{eSC}_\alpha \in \mathrm{SP}_\alpha^{\mathrm{sym}}\subseteq \mathrm{SP}_\alpha;
\qquad
\mathrm{eBH}_\alpha \in \mathrm{PP}_\alpha^{\mathrm{sym}}\subseteq \mathrm{PP}_\alpha;
$$
see \cite{WR22} and \cite{XSF25} for the FDR guarantee of these procedures.

Given that there are many procedures that control FDR at a given level, a criterion for ``better" procedures is through a notion of power comparison. For point procedures this is straightforward: rejecting a larger set at every input is seen as being superior in power. For simultaneous procedures, since their outputs are collections of rejection sets, there are two natural notions of comparison that favor a procedure $\mathcal G$ 
over a procedure $\mathcal D$, which need careful formulation. 
In the first sense, the rejected sets in  $\mathcal G$ completely cover those in $\mathcal D$, that is, $\mathcal D(\mathbf e) \subseteq \mathcal G(\mathbf e)$ for all $\mathbf e$. 
In the second sense, any set rejected by $\mathcal D$ 
is contained by  one rejected by $\mathcal G$, 
that is, for any $R\in \mathcal D(\mathbf e)$
there exists $R'\in \mathcal G(\mathbf e)$ such that $R\subseteq R'$, for all $\mathbf e$. This means $\mathcal D^\downarrow(\mathbf e) \subseteq \mathcal G^\downarrow(\mathbf e)$, where we write $$\mathcal C^\downarrow=\{B\subseteq\cK:\exists A\in\mathcal C\text{ such that }B\subseteq A\},$$
for any set $\mathcal  C\subseteq 2^{\cK}$.
Note that when $\mathcal D^\downarrow(\mathbf e)  = \mathcal G^\downarrow(\mathbf e)$,
it may still be possible that $\mathcal D  (\mathbf e)  \subsetneq  \mathcal G(\mathbf e)$,
and in such a case we favor $\mathcal G$ over $\mathcal D$.
The definitions below formalize the above ideas.

\begin{definition}
\label{def:dominance}
Let $\mathcal D,\mathcal G$ be two simultaneous procedures and $\mathfrak D,\mathfrak G$ be two point procedures.
\begin{enumerate}[(i)]
\item $\mathcal G$ \emph{strongly dominates} $\mathcal D$, written $\mathcal D\preceq_{\mathrm s}\mathcal G$, if $\mathcal D(\mathbf e)\subseteq\mathcal G(\mathbf e)$ for every $\mathbf e\in[0,\infty]^K$;
\item \(\mathcal G\) \emph{weakly dominates} \(\mathcal D\), written
\(\mathcal D\preceq_{\mathrm w}\mathcal G\), if, for every $\mathbf e\in[0,\infty]^K$, either $\mathcal D^\downarrow(\mathbf e) \subsetneq \mathcal G^\downarrow(\mathbf e)$ or $\mathcal D(\mathbf e)\subseteq\mathcal G(\mathbf e)$ holds. 
\item $\mathfrak G$ \emph{dominates} $\mathfrak D$, written $\mathfrak D\subseteq \mathfrak G$, if $\mathfrak D(\mathbf e)\subseteq \mathfrak G(\mathbf e)$ for every $\mathbf e\in[0,\infty]^K$.
\end{enumerate}
\emph{Strict dominance} between two procedures means dominance but not equality.
\end{definition}

Viewing point procedures as simultaneous procedures, the dominance $\mathfrak D\subseteq \mathfrak G$ is equivalent to weak dominance. Strong dominance never occurs between two different point procedures.

To make the implications of dominance meaningful, the procedures at comparison should have the same FDR level. This requirement is not needed to define dominance, but it is needed to define admissibility later.

\begin{proposition}
\label{prop:dominance-orders}
    The two relations have the following properties:
    \begin{itemize}
    \item[(i)] Both $\preceq_{\mathrm s}$ and $\preceq_{\mathrm w}$ are reflexive, antisymmetric, and transitive. Thus they are partial orders.
    
    \item[(ii)] $\mathcal D\preceq_{\mathrm s}\mathcal G$ implies $\mathcal D\preceq_{\mathrm w}\mathcal G$, but the converse is false.
    \end{itemize}
\end{proposition}

The proof of Proposition~\ref{prop:dominance-orders} is  presented in Appendix~\ref{app:A1}.

Since $\preceq_{\mathrm s}$ is a partial order, 
$\mathcal D\preceq_{\rm s} \mathcal D'$
and 
$\mathcal D'\preceq_{\rm s} \mathcal D$
jointly force $\mathcal D = \mathcal D'$. The situation for $\preceq_{\mathrm w}$ is similar.
Therefore, we can define admissibility in the following concise form.

\begin{definition}
An e-testing procedure $\mathcal D\in\mathrm{SP}_\alpha$ is said to be \emph{admissible (at FDR level $\alpha$)}  if, for every $\mathcal D'\in\mathrm{SP}_\alpha$, $\mathcal D\preceq_{\mathrm w}\mathcal D'$ implies $\mathcal D' = \mathcal D$. A point procedure $\mathfrak D\in \mathrm{PP}_\alpha$ is said to be \emph{admissible} if, for every $\mathfrak D'\in\mathrm{PP}_\alpha$, $\mathfrak D\subseteq \mathfrak D'$ implies $\mathfrak D' = \mathfrak D$.
\end{definition}

An admissible procedure  cannot be  uniformly improved, and hence it is a natural requirement when choosing suitable procedures in statistical applications.


Since we have two distinct versions of dominance for simultaneous procedures, it is reasonable to also define 
another version of admissibility. We say that $\mathcal D\in\mathrm{SP}_\alpha$ is \emph{s-admissible} 
if, for every $\mathcal D'\in\mathrm{SP}_\alpha$,  $\mathcal D\preceq_{\mathrm s}\mathcal D'$ implies $\mathcal D' = \mathcal D$.
The next result, perhaps surprisingly, shows that there is no need to distinguish these two forms of admissibility.

\begin{proposition} 
\label{prop:strong-weak-admissibility-equivalence}
For $\alpha\in (0,1)$, s-admissibility and admissibility coincide on $\mathrm{SP}_\alpha$. 
\end{proposition}

The proof of Proposition \ref{prop:strong-weak-admissibility-equivalence} is quite sophisticated and requires detailed analysis on the two partial orders, which are distinct even though they induce equivalent notions of admissibility.

In the rest of this paper, we sometimes say that a procedure $\mathcal D$ is admissible without explicitly mentioning the level $\alpha$ or the set $\mathrm{SP}_\alpha$. In such situations, the level $\alpha$ should be clear from the definition of $\mathcal D$. For instance, when we discuss the admissibility of $\mathrm{eSC}_{\alpha}$, 
the set  for comparison is $\mathrm{SP}_{\alpha}$, although  $\mathrm{eSC}_{\alpha}$ is also a member of $\mathrm{SP}_{\beta}$ for $\beta>\alpha$.
The next result justifies that there is no ambiguity caused by omitting $\alpha$ in admissibility statements. 

\begin{proposition}
\label{prop:admissibility-level-issue}
For $\alpha\in(0,1)$, if a member of $\mathrm{SP}_{\alpha}$ is admissible within $\mathrm{SP}_{\alpha}$, then it cannot be a member of $\mathrm{SP}_{\beta}$ for any $\beta\in (0,\alpha)$.
\end{proposition}

The analogous statement in Proposition~\ref{prop:admissibility-level-issue} also holds for $\mathrm{PP}_\alpha$, and hence
we follow the same convention when mentioning admissibility within $\mathrm{PP}_\alpha$.


\section{Closed eBH procedures}\label{sec:close-eBH}

This section introduces the concrete family of procedures that will turn out to organize the admissibility theory. 

For every $n\in \N$, denote by $\Delta_n$ the standard simplex in $\mathbb R^n$, that is,
$$\Delta_n=\left\{(\lambda_1,\dots,\lambda_n)\in [0,1]^n: \sum_{i=1}^{n}\lambda_i=1\right\}.$$
For $\bm{\lambda}=(\lambda_0,\dots,\lambda_n)\in \Delta_{n+1}$, define the weighted mean function 
$$\mathbb M_{\bm{\lambda}}(\mathbf e)=
\begin{cases}
\lambda_0+\sum_{i=1}^{n}\lambda_i e_i,
& \max (\mathbf e) <\infty,\\
\infty,
& \text{otherwise}.
\end{cases}
$$
Further, for $\lambda \in [0,1]$,  define
$$\mathbb M_{\lambda}^{\mathrm{sym}}(\mathbf e)=
\begin{cases}
\lambda+(1-\lambda)\sum_{i=1}^{n}e_i / n,
& \max  (\mathbf e) <\infty,\\
\infty,
& \text{otherwise},
\end{cases}
$$
which corresponds to $\mathbb M_{\bm{\lambda}}$ with $\lambda_1=\dots=\lambda_n=(1-\lambda)/n$.
The functions $\mathbb M_{\bm{\lambda}}$ and $\mathbb M_{\lambda}^{\mathrm{sym}}$ are e-merging functions \citep{VW21}, and $\mathbb M_{\bm{\lambda}}, \bm{\lambda}\in \Delta_{n+1}$ are the only admissible e-merging functions for arbitrary e-values \citep{W25}. The reason for treating infinity separately is explained in Appendix C in \citet[Online supplement]{VW21}, as an admissible e-merging function should assign the value $\infty$ whenever at least one input is $\infty$.


We next explain the closed eBH procedure of \cite{XSF25}. 
For any $A\subseteq \cK$, define the intersection hypothesis $\cP_A=\bigcap_{k\in A}\cP_k$. The collection $\{E_A\}_{A\subseteq \cK}$ is called an \emph{e-collection} if $E_A$ is an e-variable for $\cP_A$ for every $A\subseteq \cK$. We take $E_\varnothing=1$.  

Given an FDR level $\alpha\in (0,1)$, the \emph{closed eBH procedure} ($\overline{\mathrm{eBH}}$) based on the e-collection $\{E_A\}_{A\subseteq \cK}$ is defined as
$$\overline{\mathrm{eBH}}_\alpha=\left\{R\subseteq \cK:E_A\ge\frac{\mathrm{FDP}_A(R)}{\alpha}\text{ for all } A\subseteq \cK\right\}.$$
Note that this procedure does not belong to
$\mathrm{SP}_\alpha$
in our framework, because its input is not a vector in $[0,\infty]^K$, but a collection $\{E_A\}_{A\subseteq \cK}$ of e-variables. Nevertheless,
in the primary examples of $\overline{\mathrm{eBH}}$, 
the collection $\{E_A\}_{A\subseteq \cK}$ 
can be computed as functions of $K$ input e-values, thereby giving rise to elements of $\mathrm{SP}_\alpha$. 
In particular, via the weighted-mean e-collection, we can define the main procedures in our study.


\begin{definition}\label{def:closed-eBH-procedure}
For $\bm \lambda=(\bm \lambda^A)_{A\subseteq \cK},
$ 
where $\bm \lambda^A=(\lambda_0^A,(\lambda_i^A)_{i\in A})\in \Delta_{|A|+1}$,
define the \emph{weighted-mean $\overline{\mathrm{eBH}}$ procedure} as the simultaneous procedure given by
\begin{equation}
\label{eq:closed-def}
\overline{\mathrm{eBH}}^{\bm \lambda}_\alpha(\mathbf e)=\left\{R\subseteq \cK:\mathbb M_{\bm{\lambda}^A}(\mathbf e^{A}) \ge\frac{\mathrm{FDP}_A(R)}{\alpha}\text{ for all } A\subseteq \cK\right\}.
\end{equation}
As a special case, 
the \emph{mean $\overline{\mathrm{eBH}}$ procedure} is defined by 
%
$$\overline{\mathrm{eBH}}_\alpha^{\mathsf{m}}(\mathbf e)=\left\{R\subseteq \cK:\mathbb M_{0}^{\rm sym}(\mathbf e^{A}) \ge\frac{\mathrm{FDP}_A(R)}{\alpha}\text{ for all } A\subseteq \cK\right\}.$$ 
 
\end{definition}

\begin{remark} The parameter $\bm \lambda$ in \eqref{eq:closed-def} has $K2^{K-1}$ degrees of freedom, and hence its choice is a highly nontrivial task, which we will discuss in later sections. 
Different choices of the parameter $\bm\lambda$ may define the same weighted-mean $\overline{\mathrm{eBH}}$ procedure; see Example~\ref{ex:same_procedure_different_lambda} in Appendix~\ref{app:counterexamples}. 
\end{remark}


For every e-collection, the corresponding $\overline{\mathrm{eBH}}$ procedure $\overline{\mathrm{eBH}}_\alpha$ belongs to $\mathrm{SP}_\alpha$, as shown by \cite{XSF25}. A weighted-mean $\overline{\mathrm{eBH}}$ procedure is \emph{constant-free} if $\lambda_0^A=0$ for each nonempty $A\subseteq\cK$; a weighted-mean $\overline{\mathrm{eBH}}$ procedure is \emph{strictly positive} if $\lambda_i^A>0$ for all $i\in A$ and for each nonempty $A\subseteq\cK$. 
Every weighted-mean $\overline{\mathrm{eBH}}$ procedure is increasing in the sense of $\mathcal D(\mathbf x)\subseteq\mathcal D(\mathbf y)$ whenever $\mathbf x\le\mathbf y$. 


\begin{definition}\label{def:point-procedure}
A \emph{point weighted-mean} (resp.~\emph{mean}) \emph{$\overline{\mathrm{eBH}}$ procedure}, denoted generically by $\overline{\mathrm{ebh}}_\alpha^{\bm \lambda}$ (resp.~$\overline{\mathrm{ebh}}_\alpha^{\mathsf{m}}$), is a point e-testing procedure given by rejecting a maximal set in the corresponding weighted-mean (resp.~mean) $\overline{\mathrm{eBH}}$ procedure.
\end{definition}

In the original definition of point $\overline{\mathrm{eBH}}$ procedure \citep{XSF25}, it rejects one set of the largest size in the corresponding $\overline{\mathrm{eBH}}$ procedure. That definition, which is based on the size-induced order, is slightly different from Definition~\ref{def:point-procedure}, which is based on an inclusion-induced partial order. For example, if the $\overline{\mathrm{eBH}}$ procedure rejects $\{\{1,2\},\{1,3,4\},\{2,3,4\}\}$, the size-based point $\overline{\mathrm{eBH}}$ procedure rejects either $\{1,3,4\}$ or $\{2,3,4\}$, whereas our definition can also reject $\{1,2\}$. There are two reasons for this upgrade.
First, it is a priori unclear whether rejecting $\{1,2\}$  
or $\{1,3,4\}$ 
is more useful; it depends on the actual application. 
Second, considering only the size-based $\overline{\mathrm{eBH}}$ procedures does not yield a complete class (Theorem~\ref{thm:point-weighted-ebh-complete-class}). 
Therefore, the upgraded definition is both natural and technically necessary.

Note that 
$\overline{\mathrm{ebh}}_\alpha^{\bm \lambda}$ and $\overline{\mathrm{ebh}}_\alpha^{\mathsf{m}}$
are not uniquely specified by $K$, $\bm \lambda$ and $\alpha$,
because the chosen maximal set 
from $\cebh$ has flexibility.
This contrasts with $\cebh_\alpha^{\bm \lambda}$, which is a mapping uniquely specified by the parameters.


\cite{XSF25} noted that the point mean $\overline{\mathrm{eBH}}$ procedure $\overline{\mathrm{ebh}}_\alpha^{\mathsf{m}}$ dominates the eBH procedure. The next result identifies when a weighted-mean $\overline{\mathrm{eBH}}$ procedure strongly dominates the self-consistent procedure. 
For $\bm\lambda^A=(\lambda^A_0,(\lambda_i^A)_{i\in A})\in \Delta_{|A|+1}$,
we write $ \underline{\lambda}^A=\min_{i\in A} \lambda_i^A$, which corresponds to the smallest weight assigned to any e-value.

\begin{proposition}\label{prop:esc-dominated-by-weighted-ebh}
Let $\alpha\in(0,1)$ and
consider the procedure $\overline{\mathrm{eBH}}_\alpha^{\bm\lambda}$ in \eqref{eq:closed-def}.

\begin{enumerate}[(i)]
\item $\mathrm{eSC}_\alpha\preceq_{\mathrm s}\overline{\mathrm{eBH}}_\alpha^{\bm\lambda}$ if and only if 
\begin{align}
\label{eq:dom-eBH}
\alpha\lambda_0^A+K\underline{\lambda}^A\ge1 ~\mbox{ for all~}  A\subseteq\cK.
\end{align}

\item Suppose that at least one of the following holds: the procedure is constant-free; $K=2$; or $K\ge3$ and $\alpha<K/(K+1)$. Then $\mathrm{eSC}_\alpha\preceq_{\mathrm w}\overline{\mathrm{eBH}}_\alpha^{\bm\lambda}$ if and only if  \eqref{eq:dom-eBH} holds. 

\item Suppose that at least one of the following holds: the procedure is constant-free; or $\alpha<1/2$. There exists  an $\overline{\mathrm{ebh}}_\alpha^{\bm\lambda}\supseteq \mathrm{eBH}_\alpha$ if and only if   \eqref{eq:dom-eBH} holds.
\end{enumerate}
\end{proposition}

The proof of part~(i) is short and provides intuition for condition~\eqref{eq:dom-eBH}, so we include it below. The proofs of parts~(ii) and~(iii) are deferred to Appendix~\ref{app:A2}.

\begin{proof}[Proof of part~(i)]
\underline{The ``only if'' direction}: Fix an $A\subseteq\cK$ and $i\in A$, take $\mathbf e=(e_1,\dots,e_K)$ with $e_i=K/\alpha$ and $e_j=0$ for $j\ne i$. Then $\{i\}\in\mathrm{eSC}_\alpha(\mathbf e)$, and hence $\{i\}\in\overline{\mathrm{eBH}}_\alpha^{\bm{\lambda}}(\mathbf e)$. Thus we obtain $\lambda_0^A+K\lambda_i^A/\alpha\ge1/\alpha$. Since $i\in A$ and $A\subseteq\cK$ can be arbitrary, \eqref{eq:dom-eBH} follows.

\underline{The ``if'' direction}: For any $\mathbf e\in[0,\infty]^K$ and $R\in\mathrm{eSC}_\alpha(\mathbf e)$, we show $R\in\overline{\mathrm{eBH}}_\alpha^{\bm{\lambda}}(\mathbf e)$. By definition, $e_i\ge K/(\alpha |R|)$ for all $i\in R$. Fix $A\subseteq\cK$. If $A\cap R=\varnothing$, there is nothing to prove. If $e_i=\infty$ for some $i\in A$, then $\mathbb M_{\bm{\lambda}^A}(\mathbf e^{A})=\infty$. It remains to consider the case where $A\cap R\neq \varnothing$ and $e_i<\infty$ for all $i\in A$. Summing $\alpha\lambda_0^A+K\lambda_i^A\ge1$ over $i\in A\cap R$ gives $$K\sum_{i\in A\cap R}\lambda_i^A\ge |A\cap R|(1-\alpha\lambda_0^A).$$ Then
\[
\begin{split}
\mathbb M_{\bm{\lambda}^A}(\mathbf e^{A})
&\ge \lambda_0^A+\frac{K}{\alpha |R|}\sum_{i\in A\cap R}\lambda_i^A
\ge \lambda_0^A+\frac{|A\cap R|(1-\alpha\lambda_0^A)}{\alpha |R|} \\
&= \frac{|A\cap R|}{\alpha |R|}+\lambda_0^A\left(1-\frac{|A\cap R|}{|R|}\right)
\ge \frac{|A\cap R|}{\alpha |R|}.
\end{split}
\]
Thus $R\in\overline{\mathrm{eBH}}_\alpha^{\bm{\lambda}}(\mathbf e)$, so $\mathrm{eSC}_\alpha\preceq_{\mathrm s}\overline{\mathrm{eBH}}_\alpha^{\bm{\lambda}}$.
\end{proof}

Proposition~\ref{prop:esc-dominated-by-weighted-ebh} provides a useful calibration rule: whenever \eqref{eq:dom-eBH} holds, there exists an
$\overline{\mathrm{ebh}}_\alpha^{\bm\lambda}$ that dominates the base eBH procedure (but this is not true for  every $\overline{\mathrm{ebh}}_\alpha^{\bm\lambda}$). When $\alpha\ge1/2$, the condition \eqref{eq:dom-eBH} is not necessary in general.  Both phenomena are illustrated by Example~\ref{ex:point-version-need-not-dominate-ebh} in Appendix~\ref{app:counterexamples}.

Equivalently,  \eqref{eq:dom-eBH} can be rewritten as $\underline{\lambda}^A\ge (1-\alpha\lambda_0^A) / K$ for all $A$. Since $\alpha / K$ is typically very small, this lower bound is close to $1/K$. In particular, the mean $\overline{\mathrm{eBH}}$ procedure satisfies the condition, hence $\mathrm{eSC}_\alpha\preceq_{\mathrm s}\overline{\mathrm{eBH}}_\alpha^{\mathsf{m}}$.
Intuitively, if, for some $A\subseteq\cK$ and $i\in A$, the weight $\lambda_i^A$ is too small, then the weighted-mean $\overline{\mathrm{eBH}}$ procedure may fail to reject $\{i\}$ even for the vector $\mathbf e$ with a very large $e_i$ and all other coordinates equal to $0$. 


\section{Admissibility of closed eBH procedures}
\label{sec:adm-closed}

This section develops the admissibility results for weighted-mean $\overline{\mathrm{eBH}}$ procedures. This study helps
to give explicit ways for designing a natural admissible subclass, as we discuss in Section~\ref{sec:prior-score-weights}.

\subsection{General admissibility results}

This section focuses on the admissibility of weighted-mean $\overline{\mathrm{eBH}}$ procedures and point weighted-mean $\overline{\mathrm{eBH}}$ procedures. 


\begin{theorem}
\label{thm:constant-free-weighted-ebh-admissible}
For $\alpha\in (0,1)$, any constant-free weighted-mean $\overline{\mathrm{eBH}}$ procedure at level $\alpha$
is admissible. 
\end{theorem}
 

\begin{proof} 
Fix such a procedure $\overline{\mathrm{eBH}}_\alpha^{\bm{\lambda}}$. We show that if there exists $\mathcal G \in \mathrm{SP}_\alpha$ such that $\overline{\mathrm{eBH}}_\alpha^{\bm{\lambda}} \preceq_{\mathrm w} \mathcal G$, then $\mathcal G \preceq_{\mathrm s} \overline{\mathrm{eBH}}_\alpha^{\bm{\lambda}}$, which directly gives admissibility. Suppose otherwise for contradiction. Then we can find $\mathbf e=(e_1,\dots,e_K)$, $R\in\mathcal G(\mathbf e)\setminus\overline{\mathrm{eBH}}_\alpha^{\bm{\lambda}}(\mathbf e)$ and a nonempty $A\subseteq \cK$ such that
$$\mathbb M_{\bm{\lambda}^A}(\mathbf e^{A})<\frac{|A\cap R|}{\alpha |R|}.$$ 
Since $\mathbb M_{\bm{\lambda}^A}(\mathbf e^{A})$ is finite, $e_i<\infty$ follows for all $i\in A$.

Choose $p>0$ so small that $p e_i<1$ for every $i\in A$ and $p+q<1$, where
$q=\alpha(1-p \mathbb M_{\bm{\lambda}^A}(\mathbf e^{A}))$. Define $\mathbf y=(y_1,\dots,y_K)$ by
$$
y_i=
\begin{cases}
\dfrac{1-p e_i}{q}, & i\in A,\\[0.8ex]
\infty, & i\notin A.
\end{cases}
$$
We claim that $A\in\overline{\mathrm{eBH}}_\alpha^{\bm{\lambda}}(\mathbf y)$ and $\cK \in\overline{\mathrm{eBH}}_\alpha^{\bm{\lambda}}(\mathbf y)$. Fix $B\subseteq \cK$. It suffices to check $B\subseteq A$, since if $B\setminus A\ne\varnothing$, then $\mathbb M_{\bm{\lambda}^B}(\mathbf y^B)=\infty$ and there is nothing to prove.
We first verify $\mathbb M_{\bm{\lambda}^B}(\mathbf y^B) \ge |B\cap A| / (\alpha|A|)=|B| / (\alpha|A|)$.
\begin{itemize}
\item[(i)] If $B=A$, then $\mathbb M_{\bm{\lambda}^B}(\mathbf y^B)=1 / \alpha =|B| / (\alpha|A|) $.
\item[(ii)] If $B\subsetneq A$, then we know $\mathbb M_{\bm{\lambda}^B}(\mathbf y^B)\ge |B|/(\alpha |A|)$ for sufficiently small $p$ by the fact that $\mathbb M_{\bm{\lambda}^B}(\mathbf y^B)=(1-p \mathbb M_{\bm{\lambda}^B}(\mathbf e^B))/q$ converges to $1/\alpha$ as $p\downarrow 0$.
\end{itemize}
Thus $A\in\overline{\mathrm{eBH}}_\alpha^{\bm{\lambda}}(\mathbf y)$. The above arguments also give $\cK \in\overline{\mathrm{eBH}}_\alpha^{\bm{\lambda}}(\mathbf y)$. 

Since $\cK\in\overline{\mathrm{eBH}}_\alpha^{\bm{\lambda}}(\mathbf y)$, we have $\overline{\mathrm{eBH}}_\alpha^{\bm{\lambda}}(\mathbf y)^\downarrow=2^{\cK}$. Because $\overline{\mathrm{eBH}}_\alpha^{\bm{\lambda}}\preceq_{\mathrm w}\mathcal G$, the inclusion $\overline{\mathrm{eBH}}_\alpha^{\bm{\lambda}}(\mathbf y)\subseteq\mathcal G(\mathbf y)$ must hold, which implies $A\in\mathcal G(\mathbf y)$.

Now we show that $\mathcal G\notin \mathrm{SP}_\alpha$. Set $A$ as the true null and define a random vector $\mathbf X=(X_1,\dots,X_K)$ by
$$
\mathbf X=
\begin{cases}
\mathbf e, & \text{with probability }p,\\
\mathbf y, & \text{with probability }q,\\
\mathbf 0, & \text{with probability }1-p-q.
\end{cases}
$$

\begin{itemize}
\item[(i)] For every $i\in A$, $\mathbb E[X_i]=p e_i+q y_i=1$.

\item[(ii)] We have 
$$
\begin{aligned}
\FDR_\mathcal G(\mathbf X)
&=
\E\left[\max_{R\in \mathcal G(\mathbf X)}\FDP_A(R)\right]
\ge 
p\frac{|A\cap R|}{|R|}
+
\alpha(1-p \mathbb M_{\bm{\lambda}^A}(\mathbf e^{A}))\\
&=
\alpha+p\left(\frac{|A\cap R|}{|R|}-\alpha \mathbb M_{\bm{\lambda}^A}(\mathbf e^{A})\right)
>
\alpha.
\end{aligned}
$$
\end{itemize}
Therefore $\mathcal G \preceq_{\mathrm s} \overline{\mathrm{eBH}}_\alpha^{\bm{\lambda}}$. 
\end{proof}

Theorem~\ref{thm:constant-free-weighted-ebh-admissible} establishes the first admissibility results in the literature on FDR-controlling procedures. It says that any weighted-mean $\overline{\mathrm{eBH}}$ procedure is admissible as long as the procedure is constant-free. Asymmetry in the weights can therefore be used to incorporate side information without sacrificing the decision-theoretic property of being unimprovable.

\begin{theorem}\label{thm:positive-point-weighted-ebh-admissible}
For $\alpha\in (0,1/2)$, any point weighted-mean $\overline{\mathrm{eBH}}$ procedure at level $\alpha$ that is constant-free and strictly positive is admissible.
\end{theorem}

\begin{proof}
We first describe the outline of the proof. 
Suppose, for contradiction, that   a point procedure $\mathfrak D$ strictly dominates the point $\overline{\mathrm{eBH}}$ procedure. Then it rejects a superset $R$ at some input $\mathbf e$, so maximality gives a witness set $A$ where the $\overline{\mathrm{eBH}}$ constraint fails. We use $A$ as the true null set, and split the following proof into two distinct cases. The first case, $|A|>\alpha K$, is straightforward: we construct a ``bad" point $\mathbf y$ with $\FDP_A(\mathfrak D(\mathbf y))=|A| / K >\alpha$; allocating suitable probabilities among $\mathbf e$, $\mathbf y$, and $\mathbf 0$ yields valid e-values, but FDR larger than $\alpha$. The second case, $|A| \le \alpha K$, is more delicate: the single bad point $\mathbf y$ with $\FDP_A(\mathfrak D(\mathbf y)) =|A| / K$ is insufficient because $|A| / K \le \alpha$. Here, the core idea is to construct other $|A|$  bad points $\mathbf u^{(j,\varepsilon)}$, $j\in A$, satisfying $\FDP_A(\mathfrak D(\mathbf u^{(j,\varepsilon)}))=1/2>\alpha$ for every $j\in A$; allocating suitable probabilities among $\mathbf e$, $\mathbf y$, $\mathbf 0$, and these $|A|$ bad points yields valid e-values, but FDR larger than $\alpha$.
 
Now we present the formal proof. 
Let $\overline{\mathrm{ebh}}_\alpha^{\bm{\lambda}}$ be any constant-free and strictly positive point weighted-mean $\overline{\mathrm{eBH}}$ procedure and $\overline{\mathrm{eBH}}_\alpha^{\bm{\lambda}}$ be the corresponding simultaneous procedure. Suppose, for contradiction, that there exists a point procedure $\mathfrak D\in \mathrm{PP}_\alpha$ that strictly dominates $\overline{\mathrm{ebh}}_\alpha^{\bm \lambda}$. Then we can find $\mathbf e$ such that $\overline{\mathrm{ebh}}_\alpha^{\bm \lambda}(\mathbf e)\subsetneq \mathfrak D(\mathbf e)$. Write $R=\mathfrak D(\mathbf e)$.
Since $\overline{\mathrm{ebh}}_\alpha^{\bm \lambda}(\mathbf e)$ is a maximal set of $\overline{\mathrm{eBH}}_\alpha^{\bm{\lambda}}(\mathbf e)$, we have $R\notin\overline{\mathrm{eBH}}_\alpha^{\bm{\lambda}}(\mathbf e)$. Hence there exists a nonempty
$A\subseteq\cK$ such that
\begin{equation}\label{eq:ad-no-comp-bad}
\mathbb M_{\bm{\lambda}^A}(\mathbf e^{A})<\frac{|A\cap R|}{\alpha |R|}.
\end{equation}
It is easy to see that $e_i<\infty$ for all $i\in A$. Set $n=|A|$. We split into two cases.

\textbf{Case 1}: $n>\alpha K$. Choose $p>0$ so small that $1-p e_i\ge0$ for all $i\in A$. Denote $\mathbf x=(x_1,\dots,x_K)$, where $x_i=(1-p e_i)/(1-p \mathbb M_{\bm{\lambda}^A}(\mathbf e^{A}))$ for $i\in A$, and $x_i=0$ otherwise.
Then $\mathbb M_{\bm{\lambda}^A}(\mathbf x^A)=1$, and $x_i\to 1$ for $i\in A$ as $p\downarrow 0$. For all $T\subseteq A$, we may pick a small $p$ such that $\mathbb M_{\bm{\lambda}^T}(\mathbf x^T)\ge |T| / n$. Define $\mathbf y=(y_1,\dots,y_K)$ by
$$
y_i=
\begin{cases}
\dfrac{n}{\alpha K}x_i, & i\in A,\\[1ex]
\infty, & i\in A^c.
\end{cases}
$$
We claim $\overline{\mathrm{ebh}}_\alpha^{\bm \lambda}(\mathbf y)=\mathfrak D(\mathbf y)=\cK$. It suffices to show that $\cK \in\overline{\mathrm{eBH}}_\alpha^{\bm{\lambda}}(\mathbf y)$.

\begin{itemize}
\item[(i)] If $T\subseteq A$, then $\mathbb M_{\bm{\lambda}^T}(\mathbf y^T)=n \mathbb M_{\bm{\lambda}^T}(\mathbf x^T)/(\alpha K)\ge |T|/(\alpha K)$.

\item[(ii)] If $T\setminus A\neq\varnothing$, then $\mathbb M_{\bm{\lambda}^T}(\mathbf y^T)=\infty$ and there is nothing to prove.
\end{itemize}

Now we show that $\mathfrak D\notin \mathrm{PP}_\alpha$. Set $A$ as the true null. Let
$$q=\frac{\alpha K}{n}\bigl(1-p \mathbb M_{\bm{\lambda}^A}(\mathbf e^{A})\bigr).$$
Since $n>\alpha K$, by taking $p$ sufficiently small we have $p+q<1$. Define the random vector $\mathbf X=(X_1,\dots,X_K)$ by
$$
\mathbf X=
\begin{cases}
\mathbf e, & \text{with probability }p,\\
\mathbf y, & \text{with probability }q,\\
\mathbf 0, & \text{with probability }1-p-q.
\end{cases}
$$
\begin{itemize}
\item[(i)] For each $i\in A$, we have $\mathbb E[X_i]=p e_i+q y_i=1$. 

\item[(ii)] By \eqref{eq:ad-no-comp-bad}, we have 
$$
\begin{aligned}
\mathrm{FDR}_{\mathfrak D}(\mathbf X)
&\ge
p\frac{|A\cap R|}{|R|}+q\frac{n}{K}
=
p\frac{|A\cap R|}{|R|}+\alpha\bigl(1-p \mathbb M_{\bm{\lambda}^A}(\mathbf e^{A})\bigr)\\
&=
\alpha+p\left(\frac{|A\cap R|}{|R|}-\alpha \mathbb M_{\bm{\lambda}^A}(\mathbf e^{A})\right)
>
\alpha.
\end{aligned}
$$
\end{itemize}

\textbf{Case 2}: $n\le \alpha K$. Since $\alpha<1/2$, we have $n<K/2$.
For each $j\in A$, choose some $j^*\in A^c$, and define
$W_j=\cK \setminus\{j^*\}$, $S_j=\{j,j^*\}$ and $C_j=\cK \setminus S_j$. Denote
$m_0=\min_{j\in A}\min\{\lambda_j^T:\ T\subseteq W_j,\ j\in T\}>0$.
For each $j\in A$, define $\mathbf u^{(j,\varepsilon)}=\bigl(u_1^{(j,\varepsilon)},\dots,u_K^{(j,\varepsilon)}\bigr)$ by
$$
u_i^{(j,\varepsilon)}
=
\begin{cases}
\dfrac{1}{2\alpha}+\dfrac{1-m_0}{\alpha m_0}\varepsilon, & i=j,\\[1ex]
\infty, & i=j^*,\\[1ex]
\dfrac{1}{2\alpha}-\dfrac1\alpha\varepsilon, & i\in C_j.
\end{cases}
$$
We claim that, if $\varepsilon>0$ is sufficiently small, then
\begin{equation}\label{eq:ad-no-comp-singleton}
\overline{\mathrm{ebh}}_\alpha^{\bm \lambda}(\mathbf u^{(j,\varepsilon)})=\mathfrak D(\mathbf u^{(j,\varepsilon)})=S_j.
\end{equation}

\begin{itemize}
\item[(i)] We show that $S_j\in\overline{\mathrm{eBH}}_\alpha^{\bm{\lambda}}(\mathbf u^{(j,\varepsilon)})$. Let $T\subseteq \cK$. If $T\cap S_j=\varnothing$, there is nothing to prove. If $j^*\in T$, then $\mathbb M_{\bm{\lambda}^T}((u_i^{(j,\varepsilon)})_{i\in T})=\infty$. If $j^*\notin T$ and $j\in T$, then
\[
\begin{split}
    \mathbb M_{\bm{\lambda}^T}((u_i^{(j,\varepsilon)})_{i\in T})
    &=\lambda_j^T\left(\frac{1}{2\alpha}+\frac{1-m_0}{\alpha m_0}\varepsilon\right)+(1-\lambda_j^T)\left(\frac{1}{2\alpha}-\frac1\alpha\varepsilon\right) \\
    &\ge\frac{1}{2\alpha}=\frac{|T\cap S_j|}{\alpha |S_j|}.
\end{split}
\]

\item[(ii)] We show no set of size $1$ in $C_j$ is in $\overline{\mathrm{eBH}}_\alpha^{\bm{\lambda}}(\mathbf u^{(j,\varepsilon)})$. Let $c\in C_j$. Taking $T=\{c\}$ gives $\mathbb M_{\bm{\lambda}^T}((u_i^{(j,\varepsilon)})_{i\in T}) < 1 / \alpha = |T\cap \{c\}| / \alpha$.

\item[(iii)] We show no set of size $2$ other than $S_j$ is in $\overline{\mathrm{eBH}}_\alpha^{\bm{\lambda}}(\mathbf u^{(j,\varepsilon)})$. Let $S'\subseteq \cK$ satisfy $|S'|=2$ and $S'\ne S_j$. Then $S'$ contains some $c\in C_j$. Taking $T=\{c\}$ gives $\mathbb M_{\bm{\lambda}^T}((u_i^{(j,\varepsilon)})_{i\in T}) <1 / (2\alpha)=|T\cap S'| / (\alpha |S'|)$.

\item[(iv)] We show $\overline{\mathrm{ebh}}_\alpha^{\bm \lambda}(\mathbf u^{(j,\varepsilon)})=S_j$, that is, $S_j$ is the unique maximal set in $\overline{\mathrm{eBH}}_\alpha^{\bm{\lambda}}(\mathbf u^{(j,\varepsilon)})$. Consider a set $L\subseteq \cK$ with $|L|>2$. Note that
$$\mathbb M_{\bm{\lambda}^{W_j}}((u_i^{(j,\varepsilon)})_{i\in W_j})=\lambda_j^{W_j}\left(\frac{1}{2\alpha}+\frac{1-m_0}{\alpha m_0}\varepsilon\right)+(1-\lambda_j^{W_j})\left(\frac{1}{2\alpha}-\frac1\alpha \varepsilon\right).$$
We may choose $\varepsilon>0$ so small that $\mathbb M_{\bm{\lambda}^{W_j}}((u_i^{(j,\varepsilon)})_{i\in W_j})<2/(3\alpha)$ for all $j\in A$. However, since $|L|>2$, we have $|W_j\cap L|/|L|\ge2/3$. Hence
$$\mathbb M_{\bm{\lambda}^{W_j}}((u_i^{(j,\varepsilon)})_{i\in W_j})<\frac{|W_j\cap L|}{\alpha |L|}.$$
Therefore no set in $\overline{\mathrm{eBH}}_\alpha^{\bm{\lambda}}(\mathbf u^{(j,\varepsilon)})$ has size strictly larger than $2$. Combining with (ii) and (iii), we know that $S_j$ is the unique maximal set in $\overline{\mathrm{eBH}}_\alpha^{\bm{\lambda}}(\mathbf u^{(j,\varepsilon)})$.

\item[(v)] We show $\mathfrak D(\mathbf u^{(j,\varepsilon)})=S_j$. Since $\mathfrak D$ dominates $\overline{\mathrm{ebh}}_\alpha^{\bm \lambda}$, we have $S_j\subseteq\mathfrak D(\mathbf u^{(j,\varepsilon)})$. If the inclusion is strict, then $|\mathfrak D(\mathbf u^{(j,\varepsilon)})|>2$. By the same argument in (iv), we get
$$\mathbb M_{\bm{\lambda}^{W_j}}((u_i^{(j,\varepsilon)})_{i\in W_j})<\frac{|W_j\cap \mathfrak D(\mathbf u^{(j,\varepsilon)})|}{\alpha |\mathfrak D(\mathbf u^{(j,\varepsilon)})|}.$$
Since $K\ge 2$ and $\alpha<1/2$, we have $|W_j|=K-1>\alpha K$. Applying Case~1 to the bad witness $(\mathbf u^{(j,\varepsilon)},W_j)$ would imply $\mathfrak D\notin \mathrm{PP}_\alpha$, a contradiction. Hence $\mathfrak D(\mathbf u^{(j,\varepsilon)})=S_j$.
\end{itemize}

Set $A$ as the true null. For all $j\in A$, \eqref{eq:ad-no-comp-singleton} gives
\begin{equation}\label{eq:ad-no-comp-singleton-fdp}
\mathrm{FDP}_A\!\bigl(\mathfrak D(\mathbf u^{(j,\varepsilon)})\bigr)=\frac12.
\end{equation}
Since $\alpha<1/2$, choose a small $s>0$ so that
\begin{equation}\label{eq:ad-no-comp-s}
2\alpha(1-s)+\frac{\alpha K}{n}s<1.
\end{equation}
Choose $p>0$ so small that $s-p e_i\ge0$ for all $i\in A$. Denote $\mathbf z=(z_1,\dots,z_K)$, where $z_i=(s-p e_i)/(s-p \mathbb M_{\bm{\lambda}^A}(\mathbf e^{A}))$ for $i\in A$, and $z_i=0$ otherwise. Then $\mathbb M_{\bm{\lambda}^A}(\mathbf z^A)=1$. By taking $p$ sufficiently small, we may ensure $\mathbb M_{\bm{\lambda}^T}(\mathbf z^T)\ge |T| / n$ for all $T\subseteq A$. Define $\mathbf y=(y_1,\dots,y_K)$ by
$$
y_i=
\begin{cases}
\dfrac{n}{\alpha K}z_i, & i\in A,\\[1ex]
\infty, & i\in A^c.
\end{cases}
$$
The same verification used in Case~1 gives $\mathfrak D(\mathbf y)=\overline{\mathrm{ebh}}_\alpha^{\bm \lambda}(\mathbf y)=\cK$, and hence
\begin{equation}
\label{eq:ad-all-fdp}
\mathrm{FDP}_A(\mathfrak D(\mathbf y))=\frac nK.
\end{equation}

Now we show that $\mathfrak D\notin \mathrm{PP}_\alpha$. Set
$$q=\frac{2\alpha(1-s)}{n+2\bigl(\frac{1}{m_0}-n\bigr)\varepsilon}\quad \text{and} \quad q_y=\frac{\alpha K}{n}\bigl(s-p \mathbb M_{\bm{\lambda}^A}(\mathbf e^{A})\bigr).$$
As $\varepsilon\downarrow0$, we have $nq\to2\alpha(1-s)$. Also, as $p\downarrow0$, we have $q_y\to \alpha Ks/n$. Therefore, by \eqref{eq:ad-no-comp-s}, we have $p+nq+q_y<1$. Write $A=\{h_1,\dots,h_n\}$, and define the random vector $\mathbf X=(X_1,\dots,X_K)$ by
$$
\mathbf X=
\begin{cases}
\mathbf e, & \text{with probability }p,\\
\mathbf u^{(h_1,\varepsilon)}, & \text{with probability }q,\\
\cdots, & \cdots\\
\mathbf u^{(h_n,\varepsilon)}, & \text{with probability }q,\\
\mathbf y, & \text{with probability }q_y,\\
\mathbf 0, & \text{with probability }1-p-nq-q_y.
\end{cases}
$$

\begin{itemize}
\item[(i)] For each $r\in A$,
$$\sum_{j\in A}q\,u_r^{(j,\varepsilon)}=q\left(\frac{1}{2\alpha}+\frac{1-m_0}{\alpha m_0}\varepsilon\right)+q(n-1)\left(\frac{1}{2\alpha}-\frac1\alpha\varepsilon\right)=1-s.$$
Therefore $\mathbb E[X_r]=p e_r+\sum_{j\in A}q\,u_r^{(j,\varepsilon)}+q_y y_r=p e_r+(1-s)+(s-p e_r)=1$.

\item[(ii)] By \eqref{eq:ad-no-comp-bad}, \eqref{eq:ad-no-comp-singleton-fdp} and \eqref{eq:ad-all-fdp}, we have
$$\mathrm{FDR}_{\mathfrak D}(\mathbf X)\ge p\frac{|A\cap R|}{|R|}+\frac{nq}{2}+q_y\frac nK=p\frac{|A\cap R|}{|R|}+\frac{nq}{2}+
\alpha\bigl(s-p \mathbb M_{\bm{\lambda}^A}(\mathbf e^{A})\bigr).$$
Since $nq/2\to\alpha(1-s)$ as $\varepsilon\downarrow0$, choose $\varepsilon>0$ sufficiently small so that
$$\frac{nq}{2}>\alpha(1-s)-\frac p2 \left(\frac{|A\cap R|}{|R|}-\alpha \mathbb M_{\bm{\lambda}^A}(\mathbf e^{A})\right).$$
Then we have
$$\mathrm{FDR}_{\mathfrak D}(\mathbf X)>\alpha+\frac p2 \left(\frac{|A\cap R|}{|R|}-\alpha \mathbb M_{\bm{\lambda}^A}(\mathbf e^{A})\right)>\alpha.$$
\end{itemize}

Thus both cases are impossible, and $\overline{\mathrm{ebh}}_\alpha^{\bm \lambda}$ is admissible.
\end{proof}

The proof of Theorem~\ref{thm:positive-point-weighted-ebh-admissible} is more delicate than that of 
Theorem~\ref{thm:constant-free-weighted-ebh-admissible}
because point procedures must commit to one maximal set. The strict-positivity condition cannot be removed in general. We provide a counterexample in Example~\ref{ex:zero-weight-point-not-admissible} in Appendix~\ref{app:counterexamples}.

For the point mean $\overline{\mathrm{eBH}}$ procedure, the threshold for admissibility is exactly one half.
Note that for  given $\alpha\in (0,1)$ and $K\ge 3$, there 
is a unique mean $\overline{\mathrm{eBH}}$  procedure, but there 
are many  point mean $\overline{\mathrm{eBH}}$  procedures, because of the flexibility in the selection of the maximal sets. 
 
\begin{proposition}
\label{prop:point-mean-admissibility-threshold}
Any point mean $\overline{\mathrm{eBH}}$ procedure  at level $\alpha$ is admissible if and only if $\alpha<1/2$.
\end{proposition}

\begin{proofidea}
For $\alpha\ge1/2$, the point mean procedure can be strictly improved. For $K=2$, the improvement simply rejects both hypotheses whenever $e_1+e_2\ge2/\alpha$. For $K\ge3$, the same idea is implemented by promoting the output to $\mathcal K$ when every subset of size at least two has large enough realized e-values. The promotion does not violate FDR because it is controlled by the mean e-value when there are at least two true nulls, and by the bound $1/2\le\alpha$ when there is only one true null.
\end{proofidea}

This sharp threshold for admissibility is informative: it shows that one half in Theorem~\ref{thm:positive-point-weighted-ebh-admissible} is essential, not just a simplifying assumption.
For general weighted-mean $\cebh$ procedures, the above results do not identify their admissibility for $\alpha\ge 1/2$.

\subsection{Choosing asymmetric weights with side information}
\label{sec:prior-score-weights}

The general weighted-mean $\overline{\mathrm{eBH}}$ procedure in Definition~\ref{def:closed-eBH-procedure} has $K2^{K-1}$ degrees of freedom, 
so some principle for choosing the parameter is needed in practice. 
If side information is available, we can build a practically interpretable subclass. Suppose that  the side information gives each hypothesis $H_i$ a score $s_i > 0$, where a larger score indicates a more promising alternative.
For a chosen parameter $\gamma\ge0$ and every  $A\subseteq\cK$ and every $i\in A$, we propose  two constant-free weight systems, called the \emph{normalized} and \emph{adjusted-normalized} rules, respectively:
\begin{equation}\label{eq:prior-score-weights}
\lambda_{i,{\rm nor}}^A=\frac{s_i^\gamma}{\sum_{j\in A}s_j^\gamma},
\quad\text{and}\quad
\lambda_{i,{\rm adj}}^A=\frac1K+\frac{K-|A|}{K}\frac{s_i^\gamma}{\sum_{j\in A}s_j^\gamma}.
\end{equation}

For $\ell\in \{\rm nor,adj\}$, write $\bm\lambda_\ell=\{\bm\lambda_\ell^A\}_{A\subseteq\cK}$, where $\bm\lambda_\ell^A=(0,(\lambda_{i,\ell}^A)_{i\in A})$, and let $\overline{\mathrm{eBH}}_\alpha^{\bm\lambda_\ell}$ denote the corresponding weighted-mean $\overline{\mathrm{eBH}}$ procedure. A smaller $\gamma$ makes the procedures less sensitive to the scores, whereas a larger $\gamma$ concentrates more weight on hypotheses with larger scores. When $\gamma=0$, both procedures reduce to the mean $\overline{\mathrm{eBH}}$ procedure. The normalized rule allocates all weight according to the scores. The resulting weights are coherent across intersection sets, permutation equivariant, invariant to multiplying all scores by the same positive constant, and order preserving. Moreover, $\lambda_{i,{\rm nor}}^A/\lambda_{j,{\rm nor}}^A=(s_i/s_j)^\gamma$, so the relative weight is monotone in their score ratio and unaffected by the other hypotheses. The adjusted-normalized rule first assigns every $i\in A$ the lower bound $1 / K$, which guarantees \eqref{eq:dom-eBH}, and then distributes all remaining mass proportionally to the scores. It is therefore a canonical rule within the baseline-plus-residual construction that retains the eBH domination.  
Both $\overline{\mathrm{eBH}}_\alpha^{\bm\lambda_{\rm nor}}$ and $\overline{\mathrm{eBH}}_\alpha^{\bm\lambda_{\rm adj}}$ are admissible by Theorem~\ref{thm:constant-free-weighted-ebh-admissible}. Similarly,  every $\overline{\mathrm{ebh}}_\alpha^{\bm\lambda_{\rm nor}}$ or $\overline{\mathrm{ebh}}_\alpha^{\bm\lambda_{\rm adj}}$ is admissible under the simple conditions in Theorem~\ref{thm:positive-point-weighted-ebh-admissible}.
 
In applications, the scores and $\gamma$ should be specified in advance or learned from data independent of the e-values used for testing. Without further information, a default value is $\gamma=1$.
The normalized rule has  greater sensitivity to the prior scores,  whereas the adjusted-normalized rule guarantees a domination over eBH. In Section~\ref{sec:sim-asymmetric-weights}, we illustrate potential power gains from the resulting asymmetric weights learned from side information. 

\begin{remark}
In some applications, not all hypotheses have side information. If $s_i$ is only available for some hypotheses $i$ but not all of them, we can set $s_j=\min_{i:s_i>0}s_i$ for the rest hypotheses $j$. 
\end{remark}

\section{Complete classes for admissible procedures}
\label{sec:complete-class}

This section establishes complete class results for weighted-mean $\overline{\mathrm{eBH}}$ procedures. We first study simultaneous procedures, then turn to point procedures. Moreover, we show that the closed Benjamini--Yekutieli (BY) procedure for p-values is not admissible.


\subsection{Simultaneous procedures}\label{sec:simultaneous-fdr}

We begin by showing that any FDR-controlling simultaneous procedure can be strongly dominated by an admissible weighted-mean $\overline{\mathrm{eBH}}$ procedure.


\begin{theorem}
\label{thm:admissible-weighted-ebh-dominator}
For $\mathcal D\in\mathrm{SP}_\alpha$, there exists an admissible weighted-mean $\overline{\mathrm{eBH}}$ procedure in $ \mathrm{SP}_\alpha$ that strongly dominates $\mathcal D$.
In particular, all admissible elements of $ \mathrm{SP}_\alpha$ are weighted-mean $\overline{\mathrm{eBH}}$ procedures. 
\end{theorem}

\begin{proof}
We proceed in three steps.

Step 1: We show $\mathcal D$ is strongly dominated by a weighted-mean $\overline{\mathrm{eBH}}$ procedure. For this purpose, we use the following fact. Let $F:[0,\infty)^m\to[0,\infty)$ be bounded. Assume that $\mathbb E[F(\mathbf X)]\le 1$ for every finitely supported $\mathbf X=(X_1,\dots,X_m)$ satisfying $\mathbb E[X_i]\le1$ for all $i=1,\dots,m$. Then there exists $(\lambda_0,\lambda_1,\dots,\lambda_m)\in\Delta_{m+1}$ such that
$F(\mathbf x)\le\lambda_0+\sum_{i=1}^m\lambda_i x_i$. Although its assumptions differ slightly from those in Theorem~1 of \cite{C26}, the same proof idea applies after a simple adjustment, and we omit the details.

For each nonempty $A\subseteq\cK$ and each $\mathbf e\in[0,\infty]^K$, define
$$f_A(\mathbf e)=\sup_{R\in\mathcal D(\mathbf e)}\frac{|A\cap R|}{\alpha(|R|\vee1)}.$$
For finite $(x_i)_{i\in A}$, consider $F_A((x_i)_{i\in A})=\sup_{\mathbf z}f_A((x_i)_{i\in A},\mathbf z)$. We verify that the useful fact applies to $F_A$. Clearly, the function $F_A$ is bounded. Let $\mathbf X=(X_i)_{i\in A}$ be finitely supported and satisfy $\mathbb E[X_i]\le1$ for all $i\in A$. Since $f_A$ takes values in a finite set, for every point $\mathbf x$ in the support of $\mathbf X$, we may choose $\mathbf z(\mathbf x)$ such that $F_A(\mathbf x)=f_A(\mathbf x,\mathbf z(\mathbf x))$.
Define $\mathbf Z(\mathbf X)=\mathbf z(\mathbf X)$. Since $\mathcal D\in\mathrm{SP}_\alpha$, we have $\mathbb E\left[f_A(\mathbf X,\mathbf Z(\mathbf X))\right]\le 1$, so $\mathbb E[F_A(\mathbf X)]\le 1$. Thus we can find $\bm\lambda^A=(\lambda_0^A,(\lambda_i^A)_{i\in A}) \in\Delta_{|A|+1}$ such that $F_A((x_i)_{i\in A})\le\lambda_0^A+\sum_{i\in A}\lambda_i^A x_i$.
Take $E_\varnothing^{\bm\lambda}=1$, and for every nonempty $A\subseteq\cK$, define $E_A^{\bm\lambda}(\mathbf e)=\mathbb M_{\bm\lambda^A}(\mathbf e^{A})$. Let $\overline{\mathrm{eBH}}_\alpha^{\bm\lambda}$ be the corresponding weighted-mean $\overline{\mathrm{eBH}}$ procedure. If $e_i<\infty$ for all $i\in A$, then
$$f_A(\mathbf e)\le F_A(\mathbf e^{A})\le\lambda_0^A+\sum_{i\in A}\lambda_i^A e_i=E_A^{\bm\lambda}(\mathbf e).$$
If $e_i=\infty$ for some $i\in A$, then $E_A^{\bm\lambda}(\mathbf e)=\infty$, and there is nothing to prove. Consequently, for every $R\in\mathcal D(\mathbf e)$ and every nonempty $A\subseteq\cK$,
$$E_A^{\bm\lambda}(\mathbf e)\ge f_A(\mathbf e)\ge\frac{|A\cap R|}{\alpha(|R|\vee1)}.$$
Hence $\mathcal D\preceq_{\mathrm s}\overline{\mathrm{eBH}}_\alpha^{\bm\lambda}$.

Step 2: We show there exists an admissible weighted-mean $\overline{\mathrm{eBH}}$ procedure in $ \mathrm{SP}_\alpha$ that strongly dominates $\mathcal D$. Let $\Theta=\prod_{A\subseteq\cK}\Delta_{|A|+1}$. Each parameter $\bm \theta\in\Theta$ with $\bm \theta^A=(\theta_0^A,(\theta_i^A)_{i\in A})$ specifies a weighted-mean e-collection by $E_A^{\bm \theta}(\mathbf e)=\mathbb M_{\bm \theta^A}(\mathbf e^{A})$. Let $\overline{\mathrm{eBH}}_\alpha^{\bm \theta}$ denote the corresponding weighted-mean $\overline{\mathrm{eBH}}$ procedure. Define
$$\Theta_{\mathrm s}=\left\{\bm \theta\in\Theta:\mathcal D\preceq_{\mathrm s}\overline{\mathrm{eBH}}_\alpha^{\bm \theta}\right\}.$$
By the above claim, $\Theta_{\mathrm s}$ is nonempty.

For fixed $\mathbf e\in[0,\infty]^K$ and $R\subseteq\cK$, define
$$\Gamma_{\mathbf e,R}=\left\{\bm \theta\in\Theta:R\in\overline{\mathrm{eBH}}_\alpha^{\bm \theta}(\mathbf e)\right\}.$$
The set $\Gamma_{\mathbf e,R}$ is closed. Indeed, for each nonempty $A\subseteq\cK$, the constraint is automatic if $e_i=\infty$ for some $i\in A$. Otherwise, it is a finite class of closed linear inequalities
$$\theta_0^A+\sum_{i\in A}\theta_i^A e_i\ge\frac{|A\cap R|}{\alpha(|R|\vee1)}.$$
Since the set $\Gamma_{\mathbf e,R}$ is closed,
$\Theta_{\mathrm s}=\bigcap_{\mathbf e\in[0,\infty]^K}\bigcap_{R\in\mathcal D(\mathbf e)}\Gamma_{\mathbf e,R}$
is closed in $\Theta$. Together with the compactness of $\Theta$, we know $\Theta_{\mathrm s}$ is compact.

Consider the family $\mathfrak P_{\mathrm s}=\left\{\overline{\mathrm{eBH}}_\alpha^{\bm \theta}:\bm \theta\in\Theta_{\mathrm s}\right\}$ equipped with the order $\preceq_{\mathrm s}$. We show that every
nonempty chain $\left\{\overline{\mathrm{eBH}}_\alpha^{\bm \theta_t}:t\in I\right\}$ in $\mathfrak P_{\mathrm s}$ has an upper bound. For each $\mathbf e\in[0,\infty]^K$, set $U(\mathbf e)=\bigcup_{t\in I}\overline{\mathrm{eBH}}_\alpha^{\bm \theta_t}(\mathbf e)$. Since $2^{\cK}$ is finite, there exists $t(\mathbf e)\in I$ such that $U(\mathbf e)=\overline{\mathrm{eBH}}_\alpha^{\bm \theta_{t(\mathbf e)}}(\mathbf e)$. Define
$$F_{\mathbf e}=\left\{\bm \theta\in\Theta_{\mathrm s}:U(\mathbf e)\subseteq\overline{\mathrm{eBH}}_\alpha^{\bm \theta}(\mathbf e)\right\}.$$
By definition, $F_{\mathbf e}=\Theta_{\mathrm s}\cap\left(\bigcap_{R\in U(\mathbf e)}\Gamma_{\mathbf e,R}\right)$, so $F_{\mathbf e}$ is closed in $\Theta_{\mathrm s}$. We next claim that the family
$\{F_{\mathbf e}\}_{\mathbf e\in[0,\infty]^K}$ has the finite
intersection property. Indeed, for finitely many points $\mathbf e^{(1)},\dots,\mathbf e^{(m)}$, choose $t_r\in I$ such that
$U(\mathbf e^{(r)})=\overline{\mathrm{eBH}}_\alpha^{\bm \theta_{t_r}}(\mathbf e^{(r)})$.
Among the finitely many procedures $\overline{\mathrm{eBH}}_\alpha^{\bm \theta_{t_1}},\dots,\overline{\mathrm{eBH}}_\alpha^{\bm \theta_{t_m}}$, choose a largest one in the chain, denoted by
$\overline{\mathrm{eBH}}_\alpha^{\bm \theta_{t^\dagger}}$. Then $\bm \theta_{t^\dagger}\in\bigcap_{r=1}^m F_{\mathbf e^{(r)}}$.
By compactness of $\Theta_{\mathrm s}$, we have $\bigcap_{\mathbf e\in[0,\infty]^K} F_{\mathbf e}\ne\varnothing$, so we may take
$\bm \theta'\in\bigcap_{\mathbf e\in[0,\infty]^K}F_{\mathbf e}$. Then $\overline{\mathrm{eBH}}_\alpha^{\bm \theta_t}(\mathbf e)\subseteq U(\mathbf e)\subseteq\overline{\mathrm{eBH}}_\alpha^{\bm \theta'}(\mathbf e)$ follows for every $t\in I$ and every $\mathbf e\in[0,\infty]^K$. Thus $\overline{\mathrm{eBH}}_\alpha^{\bm \theta'}$ is an upper bound for the chain. By Zorn's lemma, $\mathfrak P_{\mathrm s}$ has a maximal element $\overline{\mathrm{eBH}}_\alpha^{\bm \theta^*}$. Since $\bm \theta^*\in\Theta_{\mathrm s}$, we have $\mathcal D \preceq_{\mathrm s} \overline{\mathrm{eBH}}_\alpha^{\bm \theta^*}$.

We claim that $\overline{\mathrm{eBH}}_\alpha^{\bm \theta^*}$ is admissible. Let $\mathcal G\in\mathrm{SP}_\alpha$ satisfy $\overline{\mathrm{eBH}}_\alpha^{\bm \theta^*}\preceq_{\mathrm s}\mathcal G$. By Proposition~\ref{prop:strong-weak-admissibility-equivalence}, it suffices to verify $\mathcal G=\overline{\mathrm{eBH}}_\alpha^{\bm \theta^*}$. Step 1 gives some $\bm \gamma\in\Theta$ such that $\mathcal G\preceq_{\mathrm s}\overline{\mathrm{eBH}}_\alpha^{\bm \gamma}$. Consequently, $\mathcal D\preceq_{\mathrm s}\overline{\mathrm{eBH}}_\alpha^{\bm \gamma}$ holds, and hence $\bm \gamma\in\Theta_{\mathrm s}$. Moreover, $\overline{\mathrm{eBH}}_\alpha^{\bm \theta^*}\preceq_{\mathrm s}\overline{\mathrm{eBH}}_\alpha^{\bm \gamma}$. By maximality of $\overline{\mathrm{eBH}}_\alpha^{\bm \theta^*}$ in $\mathfrak P_{\mathrm s}$, $\overline{\mathrm{eBH}}_\alpha^{\bm \gamma}=\overline{\mathrm{eBH}}_\alpha^{\bm \theta^*}$ necessarily holds.
Therefore $\overline{\mathrm{eBH}}_\alpha^{\bm \theta^*}\preceq_{\mathrm s}\mathcal G\preceq_{\mathrm s}\overline{\mathrm{eBH}}_\alpha^{\bm \theta^*}$ forces $\mathcal G=\overline{\mathrm{eBH}}_\alpha^{\bm \theta^*}$. Thus $\overline{\mathrm{eBH}}_\alpha^{\bm \theta^*}$ is admissible.

Step 3: We show that all admissible elements of $\mathrm{SP}_\alpha$ are weighted-mean $\overline{\mathrm{eBH}}$ procedures. 
If $\mathcal D$ is admissible, then $\mathcal D\preceq_{\mathrm s}\overline{\mathrm{eBH}}_\alpha^{\bm \theta^*}$ implies
$\mathcal D=\overline{\mathrm{eBH}}_\alpha^{\bm \theta^*}$. This concludes the proof.
\end{proof}

Theorem~\ref{thm:admissible-weighted-ebh-dominator} shows that the admissible weighted-mean $\overline{\mathrm{eBH}}$ procedures form a complete class: it significantly reduces the study of admissibility, which only needs to be carried out inside the weighted-mean $\overline{\mathrm{eBH}}$ class. In decision theory, this means that procedures outside this class are less attractive. They may be useful as starting points, but if they are not already of weighted-mean $\overline{\mathrm{eBH}}$ form, then there is an admissible procedure of that form that is at least as powerful under strong dominance.

The $\infty$ convention in Definition~\ref{def:closed-eBH-procedure} is essential. Replacing it by specifying $0\cdot\infty=0$ would invalidate the theorem. This point is illustrated by Example~\ref{ex:infty-convention} in Appendix~\ref{app:counterexamples}.

\begin{remark}
Proposition~51 of \citet{XSF25}
has a result similar to Theorem~\ref{thm:admissible-weighted-ebh-dominator}. 
In their result, they assume $\E[E_i]>1$ for any $i$ that represents an alternative hypothesis, which imposes a strong restriction. For instance, under any configuration containing at least one
alternative, the  constant vector $(1,\dots,1)$ does not satisfy their assumption and thus it is excluded from the input of the procedure.
Our technical contributions over that result also include the point that any simultaneous procedure admits an admissible improvement. 
Our new formulation of dominance and admissibility  in Section \ref{sec:setting}  is crucial for properly presenting the admissibility results in this paper.
Moreover, 
 \cite{XSF25} did not prove the admissibility of any mean $\overline{\mathrm{eBH}}$ procedures; see their Remark 19.
\end{remark}

\subsection{Point procedures}\label{sec:point-procedures}

We next turn from simultaneous procedures to point procedures, where a single rejection set must be selected.
The point weighted-mean $\overline{\mathrm{eBH}}$ procedures remain fundamental.

\begin{theorem}\label{thm:point-weighted-ebh-complete-class}
For $\mathfrak D\in\mathrm{PP}_\alpha$, there exists an admissible point weighted-mean $\overline{\mathrm{eBH}}$ procedure in $ \mathrm{PP}_\alpha$ that dominates $\mathfrak D$. In particular, all admissible elements of $ \mathrm{PP}_\alpha$ are point weighted-mean $\overline{\mathrm{eBH}}$ procedures. 
\end{theorem}

\begin{proofidea}
The proof uses the weighted-mean representation from Theorem~\ref{thm:admissible-weighted-ebh-dominator}. Since strong dominance never occurs between two different point procedures, we instead maximize pointwise the size of a certified superset of the original rejection set. A similar argument based on Zorn's lemma, followed by a fixed selection rule, yields the desired admissible point procedure.
\end{proofidea}

Thus, the admissible point weighted-mean $\overline{\mathrm{eBH}}$ procedures also form a complete class. When studying admissibility of point procedures, we only need to focus on the point weighted-mean $\overline{\mathrm{eBH}}$ class.

\subsection{Closed BY procedure is not admissible}

In this section, we turn  briefly to p-testing procedures and provide a side result.\footnote{We thank Aaditya Ramdas for raising this question.} Following \cite{XSF25}, we introduce a class of point procedures that work with p-values and are based on a p-to-e calibrator. For such procedures, we use the same dominance and admissibility terminology as for point procedures. A point p-testing procedure that controls FDR at level $\alpha$ is admissible if it is not strictly dominated by another FDR-controlling point p-testing procedure. Write $\mathbf p=(p_1,\dots,p_K)$ for the vector of p-values, let $p_{(1)}\le \dots \le p_{(K)}$ denote its order statistics, and let $h_n=\sum_{j=1}^n j^{-1}$ for $n\in \cK$.
Next, we define the closed BY procedure as in \cite{XSF25}, which is an upgrade of the procedure of \cite{BY01} for arbitrarily dependent p-values.

\begin{definition}
\label{def:cBY}
Given a level $\alpha\in (0,1)$ and a nonempty $A\subseteq \cK$, let
$$E_A^{\mathrm{BY}}(\mathbf p)=\sum_{i\in A}\frac{\mathds 1_{\{h_{|A|}p_i\le \alpha\}}}{\alpha\bigl(\lceil |A|h_{|A|}p_i/\alpha\rceil\vee 1\bigr)},$$
and take $E_\varnothing^{\mathrm{BY}}(\mathbf p)=1$. For every input $\mathbf p$, the \emph{closed BY procedure} $\overline{\mathrm{BY}}_\alpha(\mathbf p)$ rejects one set of the largest size in the candidate discovery sets
$$\mathcal C_\alpha^{\mathrm{BY}}(\mathbf p)=\left\{R\subseteq\cK:E_A^{\mathrm{BY}}(\mathbf p)\ge\frac{\FDP_A(R)}{\alpha}\text{ for all }A\subseteq\cK\right\}.$$
\end{definition}

The only result in this section shows that the closed BY procedure is not admissible and there exists an explicit upgrade. Choose any $t_0\in(\alpha,5\alpha/3]$, and set
$$\mathcal A=\left\{\mathbf p\in [0,1]^K: p_{(K-1)}\le \frac{\alpha}{K h_K},\ \alpha<p_{(K)}\le t_0\,\right\}.$$

\begin{proposition}\label{prop:closed-by-not-admissible}
For $\alpha\in (0,1)$, the closed BY procedure $\overline{\mathrm{BY}}_\alpha$ is strictly dominated by the following point procedure
$$
\mathfrak D(\mathbf p)
=
\begin{cases}
\cK, & \mathbf p\in \mathcal A,\\[1mm]
\overline{\mathrm{BY}}_\alpha(\mathbf p), & \mathbf p\notin \mathcal A.
\end{cases}
$$
\end{proposition}

The inadmissibility result shown in
Proposition~\ref{prop:closed-by-not-admissible} 
depends on the choice of the e-collection in Definition \ref{def:cBY}.
It is unclear whether a different choice of the e-collection can lead to an admissible p-testing procedure.

\section{Symmetric procedures}\label{sec:symmetric-ebh}

We now specialize the preceding results to symmetric procedures. 
In many statistical applications, it is desirable to use a symmetric procedure when there is no side information about the relative importance or supporting evidence among the hypotheses.
Symmetry is satisfied by most commonly  used procedures for both p-values and e-values. In our context, it is satisfied by  the mean $\overline{\mathrm{eBH}}$  procedure. 

This class is also technically important, because it helps to give an explicit rule of thumb for choosing the constant terms without prior side information. 

\subsection{Symmetric weighted-mean closed eBH}\label{sec:symmetric_characterization}

First, it is straightforward to verify that a weighted-mean $\overline{\mathrm{eBH}}$ procedure is symmetric if it is based on a symmetric e-collection $\{E_A\}_{A\subseteq \cK}$ given by  
\begin{equation}\label{eq:symmetry_e_collection}
    E_A=\mathbb M_{\lambda_0^{|A|}}^{\mathrm{sym}}(\mathbf E^A),\quad  \lambda_0^{|A|}\in [0,1],\quad \mbox{for all $A\subseteq \cK$}.
\end{equation}
In \eqref{eq:symmetry_e_collection}, the parameter $\bm \lambda^A$ in Definition \ref{def:closed-eBH-procedure} is symmetric: it only depends on the size of $A$, and its only flexibility is the constant term.  In particular, each   $(\lambda_0^1,\dots,\lambda_0^K)\in [0,1]^K$ defines a symmetric weighted-mean $\overline{\mathrm{eBH}}$ procedure. 

The converse statement is a bit more delicate. 
If a weighted-mean $\overline{\mathrm{eBH}}$ is symmetric, its parameter $\bm\lambda$ does not need to satisfy the symmetric form in \eqref{eq:symmetry_e_collection}.  Example~\ref{ex:same_procedure_different_lambda} in Appendix~\ref{app:counterexamples} gives  a symmetric weighted-mean $\overline{\mathrm{eBH}}$ procedure with an asymmetric $\bm\lambda$. 
Nevertheless, the following proposition shows that the class \eqref{eq:symmetry_e_collection} covers all symmetric weighted-mean $\overline{\mathrm{eBH}}$ procedures. In other words, any symmetric  weighted-mean $\overline{\mathrm{eBH}}$ procedure admits 
one representing parameter $\bm \lambda$ that is symmetric;  
recall that there are possibly many parameters $\bm\lambda$ that correspond to the same procedure. 

\begin{proposition}\label{prop:symmetric-e-collection}
A weighted-mean $\overline{\mathrm{eBH}}$ procedure $\mathcal D$ is symmetric if and only if there exists $(\lambda_0^1,\dots,\lambda_0^K)\in [0,1]^K$ such that $\mathcal D$ is based on the e-collection \eqref{eq:symmetry_e_collection}.
\end{proposition}

Using Proposition~\ref{prop:symmetric-e-collection},  it suffices to consider symmetric weighted-mean $\overline{\mathrm{eBH}}$ procedures with the e-collection in  \eqref{eq:symmetry_e_collection}. 
Next, we show that the class \eqref{eq:symmetry_e_collection} contains all admissible elements among symmetric simultaneous procedures.

\begin{proposition}\label{prop:symmetric-weighted-ebh-representation}
Let $\mathcal D\in\mathrm{SP}_\alpha^{\mathrm{sym}}$. There exists a symmetric weighted-mean $\overline{\mathrm{eBH}}$ procedure $\overline{\mathrm{eBH}}_\alpha^{\bm{\lambda}}$ such that $\mathcal D\preceq_{\mathrm s}\overline{\mathrm{eBH}}_\alpha^{\bm{\lambda}}$.    
\end{proposition}

Based on Proposition~\ref{prop:symmetric-weighted-ebh-representation}, the representation theorem for symmetric point procedures is obtained easily by viewing them as singleton-valued procedures.

\begin{proposition}\label{prop:symmetric-point-representation}
Let $\mathfrak D\in\mathrm{PP}_\alpha^{\mathrm{sym}}$. There exists a symmetric weighted-mean $\overline{\mathrm{eBH}}$ procedure $\overline{\mathrm{eBH}}_\alpha^{\bm\lambda}$ such that $\mathfrak D(\mathbf e)\in \overline{\mathrm{eBH}}_\alpha^{\bm\lambda}(\mathbf e)$ for all $\mathbf e\in[0,\infty]^K$. 
\end{proposition}

The representation respects symmetry for simultaneous procedures. If the original procedure treats all labels equally, then the dominating weighted-mean $\overline{\mathrm{eBH}}$ procedure can also be chosen to treat labels equally. However, the situation is different for point procedures. Specifically, we can find a dominating point weighted-mean $\overline{\mathrm{eBH}}$ procedure, but this procedure need not itself be symmetric. The reason is that a point weighted-mean $\overline{\mathrm{eBH}}$ procedure is defined to reject any maximal set of the corresponding weighted-mean $\overline{\mathrm{eBH}}$ procedure. This failure is illustrated by Example~\ref{ex:dominator-need-not-symmetric} in Appendix~\ref{app:counterexamples}.

\subsection{Mean closed eBH within the symmetric class}\label{sec:mean-procedures}

The mean $\overline{\mathrm{eBH}}$ procedure $\overline{\mathrm{eBH}}_\alpha^{\mathsf{m}}$ is the most commonly used symmetric procedure. By Theorem~\ref{thm:constant-free-weighted-ebh-admissible}, it is admissible in $\mathrm{SP}_\alpha$, and hence admissible in the smaller class $\mathrm{SP}_\alpha^{\mathrm{sym}}$, for every $\alpha\in(0,1)$. The following theorem further clarifies how to identify it within $\mathrm{SP}_\alpha^{\mathrm{sym}}$.


\begin{theorem}\label{thm:mean-ebh-symmetric-largest}
With respect to the class $\mathrm{SP}_\alpha^{\mathrm{sym}}$,
the following are equivalent:
\begin{enumerate}[(i)]

\item 
$\alpha<1/K$;
\item the mean $\overline{\mathrm{eBH}}$ procedure $\overline{\mathrm{eBH}}_\alpha^{\mathsf{m}}$ is the largest element under $\preceq_{\mathrm s}$; 
\item the mean $\overline{\mathrm{eBH}}$ procedure $\overline{\mathrm{eBH}}_\alpha^{\mathsf{m}}$ is the largest element under $\preceq_{\mathrm w}$; 
\item there exists a largest element under either  $\preceq_{\mathrm s}$ or $\preceq_{\mathrm w}$.
\end{enumerate}
\end{theorem}

\begin{proof}  
\underline{(i)$\Rightarrow$(ii)}: For any $\mathcal D\in\mathrm{SP}_\alpha^{\mathrm{sym}}$, Proposition~\ref{prop:symmetric-weighted-ebh-representation} gives a symmetric weighted-mean $\overline{\mathrm{eBH}}$ procedure $\overline{\mathrm{eBH}}_\alpha^{\bm{\lambda}}$ with constant terms $(\lambda_0^1,\dots,\lambda_0^K)\in [0,1]^K$ such that $\mathcal D\preceq_{\mathrm s} \overline{\mathrm{eBH}}_\alpha^{\bm{\lambda}}$. It suffices to show $\overline{\mathrm{eBH}}_\alpha^{\bm{\lambda}} \preceq_{\mathrm s}\overline{\mathrm{eBH}}_\alpha^{\mathsf{m}}$. For any $\mathbf e$, take $R\in\overline{\mathrm{eBH}}_\alpha^{\bm{\lambda}}(\mathbf e)$ and a nonempty $A\subseteq\cK$. If $A\cap R=\varnothing$ or $\max (\mathbf e^{A})=\infty$, there is nothing to prove. It remains to consider the case where $A\cap R\neq \varnothing$ and $e_i<\infty$ for all $i\in A$. When $\alpha<1/K$, it is easy to see $|A\cap R|/(\alpha |R|)\ge 1/(\alpha K)>1$. From 
$$\lambda_0^{|A|}+(1-\lambda_0^{|A|})\bar e_A\ge \frac{|A\cap R|}{\alpha |R|},$$
we must have $\lambda_0^{|A|}<1$, and then
$$\bar e_A\ge\frac{|A\cap R|/(\alpha |R|)-\lambda_0^{|A|}}{1-\lambda_0^{|A|}}\ge \frac{|A\cap R|}{\alpha |R|}.$$
Thus $R\in\overline{\mathrm{eBH}}_\alpha^{\mathsf{m}}(\mathbf e)$. Therefore $\overline{\mathrm{eBH}}_\alpha^{\bm{\lambda}}\preceq_{\mathrm s}\overline{\mathrm{eBH}}_\alpha^{\mathsf{m}}$, so $\overline{\mathrm{eBH}}_\alpha^{\mathsf{m}}$ is largest.

The implications \underline{(ii)$\Rightarrow$(iii)$\Rightarrow$(iv)} are trivial. 

\underline{(iv)$\Rightarrow$(i)}: Suppose $\alpha\ge1/K$. Define
$$
E'_A(\mathbf e)
=
\begin{cases}
1, & |A|=1,\\
\bar e_A, & |A|\ge2\text{ and }e_i<\infty\text{ for all }i\in A,\\
\infty, & |A|\ge2\text{ and }e_i=\infty\text{ for some }i\in A.
\end{cases}
$$
Let $\mathcal D'$ be the $\overline{\mathrm{eBH}}$ procedure based on the e-collection $E'_A$. Then $\mathcal D'\in\mathrm{SP}_\alpha^{\mathrm{sym}}$. Choose $\mathbf y=(0,\infty,\dots,\infty)$. We claim $\cK \in\mathcal D'(\mathbf y)$. When $|A|=1$, we have $E'_A(\mathbf y)=1\ge 1/ (\alpha K)$. When $|A|\ge2$, $E'_A(\mathbf y)\ge |A| / (\alpha K)$ is straightforward because $E'_A(\mathbf y)=\infty$. But $\cK \notin\overline{\mathrm{eBH}}_\alpha^{\mathsf{m}}(\mathbf y)$, because taking $A=\{1\}$ we have $\bar e_A=0<1/ (\alpha K)$. Thus $\mathcal D'\not\preceq_{\mathrm w}\overline{\mathrm{eBH}}_\alpha^{\mathsf{m}}$, so $\overline{\mathrm{eBH}}_\alpha^{\mathsf{m}}$ is not the largest under $\preceq_{\rm w}$.
Because 
$\overline{\mathrm{eBH}}_\alpha^{\mathsf{m}}$ is admissible within  $\mathrm{SP}_\alpha$ (Theorem~\ref{thm:constant-free-weighted-ebh-admissible}),
no procedure in  $\mathrm{SP}_\alpha^{\mathrm{sym}}$ strictly weakly dominates $\overline{\mathrm{eBH}}_\alpha^{\mathsf{m}}$. Therefore,
there exists no largest element in the class $\mathrm{SP}_\alpha^{\mathrm{sym}}$ under $\preceq_{\rm w}$. Therefore, (iv) implies $\alpha<1/K$.
\end{proof}


Unlike $\overline{\mathrm{eBH}}_\alpha^{\mathsf{m}}$, the characterization of the commonly used point procedure $\overline{\mathrm{ebh}}_\alpha^{\mathsf{m}}$ must be split into two cases. The reason for this is as follows:
although a point mean $\overline{\mathrm{eBH}}$ procedure generally depends on the choice of a maximal set, this choice is always unique when $K=2$: if both $\{1\}$ and $\{2\}$ belong to $\overline{\mathrm{eBH}}_\alpha^{\mathsf m}(\mathbf e)$, then $\{1,2\}\in\overline{\mathrm{eBH}}_\alpha^{\mathsf m}(\mathbf e)$ necessarily holds. Clearly, this uniqueness may fail for $K\ge 3$, leading to a different situation.


\begin{proposition}\label{prop:symmetric-point-largest}
Let $\alpha\in (0,1)$ and consider the class $\mathrm{PP}_\alpha^{\mathrm{sym}}$.
\begin{enumerate}[(i)]
\item If  $\alpha\ge (K-1)/K$,  $\mathrm{PP}_\alpha^{\mathrm{sym}}$ has the largest element given by 
\begin{equation}\label{eq:ppsym-max}
\mathfrak D_\alpha(\mathbf e)
=
\begin{cases}
\mathcal K, & \sum_{i=1}^K e_i\ge K/\alpha,\\
\varnothing, & \sum_{i=1}^K e_i<K/\alpha.
\end{cases}
\end{equation} 
\item For $K=2$, $\mathrm{PP}_\alpha^{\mathrm{sym}}$ 
has the largest element $\overline{\mathrm{ebh}}_\alpha^{\mathsf{m}}$  if and only if $\alpha<1/2$.

\item For $K\ge 3$, $\mathrm{PP}_\alpha^{\mathrm{sym}}$ has no largest element if  $\alpha < (K-1)/K$. 
\end{enumerate}
\end{proposition}

We can directly check that the point procedure  $\mathfrak D_\alpha$ in \eqref{eq:ppsym-max} has FDR control at $\alpha$ only when $\alpha \ge (K-1)/K$.

For a symmetric simultaneous procedure $\mathcal D$ and every $\mathbf e\in[0,\infty]^K$, define 
$$
\mathcal D^{\mathrm{inv}}(\mathbf e)
=
\left\{
R\in\mathcal D(\mathbf e):
\sigma(R)=R
\text{ for every permutation }\sigma \text{ on } \cK
\text{ with }\mathbf e_\sigma=\mathbf e
\right\}.
$$
We discuss below some subtle issues on symmetric point weighted-mean $\overline{\mathrm{eBH}}$ procedures. 

\begin{itemize}

\item[(i)] \textbf{Incompleteness.}
As Example~\ref{ex:dominator-need-not-symmetric} shows, the class of symmetric point weighted-mean $\overline{\mathrm{eBH}}$ procedures does not form a complete class for symmetric point procedures.

\item[(ii)] \textbf{Maximality versus symmetry.}
Suppose that a procedure $\overline{\mathrm{ebh}}_\alpha^{\bm \lambda}$ is symmetric. For every permutation $\sigma$ on $\cK$ with $\mathbf e_\sigma=\mathbf e$, we have $\overline{\mathrm{ebh}}_\alpha^{\bm \lambda}(\mathbf e)=\overline{\mathrm{ebh}}_\alpha^{\bm \lambda}(\mathbf e_\sigma)=\sigma^{-1}(\overline{\mathrm{ebh}}_\alpha^{\bm \lambda}(\mathbf e))$. Hence $\overline{\mathrm{ebh}}_\alpha^{\bm \lambda}(\mathbf e)\in(\overline{\mathrm{eBH}}_\alpha^{\bm \lambda})^{\mathrm{inv}}(\mathbf e)$. In particular, if $e_i=e_j$, then symmetry does not allow the procedure to include exactly one of $i$ and $j$. This may violate the definition that $\overline{\mathrm{ebh}}_\alpha^{\bm \lambda}$ rejects one maximal set of $\overline{\mathrm{eBH}}_\alpha^{\bm \lambda}$, which is the failure mechanism behind Example~\ref{ex:dominator-need-not-symmetric}.

\item[(iii)] \textbf{Monotonicity.}
Let $K=3$. Choose $1 / 2 <\alpha x<\alpha y< 2/ 3$, and set $\mathbf u=(y,x,\infty)$ and $\mathbf v=(y,y,\infty)$. We have $\mathbf u\le\mathbf v$ and, by direct verification, $\overline{\mathrm{eBH}}_\alpha^{\mathsf{m}}(\mathbf u)=\overline{\mathrm{eBH}}_\alpha^{\mathsf{m}}(\mathbf v)=\{\varnothing,\{3\},\{1,3\},\{2,3\}\}$. Then the sets of the maximal sets of $(\overline{\mathrm{eBH}}_\alpha^{\mathsf{m}})^{\mathrm{inv}}(\mathbf u)$ and $(\overline{\mathrm{eBH}}_\alpha^{\mathsf{m}})^{\mathrm{inv}}(\mathbf v)$ are $\{\{1,3\},\{2,3\}\}$ and $\{\{3\}\}$, respectively. Therefore increasing a coordinate can force the symmetric procedure to discard a rejection. Hence symmetry comes at a cost of monotonicity of the point  procedures.

\end{itemize}

These issues highlight that symmetric point weighted-mean $\overline{\mathrm{eBH}}$ procedures are not a natural class of procedures in theory. The core issue is that maximality conflicts with symmetry. Therefore, we will  only consider symmetry within weighted-mean simultaneous procedures in the rest of this section.


\subsection{A rule of thumb for choosing constant terms}\label{sec:constant-term}

The constant terms in symmetric weighted-mean $\overline{\mathrm{eBH}}$ procedures deserve careful thought because they offer a simple way to modify conservativeness. They cannot be chosen freely; for example, a large constant can make the e-collection less responsive to actual large input e-values, thereby yielding a procedure that is inadmissible. In this subsection, we provide a rule of thumb for choosing the constant terms $\bm{\lambda}$ in a symmetric $\overline{\mathrm{eBH}}_\alpha^{\bm{\lambda}}$, without any side information as in Section \ref{sec:prior-score-weights}. 

We first give a clean sufficient condition on the sequence of constant terms to make the corresponding symmetric weighted-mean $\overline{\mathrm{eBH}}$ procedures admissible.

\begin{theorem}\label{thm:constant-term-admissibility-criterion}
For $\alpha\in(0,1)$, if
\begin{equation}\label{eq:simple-constant-cap} 
\lambda_0^K\le\cdots\le\lambda_0^2\le\lambda_0^1 \le \frac{1}{\alpha(\lfloor1/\alpha\rfloor+1)}, 
\end{equation} 
\begin{equation}\label{eq:simple-constant-tail} 
\lambda_0^n=0 \text{ for all } n>K-\lfloor1/\alpha\rfloor, 
\end{equation} 
and 
\begin{equation}\label{eq:simple-constant-smooth} 
\lambda_0^{n-1}-\lambda_0^n < \frac{1-\lambda_0^n}{n(1-\alpha)} \text{ for all } n\ge 2,
\end{equation}
then the corresponding procedure $\overline{\mathrm{eBH}}_\alpha^{\bm{\lambda}}$ is admissible in $\mathrm{SP}_\alpha$.
\end{theorem}


By Proposition~\ref{prop:esc-dominated-by-weighted-ebh}, we have $\mathrm{eSC}_\alpha\preceq_{\mathrm s}\overline{\mathrm{eBH}}_\alpha^{\bm{\lambda}}$ if and only if, for all $n\in \cK$,
$$\lambda_0^n \le \frac{K-n}{K-\alpha n}.$$
For a given $\beta\in (0,1]$, define a sequence of constant terms
\begin{equation*}
\lambda_0^n(\beta)
=
\begin{cases}
\dfrac{\beta K-n}{\beta K-\alpha n}, & 1\le n<\beta K,\\[10pt]
0, & \beta K\le n\le K,
\end{cases}
\end{equation*}
Let $\overline{\mathrm{eBH}}_\alpha^{\bm{\lambda}(\beta)}$ be the corresponding symmetric weighted-mean $\overline{\mathrm{eBH}}$ procedure. Naturally, we can extend this notation by defining $\overline{\mathrm{eBH}}_\alpha^{\bm{\lambda}(0)}$ to be the mean $\overline{\mathrm{eBH}}$ procedure.
Since $(\beta K-n) / (\beta K-\alpha n) \le (K-n) / (K-\alpha n)$ for every $n<\beta K$, we have $\mathrm{eSC}_\alpha\preceq_{\mathrm s}\overline{\mathrm{eBH}}_\alpha^{\bm{\lambda}(\beta)}$ for every $\beta\in (0,1]$. However, the one-parameter family $\overline{\mathrm{eBH}}_\alpha^{\bm{\lambda}(\beta)}$, although dominating $\mathrm{eSC}_\alpha$, is not necessarily admissible. In contrast to the sufficient condition in Theorem~\ref{thm:constant-term-admissibility-criterion}, this family has a sharp admissibility bound. Let
\begin{equation}
\label{eq:beta-star-explicit}
\beta^*
=
\frac1K
\min\left\{
\frac{\alpha \lfloor 1/\alpha\rfloor}
{\alpha(\lfloor 1/\alpha\rfloor+1)-1},
\,
\max\left\{1,K-\left\lfloor \frac1\alpha\right\rfloor+1\right\}
\right\}.
\end{equation}
For example, if $1/\alpha$ is an integer and $K\ge 2/\alpha-1$, then $\beta^*=1 / (\alpha K)$. 


\begin{proposition}\label{prop:beta-threshold-admissibility}
For $\alpha\in(0,1)$, $\overline{\mathrm{eBH}}_\alpha^{\bm{\lambda}(\beta)}$ is admissible in $\mathrm{SP}_\alpha$ if and only if $\beta\le\beta^*$. 
\end{proposition}

\begin{proofidea}
When $\beta\le\beta^*$, the constants satisfy the admissibility conditions of Theorem~\ref{thm:constant-term-admissibility-criterion}. Conversely, if $\beta>\beta^*$, we can find some constant $\lambda_0^n(\beta)$ that is strictly larger than a constructed threshold. Lowering only this constant to the threshold yields a procedure that is strictly better.
\end{proofidea}

In practice, this result gives a recommendation for choosing constant terms: using $\overline{\mathrm{eBH}}_\alpha^{\bm{\lambda}(\beta)}$ for $\beta\le\beta^*$. A smaller $\beta$ gives a procedure closer to $\overline{\mathrm{eBH}}_\alpha^{\mathsf{m}}$, while a larger $\beta$ keeps positive constants for more intersection sizes. The extreme choice $\beta=\beta^*$ specifies the boundary for $\overline{\mathrm{eBH}}_\alpha^{\bm{\lambda}(\beta)}$ to be admissible.
In Section \ref{sec:sim-constant-terms}, we show by simulation that inadmissible choices of $\beta$ can lead to inferior  performance. 

\section{Simulation studies}\label{sec:simulations}

The theoretical results leave a practically important choice within the class of weighted-mean $\overline{\mathrm{eBH}}$ procedures: the allocation of weights and the selection of constant terms. We use two simple Gaussian mean testing experiments to illustrate the advantages of our methods given in Sections \ref{sec:prior-score-weights} and \ref{sec:constant-term}.

In the simulation studies, we only consider e-testing procedures, and avoid a comparison with p-testing procedures, for two possible reasons:  first, there may be  dependence among data unknown to the user (although we run simulation for independent data); second, they work in the situation where only the e-values are available, but not the full distributions of the test statistics.
 
\subsection{Learning asymmetric weights}
\label{sec:sim-asymmetric-weights}

We consider the standard problem of testing one-sided Gaussian means using e-values. Let $X_i \sim \mathcal N(\mu_i,1)$ for $i\in\cK$. We test $H_i:\mu_i\leq 0$ against the alternative $\mu_i>0$. Let the data be independent across $i$, set $K=200$ and $\alpha=0.05$, and let the number of non-nulls be $|N^c|\in\{20,30\}$, corresponding to non-null proportions $10\%$ and $15\%$.

For each hypothesis $i$, there are 20 pre-screening (or training) data
$X_{i,1},\dots,X_{i,20}$ and 40   testing data 
$Y_{i,1},\dots,Y_{i,40}$. 
For a parameter $\eta>0$, let 
$$E_{i,P}= e^{\sum_{t=1}^{20}(\eta X_{i,t}-  \eta^2/2) },~ 
E_{i,T}=e^{\sum_{t=1}^{40}(\eta Y_{i,t}-  \eta^2/2) }, \mbox{ and }
E_{i,F}= E_{i,P}E_{i,T}.$$
The statistics $E_{i,P}$, $E_{i,T}$
and $E_{i,F}$ are computed from the pre-screening data, the testing data, and the full data, respectively. 
The testing data $Y_{i,1},\dots,Y_{i,40}$
are iid copies of $X_i$, and hence $E_{i,T}$ is a valid e-value for hypothesis $i$, as used by \cite{VW21}. In contrast, the pre-screening data $X_{i,1},\dots,X_{i,20}$ can have lower quality than the testing data (specified below), and hence $E_{i,P}$ and $E_{i,F}$ may not be valid e-values. 
We take $\mu_i=0$ for every null hypothesis and the common signal strength $\mu_i=0.45$ for every alternative hypothesis, and let $\eta=0.45$. The effective signal strength for the testing data is $0.45\sqrt{40}\approx 2.846$.

We consider three settings of the pre-screening data for each hypothesis $i$. In all these settings, the pre-screening data are independent of the testing data.
\begin{enumerate}[(a)]
\item Clean data: The pre-screening data and the testing data are iid.  

\item Power distortion:    $X_{i,t}$ are iid across $t$ following the distribution of $X_{i}^2\operatorname{sgn}(X_{i})$. This power distortion models nonlinear miscalibration: the recorded measurements preserve their signs and ordering, but large responses are disproportionately amplified.  

\item Persistent measurement error: $X_{i,t}=Z_{i,t}+\varepsilon_i$, where $Z_{i,t} $ are iid across $t$ from $X_i$ and $\varepsilon_i \sim \mathcal N(0,0.1)$. This models measurement error in pre-screening subjects: a cheaper and less accurate measurement device is used for the training subjects. The cheaper device induces an error $\varepsilon_i$ shared by all pre-screening data for the same hypothesis. 
\end{enumerate}

We consider five methods, with the first four applied with e-values  $(E_{i,T})_{i\in \cK}$ based on the testing data:
the base eBH, 
the mean $\overline{\mathrm{eBH}}$,
the normalized weighted-mean $\overline{\mathrm{eBH}}$,
the
adjusted-normalized weighted-mean $\overline{\mathrm{eBH}}$,
and the mean $\overline{\mathrm{eBH}}$ applied with full data $(E_{i,F})_{i\in \cK}$.  
For the two weighted-mean $\overline{\mathrm{eBH}}$ procedures  in Section \ref{sec:prior-score-weights}, we compute the weights $\bm\lambda_{\rm nor}$ and 
 $\bm\lambda_{\rm adj}$  from the   pre-screening data by assigning $H_i$ the likelihood-ratio score
$s_i=E_{i,P}$  in \eqref{eq:prior-score-weights}. We set the parameter $\gamma=1$. 


In Settings (b) and (c),  the full-data mean $\overline{\mathrm{eBH}}$ 
does not have theoretical FDR control, because  $E_{1,F},\dots,E_{K,F}$ are
 not necessarily valid e-values. 
All other methods have theoretical FDR control.

For the  mean and weighted-mean $\overline{\mathrm{eBH}}$ procedures, we report a rejection set with the largest size, which may need some tie breaking; tie breaking is done through favoring hypotheses with larger e-values.
After obtaining a reported rejection set $R$, performance is measured by the empirical 
values of the
true positive rate $\mathrm{TPR}:=\mathbb E[|R\cap N^c|/|N^c|]$ and $\FDR$. All results are based on the average of  $1{,}000$ independent replications.

\begin{table}[t]
\centering
\footnotesize
\caption{Empirical TPR and FDR at $\alpha=0.05$ based on $1{,}000$ independent replications. Monte Carlo standard errors are shown in parentheses.}
\label{tab:learned-weight-power}
\begingroup
\setlength{\tabcolsep}{3.5pt}
\renewcommand{\arraystretch}{1.2}

\begin{tabularx}{\textwidth}{@{}Xcccc@{}}
\toprule 
& \multicolumn{2}{c}{$|N^c|=20$}
& \multicolumn{2}{c}{$|N^c|=30$} \\
\cmidrule(lr){2-3}\cmidrule(lr){4-5}
& TPR  & FDR 
& TPR  & FDR  \\
\midrule

\multicolumn{5}{@{}l}{\textit{Only testing data}}\\
\addlinespace[2pt]

eBH, test-only
& 0.1584 (0.0039) & 0.0026 (0.0012)
& 0.2096 (0.0038) & 0.0029 (0.0007) \\

mean $\overline{\mathrm{eBH}}$, test-only
& 0.1790 (0.0042) & 0.0036 (0.0012)
& 0.2506 (0.0040) & 0.0040 (0.0007) \\\midrule
\addlinespace[2pt]

\multicolumn{5}{@{}l}{\textit{Clean  data}}\\
\addlinespace[2pt]

normalized 
& 0.4708 (0.0055) & 0.0044 (0.0006)
& 0.5174 (0.0049) & 0.0041 (0.0005) \\

adjusted-normalized 
& 0.2478 (0.0051) & 0.0036 (0.0012)
& 0.3500 (0.0046) & 0.0043 (0.0006) \\

mean $\overline{\mathrm{eBH}}$, all data
& 0.5370 (0.0043) & 0.0071 (0.0008)
& 0.6027 (0.0034) & 0.0056 (0.0005) \\

\midrule

\multicolumn{5}{@{}l}{\textit{Power-distortion in the pre-screening data}}\\
\addlinespace[3pt]



normalized 
& 0.2742 (0.0060) & 0.0020 (0.0006)
& 0.2747 (0.0059) & 0.0016 (0.0004) \\

adjusted-normalized
& 0.2304 (0.0050) & 0.0033 (0.0012)
& 0.3151 (0.0046) & 0.0038 (0.0006) \\

mean $\overline{\mathrm{eBH}}$, all data
& 0.7652 (0.0032) & \textbf{0.0710} (0.0020)
& 0.7945 (0.0025) & \textbf{0.0580} (0.0015) \\

\midrule

\multicolumn{5}{@{}l}{\textit{Persistent measurement error in the pre-screening data}}\\
\addlinespace[3pt]



normalized
& 0.1768 (0.0044) & 0.0045 (0.0010)
& 0.2038 (0.0041) & 0.0040 (0.0007) \\

adjusted-normalized
& 0.2016 (0.0044) & 0.0032 (0.0012)
& 0.2857 (0.0043) & 0.0038 (0.0006) \\

mean $\overline{\mathrm{eBH}}$, all data
& 0.5388 (0.0042) & \textbf{0.0708} (0.0023)
& 0.5853 (0.0033) & \textbf{0.0581} (0.0017) \\

\bottomrule
\end{tabularx}
\endgroup
\end{table}

Table~\ref{tab:learned-weight-power} presents the results. Because the results for $20$ non-nulls and $30$ non-nulls are similar, we focus on the $20$ non-nulls setting. In the case of clean data, all empirical FDRs are well below the level $\alpha$, as theoretically anticipated. 
The mean $\overline{\mathrm{eBH}}$ procedure using all data attains a TPR of $0.5370$, whereas withholding the pre-screening data reduces the TPR to $0.1790$. This gap is expected because the full-data procedure uses more observations. The procedure $\overline{\mathrm{eBH}}_\alpha^{\bm\lambda_{\rm nor}}$ raises the TPR to $0.4708$, thereby recovering approximately $82\%$ of the empirical power loss caused by data splitting. The procedure $\overline{\mathrm{eBH}}_\alpha^{\bm\lambda_{\rm adj}}$ attains a TPR of $0.2478$, again better than the test-only procedures,  and its   rejection set contains the eBH rejection set by construction.

By Proposition~\ref{prop:esc-dominated-by-weighted-ebh}, the point procedures induced by $\overline{\mathrm{eBH}}_\alpha^{\bm\lambda_{\rm nor}}$ do not dominate the base eBH procedure. Nevertheless, in the iid setting with $20$ non-nulls, the reported globally largest rejection set contains the test-only eBH rejection set in $872/1000$ replications. Thus, the reported rejection set contains the eBH rejection set with high probability.

Under power distortion, both learned procedures improve upon the test-only mean $\overline{\mathrm{eBH}}$ procedure: their TPRs are 
$0.2742$ and $0.2304$, compared with $0.1790$ for the test-only mean. Persistent measurement error yields a different pattern. The TPR of $\overline{\mathrm{eBH}}_\alpha^{\bm\lambda_{\rm nor}}$ falls to $0.1768$, whereas $\overline{\mathrm{eBH}}_\alpha^{\bm\lambda_{\rm adj}}$ attains a TPR of $0.2016$. This is  
because the measurement errors reduce the quality of the signals in pre-screening and induce noises. The empirical FDRs of both learned procedures remain well below $\alpha$. By contrast, the mean $\overline{\mathrm{eBH}}$ procedure using all data has empirical FDRs of $0.0710$ under power distortion and $0.0708$ under persistent measurement error, with $95\%$ normal Monte Carlo intervals $[0.0670,0.0749]$ and $[0.0662,0.0754]$, respectively. These intervals lie   above $\alpha=0.05$. Thus,  directly pooling the pre-screening and testing observations can lose FDR control if contamination is possible but unknown. 

Together, the results illustrate that asymmetric weights learned from side information can yield substantial power gains without losing FDR control in case of unknown contamination in the pre-screening data. This  supports  our proposed rules of thumb for choosing asymmetric weights.

\subsection{Effect of admissible constant terms}
\label{sec:sim-constant-terms}

We continue with the setting of testing one-sided Gaussian means. Let $X_i\sim\mathcal{N}(\mu_i,1)$ for $i\in\cK$, and test $H_i: \mu_i\leq 0$ against the alternative $\mu_i>0$. The observations are independent across $i$, and we set $K=5000$ and $\alpha=0.05$. Consider three cases $|N^c|= 500$, $1000$, and $1500$, corresponding to non-null proportions $10\%$, $20\%$, and $30\%$, respectively. Let $X_i\sim\mathcal{N}(0,1)$ for $i\in N$ and $X_i\sim\mathcal{N}(3,1)$ for $i\in N^c$. Define $E_i=\exp(3 X_i- 9/2)$, which is a valid e-value for hypothesis $i$.

By Proposition~\ref{prop:beta-threshold-admissibility}, the exact admissibility threshold $\beta^*=0.004$. We will take $\beta\in\{0,0.004,0.80,0.90,1\}$. Thus, the first two values of $\beta$ yield admissible procedures, whereas the last three choices do not. The procedures corresponding to different values of $\beta$ are generally incomparable. For each $\beta$, we report a rejection set of maximum size consisting of the hypotheses with the largest e-values. By Proposition~\ref{prop:esc-dominated-by-weighted-ebh}, the resulting rejection set contains the base $\mathrm{eBH}_\alpha$ rejection set.
We use TPR to measure the performance. 

\begin{table}[t]
\centering
\footnotesize
\caption{Empirical TPR at $\alpha=0.05$ based on $1{,}000$ independent replications. Monte Carlo standard errors are shown in parentheses.}
\label{tab:constant-terms-tpr}
\begingroup
\setlength{\tabcolsep}{6pt}
\renewcommand{\arraystretch}{1.2}

\begin{tabularx}{\textwidth}{@{}Xccc@{}}
\toprule
& \multicolumn{3}{c}{Number of non-nulls $|N^c|$} \\
\cmidrule(lr){2-4}
& $500$ & $1000$ & $1500$ \\
\midrule
\addlinespace[2pt]

$\mathrm{eBH}_{\alpha}$
& 0.2227 (0.001041)
& 0.3509 (0.000767)
& 0.4275 (0.000600) \\

\midrule
\addlinespace[2pt]

\multicolumn{4}{@{}l}{\textit{Admissible choices:
$\beta\leq\beta^{*}=0.004$}}\\
\addlinespace[2pt]

$\beta=0$ ($\overline{\mathrm{eBH}}_{\alpha}^{\mathsf m}$) and $\beta=0.004$
& 0.2752 (0.001095)
& 0.4320 (0.000748)
& 0.5317 (0.000573) \\

\midrule
\addlinespace[2pt]

\multicolumn{4}{@{}l}{\textit{Inadmissible choices:
$\beta>\beta^{*}=0.004$}}\\
\addlinespace[2pt]

$\beta=0.80$
& 0.2752 (0.001095)
& 0.4306 (0.000752)
& 0.5233 (0.000588) \\

$\beta=0.90$
& 0.2724 (0.001100)
& 0.4211 (0.000769)
& 0.5112 (0.000599) \\

$\beta=1.00$
& 0.2623 (0.001106)
& 0.4090 (0.000783)
& 0.4980 (0.000609) \\

\bottomrule
\end{tabularx}
\endgroup
\end{table}

Table~\ref{tab:constant-terms-tpr} shows that, in each setting, the mean $\overline{\mathrm{eBH}}$ procedure and the admissible choice $\beta=0.004$ give the same empirical TPR, they reject the same sets in all replications. This occurs because $0.004$ is very close to $0$. The two underlying admissible procedures are not comparable.\footnote{ For example, at $\mathbf e=(10000,19,\ldots,19)$, the two procedures reject
$\{1,\ldots,20\}$ and $\{1\}$, respectively, whereas at $\mathbf e=(21,\ldots,21,0,\ldots,0)$, with $4990$ copies of $21$ and $10$ zeros, they reject $\{1,\ldots,4990\}$ and $\{1,\ldots,5000\}$, respectively. Hence, neither procedure dominates the other.} 
The inadmissible choices can yield lower empirical TPRs than the admissible procedures, although this loss of power is relatively small; for $\beta$ between $0.004$ and $0.8$, the change in TPR is very small, and we omit it. 
This small experiment shows that the study of admissibility has practical relevance in terms of power: in the settings considered, choosing $\beta$ substantially above the admissibility threshold $\beta^*$ can lead to fewer true discoveries, even though the resulting point procedure still dominates eBH.

\section{Conclusion}\label{sec:conclusion}

We establish a complete class theory for FDR-controlling procedures based on e-values under arbitrary dependence. We propose  two dominance relations  for simultaneous procedures, and they lead to the same notion of admissibility.  
As a main result, weighted-mean $\overline{\mathrm{eBH}}$ procedures form a complete class of admissible e-testing procedures. 
Within this complete class, every constant-free weighted-mean $\overline{\mathrm{eBH}}$ procedure is admissible at every level. 
This result is the first established admissibility for FDR-controlling procedures with inputs being either p-values or e-values, although we do not have the corresponding results on the admissibility of p-testing procedures. 


Point procedures have a more complicated theory than simultaneous procedures, and we establish several parallel results for point procedures on the complete class theory.
For instance, we show that every point mean $\overline{\mathrm{eBH}}$ procedure is admissible if and only if $\alpha<1/2$. 
Several results are also obtained for symmetric  weighted-mean $\overline{\mathrm{eBH}}$ procedures, a class that is simpler  to work with in practice. 

The weight parameters in the weighted-mean $\overline{\mathrm{eBH}}$ procedures have high degrees of freedom, leading to much flexibility. 
Our theoretical results   give rise to several useful ways to choose the parameters in practice, which we illustrate with simulation studies. For instance, the normalized and adjusted-normalized rules of choosing weight vectors are useful in the case of data contamination. 
With both the theoretical and the simulation results, we offer a more complete picture on the weighted-mean $\overline{\mathrm{eBH}}$ procedures, as well as their central role in multiple testing. 


\newpage

\begin{center}
   \Large Supplementary Material for ``Admissibility and Complete Classes for False Discovery Rate Control with E-values''
\end{center}

This supplement contains some proofs and examples omitted from the main paper.

\appendix 
 
\section{Omitted proofs}
\label{app:proofs}

\subsection{Omitted proofs from Section~\ref{sec:setting}} 
\label{app:A1}

\begin{proof}[Proof of Proposition~\ref{prop:dominance-orders}]
(i) It suffices to argue pointwise in $\mathbf e$. For $\preceq_{\mathrm s}$, this is just the usual inclusion order on collections of sets. For $\preceq_{\mathrm w}$, consider any two collections of sets $\mathcal A$ and $\mathcal B$. Write $\mathcal A\le_{\mathrm w}\mathcal B$ if either
$\mathcal A^\downarrow\subsetneq\mathcal B^\downarrow$ or $\mathcal A\subseteq\mathcal B$. Then $\mathcal A\le_{\mathrm w}\mathcal B$ always implies
$\mathcal A^\downarrow\subseteq\mathcal B^\downarrow$. Reflexivity follows from
$\mathcal A\subseteq\mathcal A$. If $\mathcal A\le_{\mathrm w}\mathcal B$ and
$\mathcal B\le_{\mathrm w}\mathcal A$, then
$\mathcal A^\downarrow=\mathcal B^\downarrow$; hence the strict-lower-closure alternative is impossible in both directions, and we get
$\mathcal A\subseteq\mathcal B\subseteq\mathcal A$. Thus $\mathcal A=\mathcal B$.
If $\mathcal A\le_{\mathrm w}\mathcal B\le_{\mathrm w}\mathcal C$, then
$\mathcal A^\downarrow\subseteq\mathcal C^\downarrow$. If the inclusion is strict, then
$\mathcal A\le_{\mathrm w}\mathcal C$; if it is equality, then all intermediate lower closures are equal, so
$\mathcal A\subseteq\mathcal B\subseteq\mathcal C$, again giving
$\mathcal A\le_{\mathrm w}\mathcal C$. Hence $\preceq_{\mathrm w}$ is also a partial order.

(ii) Strong dominance implies weak dominance by definition. The converse fails already for
$K=2$: take the constant procedures
$\mathcal D(\mathbf e)=\{\{1\}\}$ and
$\mathcal G(\mathbf e)=\{\{1,2\}\}$. Then
$\mathcal D^\downarrow(\mathbf e)\subsetneq\mathcal G^\downarrow(\mathbf e)$ for all $\mathbf e$, so
$\mathcal D\preceq_{\mathrm w}\mathcal G$, but
$\mathcal D(\mathbf e)\nsubseteq\mathcal G(\mathbf e)$, so
$\mathcal D\not\preceq_{\mathrm s}\mathcal G$.
\end{proof}

A simultaneous procedure $\mathcal D$ is said to be increasing if $\mathcal D(\mathbf x)\subseteq\mathcal D(\mathbf y)$ whenever $\mathbf x\le\mathbf y$. Before proving Proposition~\ref{prop:strong-weak-admissibility-equivalence}, a useful lemma will be given first. 

\begin{lemma}\label{lem:s-admissible-increasing}
For $\alpha\in(0,1)$, every member in $\mathrm{SP}_\alpha$ is s-admissible only if it is increasing.
\end{lemma}

\begin{proof}
Suppose, for contradiction, that there exist an s-admissible $\mathcal D\in \mathrm{SP}_\alpha$, and some $\mathbf x,\mathbf y\in[0,\infty]^K$ such that $\mathbf x\le\mathbf y$ but
$\mathcal D(\mathbf x)\nsubseteq\mathcal D(\mathbf y)$. Define 
$$
\mathcal G(\mathbf z)
=
\begin{cases}
\mathcal D(\mathbf y)\cup\mathcal D(\mathbf x),
& \mathbf z=\mathbf y,\\
\mathcal D(\mathbf z),
& \mathbf z\ne\mathbf y.
\end{cases}
$$
Clearly, the procedure $\mathcal G$ strictly strongly dominates $\mathcal D$. We claim that $\mathcal G\in\mathrm{SP}_\alpha$. 
Fix a true null set $N$ and an input e-value $\mathbf E$.
\begin{itemize}
\item[(i)] If $\max_{R\in\mathcal D(\mathbf x)}\FDP_N(R)\le\max_{R\in\mathcal D(\mathbf y)}\FDP_N(R)$, 
then adding $\mathcal D(\mathbf x)$ to $\mathcal D(\mathbf y)$ does not increase the FDP at $\mathbf y$. Hence
$\FDR_{\mathcal G}(\mathbf E)=\FDR_{\mathcal D}(\mathbf E)\le\alpha$. 

\item[(ii)] If $\max_{R\in\mathcal D(\mathbf x)}\FDP_N(R)>\max_{R\in\mathcal D(\mathbf y)}\FDP_N(R)$, define 
$$
\mathbf X
=
\begin{cases}
\mathbf x, & \mathbf E=\mathbf y,\\
\mathbf E, & \mathbf E\ne\mathbf y.
\end{cases}
$$
Since $\mathbf X\le\mathbf E$, $\mathbb E[X_i]\le\mathbb E[E_i]\le1$ holds for every $i\in N$. For every realized value, we have $\max_{R\in\mathcal G(\mathbf E)}\FDP_N(R)=\max_{R\in\mathcal D(\mathbf X)}\FDP_N(R)$.
Indeed, this is immediate when $\mathbf E\ne\mathbf y$; when $\mathbf E=\mathbf y$, both sides equal
$\max_{R\in\mathcal D(\mathbf x)}\FDP_N(R)$.
Consequently, 
$\FDR_{\mathcal G}(\mathbf E)=\FDR_{\mathcal D}(\mathbf X)\le\alpha$.
\end{itemize}
Thus $\mathcal G\in\mathrm{SP}_\alpha$, contradicting the s-admissibility of $\mathcal D$. Therefore $\mathcal D$ is
increasing.
\end{proof}

Next, we prove Proposition~\ref{prop:strong-weak-admissibility-equivalence}.  
For $\mathbf x\in[0,\infty]^K$ and $S\subseteq\cK$, define
$$
\mathcal D^{\mathbf x,S}(\mathbf z)
=
\begin{cases}
\mathcal D(\mathbf x)\cup\{S\}, & \mathbf z=\mathbf x,\\
\mathcal D(\mathbf z), & \mathbf z\ne\mathbf x.
\end{cases}
$$
If $S\notin\mathcal D(\mathbf x)$, then $\mathcal D^{\mathbf x,S}$ strictly dominates $\mathcal D$ under strong dominance. Moreover, if $\mathcal D$ is s-admissible, then $\mathcal D^{\mathbf x,S}\notin\mathrm{SP}_\alpha$. 

\begin{proof}[Proof of Proposition~\ref{prop:strong-weak-admissibility-equivalence}]
By Proposition~\ref{prop:dominance-orders}, we know that
admissibility implies s-admissibility, and hence
it suffices to verify the converse. Let $\mathcal D\in \mathrm{SP}_\alpha$ be s-admissible. By Lemma~\ref{lem:s-admissible-increasing}, we know that $\mathcal D$ is increasing. Suppose, for contradiction, that $\mathcal D$ is not admissible. Then there exists $\mathcal G\in\mathrm{SP}_\alpha$ such that $\mathcal D\preceq_{\mathrm w}\mathcal G$ and $\mathcal G\ne\mathcal D$.

Step 1: We show that there exist $\mathbf e\in[0,\infty]^K$ and $R\subseteq\cK$ such that $R\in\mathcal G(\mathbf e)\setminus \mathcal D(\mathbf e)$. Indeed, if $\mathcal G(\mathbf x)\subseteq \mathcal D(\mathbf x)$ for every $\mathbf x$, then $\mathcal G(\mathbf x)^\downarrow\subseteq \mathcal D(\mathbf x)^\downarrow$ for every $\mathbf x$. Since $\mathcal D\preceq_{\mathrm w}\mathcal G$, we must have $\mathcal D(\mathbf x)\subseteq\mathcal G(\mathbf x)$ for every $\mathbf x$, and therefore $\mathcal D=\mathcal G$, a contradiction.

Step 2: Among all triples $(\mathbf x,S,N)$ such that $S\in\mathcal G(\mathbf x)\setminus\mathcal D(\mathbf x)$ and
$\mathcal D^{\mathbf x,S}$ fails to control FDR under null set $N$, choose one, denoted by $(\mathbf e_m,R_m,N_m)$, with $|N_m|$ the smallest. Such a triple exists by Step 1. Also, $N_m$ is nonempty. If not, then all FDPs are zero and $\mathcal D^{\mathbf e_m,R_m}$ controls FDR at every level.

Step 3: Since $\mathcal D^{\mathbf e_m,R_m}$ fails to control FDR under null set $N_m$ and differs from $\mathcal D$ only at the point $\mathbf e_m$, there exists an e-value vector $\mathbf X$ such that $\mathbb P(\mathbf X=\mathbf e_m)>0$ and
$\FDR_{\mathcal D^{\mathbf e_m,R_m}}(\mathbf X)>\alpha$. Moreover, since adding $R_m$ at $\mathbf e_m$ must increase the pointwise FDP, we have
$$\FDP_{N_m}(R_m)>\max_{T\in\mathcal D(\mathbf e_m)}\FDP_{N_m}(T).$$

Step 4: Define $\mathbf Z=(Z_1,\dots,Z_K)$ by $\mathbf Z=\mathbf e_m$ when $\mathbf X=\mathbf e_m$, and otherwise
$$
Z_i
=
\begin{cases}
X_i, & i\in N_m,\\
\infty, & i\notin N_m.
\end{cases}
$$
The random vector $\mathbf Z$ is also an e-value vector. We show that 
\begin{equation}\label{eq:FDR_Z}
\FDR_{\mathcal D^{\mathbf e_m,R_m}}(\mathbf Z)>\alpha.
\end{equation}
Since $\mathbf Z \ge \mathbf X$ and $\mathcal D$ is increasing, we have $\mathcal D(\mathbf X)\subseteq\mathcal D(\mathbf Z)$. This implies 
$$\max_{T\in \mathcal D^{\mathbf e_m,R_m}(\mathbf X)}\FDP_{N_m}(T) \le \max_{T\in \mathcal D^{\mathbf e_m,R_m}(\mathbf Z)}\FDP_{N_m}(T),$$
which gives $\FDR_{\mathcal D^{\mathbf e_m,R_m}}(\mathbf X)\le \FDR_{\mathcal D^{\mathbf e_m,R_m}}(\mathbf Z)$ and we get \eqref{eq:FDR_Z}.

Step 5: For any $\mathbf y\in[0,\infty]^K$ satisfying $y_j=\infty$ for all $j\notin N_m$, we show that
\begin{equation}\label{eq:FDP_increase}
\max_{T\in\mathcal G(\mathbf y)}\FDP_{N_m}(T)\ge\max_{T\in\mathcal D(\mathbf y)}\FDP_{N_m}(T).
\end{equation}
If $\mathcal D(\mathbf y)\subseteq\mathcal G(\mathbf y)$, there is nothing to prove. Since $\mathcal D\preceq_{\mathrm w}\mathcal G$, we only need to consider $\mathcal D(\mathbf y)^\downarrow\subsetneq \mathcal G(\mathbf y)^\downarrow$.
First choose a $U\in\mathcal G(\mathbf y)^\downarrow\setminus \mathcal D(\mathbf y)^\downarrow$, and then choose an $S\in\mathcal G(\mathbf y)$ such that $U\subseteq S$. Then we have $S\notin\mathcal D(\mathbf y)$. To prove \eqref{eq:FDP_increase}, it suffices to prove
$$\FDP_{N_m}(S)\ge\max_{T\in\mathcal D(\mathbf y)}\FDP_{N_m}(T).$$
Suppose otherwise for contradiction. Then by s-admissibility of $\mathcal D$, $\mathcal D^{\mathbf y,S}\notin\mathrm{SP}_\alpha$. Hence
$\mathcal D^{\mathbf y,S}$ fails to control FDR under some null set $N'$. Since $\mathcal D^{\mathbf y,S}$ differs from $\mathcal D$ only at $\mathbf y$, any input vector of e-values that makes $\mathcal D^{\mathbf y,S}$ violate FDR must have positive mass on $\mathbf y$. By construction, $y_j=\infty$ for every $j\notin N_m$, whereas null e-values are finite almost surely. Hence $N'\subseteq N_m$. Moreover, $\FDP_{N_m}(S)<\max_{T\in\mathcal D(\mathbf y)}\FDP_{N_m}(T)$ means that adding $S$ at $\mathbf y$ does not increase the FDP under null set $N_m$. Hence $\mathcal D^{\mathbf y,S}$ will not violate FDR control under $N_m$, and therefore $N'\subsetneq N_m$. This implies that $(\mathbf y,S,N')$ is another feasible triple with a strictly smaller null set than $(\mathbf e_m,R_m,N_m)$, contradicting that $N_m$ is the smallest.

Step 6: We now show $\FDR_{\mathcal G}(\mathbf Z)>\alpha$. If $\mathbf Z=\mathbf e_m$, then
$R_m\in\mathcal G(\mathbf e_m)$, and hence
$$\max_{T\in\mathcal G(\mathbf Z)}\FDP_{N_m}(T)\ge\FDP_{N_m}(R_m)=\max\left\{\max_{T\in\mathcal D(\mathbf e_m)}\FDP_{N_m}(T),\FDP_{N_m}(R_m)\right\}.$$
The right-hand side is exactly the FDP of $\mathcal D^{\mathbf e_m,R_m}$ at $\mathbf Z$. If $\mathbf Z\ne\mathbf e_m$, then all coordinates outside $N_m$ are equal to $\infty$. Step 5 then gives
$$\max_{T\in\mathcal G(\mathbf Z)}\FDP_{N_m}(T)\ge\max_{T\in\mathcal D(\mathbf Z)}\FDP_{N_m}(T).$$
Since $\mathcal D^{\mathbf e_m,R_m}(\mathbf Z)=\mathcal D(\mathbf Z)$ when $\mathbf Z\ne\mathbf e_m$, the right-hand side is exactly the FDP of $\mathcal D^{\mathbf e_m,R_m}$ at $\mathbf Z$. Combining the above two equations gives
$\FDR_{\mathcal G}(\mathbf Z)\ge\FDR_{\mathcal D^{\mathbf e_m,R_m}}(\mathbf Z)>\alpha$, contradicting $\mathcal G\in\mathrm{SP}_\alpha$. 

Hence $\mathcal D$ is admissible, so s-admissibility and admissibility are equivalent.
\end{proof}

\begin{proof}[Proof of Proposition~\ref{prop:admissibility-level-issue}]
Let $\mathcal D$ be an admissible procedure at level $\alpha$. Fix a $\beta\in(0,\alpha)$, we show that $\mathcal D\notin\mathrm{SP}_\beta$. Suppose, for contradiction, that $\mathcal D\in\mathrm{SP}_\beta$.

Step 1: Set $\delta=\alpha-\beta$, consider a special point $\mathbf e^*=(e_1,\dots,e_K)$ with $e_1=1/\delta$ and $e_i=0$ for all $i\neq 1$, we show $\{1\} \in \mathcal D(\mathbf e^*)$. Define
$$
\mathcal H_\delta(\mathbf e)
=
\begin{cases}
\{\varnothing,\{1\}\}, & e_1\ge1/\delta,\\
\{\varnothing\}, & e_1<1/\delta.
\end{cases}
$$
We claim that $\mathcal H_\delta\in\mathrm{SP}_\delta$. 

Fix a true null set $N$ and an input e-value $\mathbf E$. 
\begin{itemize}
\item[(i)] If $1\notin N$, then the FDP is always zero.

\item[(ii)] If $1\in N$, Markov's inequality gives
$\FDR_{\mathcal H_\delta}(\mathbf E)=\mathbb P(E_1\ge1/\delta)\le\delta$.
\end{itemize}
Thus $\mathcal H_\delta\in\mathrm{SP}_\delta$. Further, define 
$\mathcal G(\mathbf e)=\mathcal D(\mathbf e)\cup\mathcal H_\delta(\mathbf e)$. We have
$$
\begin{aligned}
\max_{R\in\mathcal G(\mathbf E)}\FDP_N(R)
&=
\max\left\{\max_{R\in\mathcal D(\mathbf E)}\FDP_N(R),\max_{R\in\mathcal H_\delta(\mathbf E)}\FDP_N(R)\right\} \\
&\le
\max_{R\in\mathcal D(\mathbf E)}\FDP_N(R)+\max_{R\in\mathcal H_\delta(\mathbf E)}\FDP_N(R).
\end{aligned}
$$
Taking expectations gives $\FDR_{\mathcal G}(\mathbf E)\le\beta+\delta=\alpha$.
Thus $\mathcal G\in\mathrm{SP}_\alpha$ and $\mathcal D\preceq_{\mathrm s}\mathcal G$. The admissibility of $\mathcal D$ at level $\alpha$ forces $\mathcal G=\mathcal D$, from which we have $\{1\}\in\mathcal H_\delta(\mathbf e^*)\subseteq\mathcal G(\mathbf e^*)=\mathcal D(\mathbf e^*)$.

Step 2: we show $\mathcal D\notin\mathrm{SP}_\beta$. Set $c=\alpha/\beta>1$ and define $\mathcal D_c(\mathbf e)=\mathcal D(c\mathbf e)$. We claim that $\mathcal D_c\in\mathrm{SP}_\alpha$. First note that $\mathcal D(\mathbf 0)$ cannot contain a nonempty
set. Otherwise, input the vector $\mathbf 0$ under the true null set $\cK$ would give FDR equal to one, contradicting
$\mathcal D\in\mathrm{SP}_\beta$. 

Fix a true null set $N$ and an input e-value $\mathbf E$. Define a random vector $\mathbf X=(X_1,\dots,X_K)$ by
$$
\mathbf X=
\begin{cases}
c\mathbf E, & \text{with probability } 1/c,\\
\mathbf 0, & \text{with probability } 1-1/c.
\end{cases}
$$
For every $i\in N$, $\mathbb E[X_i]=\mathbb E[cE_i] / c=\mathbb E[E_i]\le1$, so $\mathbf X$ is a valid input e-value. Since
$\mathcal D(\mathbf 0)$ contains no nonempty set,
$$
\beta
\ge
\mathbb E\left[
\max_{R\in\mathcal D(\mathbf X)}\FDP_N(R)
\right]
=
\frac1c
\mathbb E\left[
\max_{R\in\mathcal D_c(\mathbf E)}\FDP_N(R)
\right]
=
\frac1c\FDR_{\mathcal D_c}(\mathbf E).
$$
Therefore $\FDR_{\mathcal D_c}(\mathbf E)\le c\beta=\alpha$, so $\mathcal D_c\in\mathrm{SP}_\alpha$.

From Lemma~\ref{lem:s-admissible-increasing}, we know that $\mathcal D$ is increasing, which gives $\mathcal D(\mathbf e)\subseteq\mathcal D(c\mathbf e)=\mathcal D_c(\mathbf e)$ for every $\mathbf e$. Thus $\mathcal D\preceq_{\mathrm s}\mathcal D_c$. The admissibility of $\mathcal D$ at level $\alpha$ again forces $\mathcal D_c=\mathcal D$, and consequently
$\mathcal D(c\mathbf e)=\mathcal D(\mathbf e)$ for every $\mathbf e$. Applying this equality repeatedly to the point $\mathbf e^*$ gives $\{1\}\in\mathcal D(\mathbf e^*/c^n)$. Choose $n$ sufficiently large that $1 / \delta \le c^n$. Since every coordinate of $\mathbf e^*/c^n$ is at most one, we may take $\mathbf e^*/c^n$ as a valid input e-value. Under the true null set $\cK$, we have
$$\max_{R\in\mathcal D(\mathbf e^*/c^n)}\FDP_{\cK}(R)\ge\FDP_{\cK}(\{1\})=1.$$
This contradicts $\mathcal D\in\mathrm{SP}_\beta$.
\end{proof}

\subsection{Omitted proofs from Section~\ref{sec:close-eBH}}  
\label{app:A2}

\begin{proof}[Proofs of parts~(ii) and~(iii) of Proposition~\ref{prop:esc-dominated-by-weighted-ebh}]
(ii). \underline{The ``if'' direction} follows immediately from (i).

\underline{The ``only if'' direction}: Suppose, for contradiction, that $\alpha\lambda_0^A+K\lambda_i^A<1$ for some $A\subseteq\cK$ and $i\in A$. Take $\mathbf e=(e_1,\dots,e_K)$ with $e_i=K/\alpha$ and $e_j=0$ for $j\ne i$. Then $\mathrm{eSC}_\alpha(\mathbf e)=\{\varnothing,\{i\}\}$,
whereas
$$
\mathbb M_{\bm{\lambda}^A}(\mathbf e^A)
=\lambda_0^A+\frac K\alpha\lambda_i^A
<\frac1\alpha.
$$
Thus $\{i\}\notin\overline{\mathrm{eBH}}_\alpha^{\bm{\lambda}}(\mathbf e)$.
Hence there exists $R\in\overline{\mathrm{eBH}}_\alpha^{\bm{\lambda}}(\mathbf e)$
such that $R\supsetneq\{i\}$.

(a) Suppose that the procedure is constant-free. Choose $j\in R\setminus\{i\}$. Then 
$0=\mathbb M_{\bm{\lambda}^{\{j\}}}(0)\ge 1 / (\alpha|R|)$, which is impossible.

(b) Suppose that $K=2$. In this case, $R=\{1,2\}$ necessarily holds. If $A=\{1,2\}$, then $\alpha\lambda_0^A+K\lambda_i^A<1$ yields
$\mathbb M_{\bm{\lambda}^A}(\mathbf e^A)< |A\cap R| / (\alpha |R|)$, which contradicts $R\in\overline{\mathrm{eBH}}_\alpha^{\bm{\lambda}}(\mathbf e)$. If $A=\{i\}$, let $j$ be the other element of $\{1,2\}$ and take another point $\mathbf y=(y_1,y_2)$ with $y_i=2/\alpha$ and $y_j=1/\alpha$. Since $\{i\}$ and $\{1,2\}$ belong to $\mathrm{eSC}_\alpha(\mathbf y)$, we have
$\mathrm{eSC}_\alpha(\mathbf y)^\downarrow=2^{\{1,2\}}$. Weak dominance therefore forces
$\mathrm{eSC}_\alpha(\mathbf y)\subseteq\overline{\mathrm{eBH}}_\alpha^{\bm{\lambda}}(\mathbf y)$.
In particular, $\{i\}\in\overline{\mathrm{eBH}}_\alpha^{\bm{\lambda}}(\mathbf y)$, which gives
$$
\mathbb M_{\bm{\lambda}^A}(\mathbf y^A)
=\lambda_0^A+\frac2\alpha\lambda_i^A
\ge \frac1\alpha.
$$
This contradicts $\alpha\lambda_0^A+K\lambda_i^A<1$.

(c) Suppose that $K\ge3$ and $\alpha<K/(K+1)$. Write
$C=\cK\setminus\{i\}$. Since all coordinates indexed by $C$ are zero, $R\in\overline{\mathrm{eBH}}_\alpha^{\bm{\lambda}}(\mathbf e)$ gives
\begin{equation}\label{eq:weak-esc-lower-constant}
\lambda_0^C\ge\frac{|R|-1}{\alpha |R|}.
\end{equation}
Define $\mathbf y=(y_1,\dots,y_K)$ by
$$
y_r=
\begin{cases}
1 / \alpha, & r=i,\\
K / (\alpha(K-1)), & r\neq i.
\end{cases}
$$
Both $C$ and $\cK$ belong to $\mathrm{eSC}_\alpha(\mathbf y)$. Since
$\cK\in\mathrm{eSC}_\alpha(\mathbf y)$, we have
$\mathrm{eSC}_\alpha(\mathbf y)^\downarrow=2^{\cK}$, and weak dominance forces
$\mathrm{eSC}_\alpha(\mathbf y)\subseteq\overline{\mathrm{eBH}}_\alpha^{\bm{\lambda}}(\mathbf y)$.
In particular, $C\in
\overline{\mathrm{eBH}}_\alpha^{\bm{\lambda}}(\mathbf y)$, which gives
\begin{equation}\label{eq:weak-esc-upper-constant}
\lambda_0^C\le\frac1{K-\alpha(K-1)}.
\end{equation}
Combining \eqref{eq:weak-esc-lower-constant} and
\eqref{eq:weak-esc-upper-constant} yields
$$
\frac{|R|-1}{\alpha |R|}
\le\frac1{K-\alpha(K-1)},
$$
and hence
$$
\alpha\ge\frac{K(|R|-1)}{K(|R|-1)+1}\ge\frac{K}{K+1},
$$
a contradiction. Therefore the condition \eqref{eq:dom-eBH} necessarily holds.

(iii). \underline{The ``if'' direction} follows directly from (i).

\underline{The ``only if'' direction}: Suppose, for contradiction, that
$\alpha\lambda_0^A+K\lambda_i^A<1$ for some $A\subseteq\cK$ and $i\in A$. Let $\overline{\mathrm{ebh}}_\alpha^{\bm\lambda}$ be an induced point procedure that dominates $\mathrm{eBH}_\alpha$. Take $\mathbf e=(e_1,\dots,e_K)$ with $e_i=K/\alpha$ and $e_j=0$ for $j\ne i$. Then $\mathrm{eBH}_\alpha(\mathbf e)=\{i\}$, but $\{i\}\notin\overline{\mathrm{eBH}}_\alpha^{\bm\lambda}(\mathbf e)$.

Write $R=\overline{\mathrm{ebh}}_\alpha^{\bm\lambda}(\mathbf e)$. Since $\overline{\mathrm{ebh}}_\alpha^{\bm\lambda}$ dominates $\mathrm{eBH}_\alpha$, we have $\{i\}\subseteq R$. Moreover,
$R\in\overline{\mathrm{eBH}}_\alpha^{\bm\lambda}(\mathbf e)$, and hence $R\supsetneq\{i\}$. Set $C=R\setminus\{i\}$. Since all coordinates indexed by $C$ are zero, $R\in\overline{\mathrm{eBH}}_\alpha^{\bm\lambda}(\mathbf e)$ gives
$$
\lambda_0^C
=
\mathbb M_{\bm{\lambda}^C}(\mathbf e^C)
\ge
\frac{|R|-1}{\alpha |R|}
\ge
\frac1{2\alpha}.
$$
If $\overline{\mathrm{eBH}}_\alpha^{\bm\lambda}$ is constant-free, then $\lambda_0^C=0$, which is impossible. Otherwise, the assumption in part (iii) gives $\alpha<1/2$, and hence the last display implies $\lambda_0^C>1$, contradicting $\lambda_0^C\le1$. Therefore the condition \eqref{eq:dom-eBH} necessarily holds, which completes the proof.
\end{proof}

\subsection{Omitted proofs from  Section~\ref{sec:adm-closed}}  

\begin{proof}[Proof of Proposition~\ref{prop:point-mean-admissibility-threshold}]
\underline{The ``if'' direction} follows directly from Theorem~\ref{thm:positive-point-weighted-ebh-admissible}.

\underline{The ``only if'' direction}: We show $\overline{\mathrm{ebh}}_\alpha^{\mathsf{m}}$ is not admissible when $\alpha\ge1/2$.

\textbf{Part 1:} $K=2$. Define a point procedure $\mathfrak D$ by
$$
\mathfrak D(e_1,e_2)
=
\begin{cases}
\{1,2\}, & \text{if } e_1+e_2\ge \dfrac{2}{\alpha},\\[1ex]
\varnothing, & \text{otherwise.}
\end{cases}
$$
Note that a necessary condition for $\overline{\mathrm{ebh}}_\alpha^{\mathsf{m}}(\mathbf e)\neq\varnothing$ is $e_1+e_2\ge2/\alpha$, so $\mathfrak D$ dominates $\overline{\mathrm{ebh}}_\alpha^{\mathsf{m}}$. Moreover, take $\mathbf e=(2/\alpha,0)$. Then $\mathfrak D(\mathbf e)=\{1,2\}$, but $\{1,2\}\notin\overline{\mathrm{eBH}}_\alpha^{\mathsf{m}}(\mathbf e)$ because $A=\{2\}$ gives $\bar e_A=0<1 / (2\alpha)$. Hence $\mathfrak D$ strictly dominates $\overline{\mathrm{ebh}}_\alpha^{\mathsf{m}}$.

It suffices to show $\mathfrak D\in \mathrm{PP}_\alpha$. Fix a null set $N\subseteq \{1,2\}$.

\begin{itemize}
\item[(i)] If $N=\{1,2\}$, then
$$
\FDR_{\mathfrak D}(\mathbf E)
=
\mathbb P\left(E_1+E_2\ge\frac2\alpha\right)
\le
\frac\alpha2\mathbb E[E_1+E_2]\le\alpha.
$$

\item[(ii)] If $|N|=1$, then $\mathrm{FDP}_N(\mathfrak D(\mathbf E))\le1/2\le\alpha$.

\item[(iii)] The case $N=\varnothing$ is trivial.
\end{itemize}

\textbf{Part 2:} $K\ge3$. For $2\le n\le K$, let
$$
c_n
=
\max\left\{
\frac{n^2}{\alpha K},
\frac{n(n-1)}{\alpha(K-1)}
\right\}.
$$
Define a point procedure $\mathfrak D$ by
$$
\mathfrak D(\mathbf e)
=
\begin{cases}
\cK, & \text{ if } \sum_{i\in B}e_i\ge c_{|B|} \text{ for every } B\subseteq \cK \text{ with }|B|\ge2, \\[1ex]
\overline{\mathrm{ebh}}_\alpha^{\mathsf{m}}(\mathbf e),
& \text{otherwise.}
\end{cases}
$$
By construction, $\mathfrak D$ dominates $\overline{\mathrm{ebh}}_\alpha^{\mathsf{m}}$. Moreover, take
$$\mathbf e=\left(0,\frac{K}{\alpha(K-1)},\dots,\frac{K}{\alpha(K-1)}\right).$$
We claim that the promotion condition holds, that is, $\sum_{i\in B}e_i\ge c_{|B|}$ for every $B\subseteq \cK$ with $|B|\ge 2$. 

\begin{itemize}
\item[(i)] If $1\notin B$, then
$$
\sum_{i\in B}e_i
=
\frac{|B|K}{\alpha(K-1)}
\ge
\max\left\{
\frac{|B|^2}{\alpha K},
\frac{|B|(|B|-1)}{\alpha(K-1)}
\right\}.
$$

\item[(ii)] If $1\in B$, then $\sum_{i\in B}e_i=K(|B|-1)/(\alpha(K-1))$, and
$$\frac{K(|B|-1)}{\alpha(K-1)}\ge\frac{|B|(|B|-1)}{\alpha(K-1)}.$$
Since
$$\frac{K(|B|-1)}{\alpha(K-1)}\ge\frac{|B|^2}{\alpha K}\iff K^2(|B|-1)\ge |B|^2(K-1)$$
and $K^2(|B|-1)-|B|^2(K-1)=(K-|B|)(K|B|-K-|B|)\ge0$, we have
$$\frac{K(|B|-1)}{\alpha(K-1)}\ge \frac{|B|^2}{\alpha K}.$$
\end{itemize}

Thus $\mathfrak D(\mathbf e)=\cK$. However, $\cK\notin\overline{\mathrm{eBH}}_\alpha^{\mathsf{m}}(\mathbf e)$, because taking $A=\{1\}$ gives $\bar e_A=0<1/(\alpha K)$. Thus $\mathfrak D$ strictly dominates $\overline{\mathrm{ebh}}_\alpha^{\mathsf{m}}$. 

It suffices to show $\mathfrak D\in \mathrm{PP}_\alpha$. Fix a null set $N\subseteq\cK$ and write $\mathbf E=(E_1,\dots,E_K)$.

\begin{itemize}
\item[(i)] If $|N|\ge2$, we claim that $\mathrm{FDP}_N(\mathfrak D(\mathbf E))\le\alpha \bar E_N$. Indeed, if the promotion condition does not hold, then $\mathfrak D(\mathbf E)=\overline{\mathrm{ebh}}_\alpha^{\mathsf{m}}(\mathbf E)$, and the bound directly follows. If the promotion condition holds, then $\mathfrak D(\mathbf E)=\cK$, and applying the promotion condition to $B=N$ gives $|N|/K\le\alpha \bar E_N$. Taking expectations gives FDR control.

\item[(ii)] If $|N|=1$, write $N=\{i\}$. We first claim that
\begin{equation}
\label{eq:point-mean-counter-singleton-bound}
\frac{|\{i\}\cap\mathfrak D(\mathbf E)|}{|\mathfrak D(\mathbf E)|\vee1}
\le
\frac13\,\mathds 1_{\{E_i<1/(2\alpha)\}}
+
\frac12\,\mathds 1_{\{1/(2\alpha)\le E_i\le 2/\alpha\}}
+
\mathds 1_{\{E_i>2/\alpha\}} .
\end{equation}

(a) Suppose $E_i<1/(2\alpha)$. If the promotion condition holds, then $\mathfrak D(\mathbf E)=\cK$, so the left-hand side is at most $1/3$. If the promotion condition does not hold, then $\mathfrak D(\mathbf E)=\overline{\mathrm{ebh}}_\alpha^{\mathsf{m}}(\mathbf E)$. No set in $\overline{\mathrm{eBH}}_\alpha^{\mathsf{m}}(\mathbf E)$ containing $i$ can have size $1$ or $2$, because taking $A=\{i\}$ would give $E_i\ge1/(\alpha |R|)\ge1/(2\alpha)$. Thus the left-hand side is at most $1/3$.

(b) Suppose $1/(2\alpha)\le E_i\le2/\alpha$. We show that $\mathfrak D(\mathbf E)\ne\{i\}$, which implies the left-hand side is at most $1/2$. This is clear if the promotion condition holds. We only need to consider the case $\mathfrak D(\mathbf E)=\overline{\mathrm{ebh}}_\alpha^{\mathsf{m}}(\mathbf E)$. Suppose for contradiction that $\overline{\mathrm{ebh}}_\alpha^{\mathsf{m}}(\mathbf E)=\{i\}$. Then taking $A=\cK$ gives $\sum_{k=1}^K E_k\ge K/\alpha$. Using $E_i\le2/\alpha$, we get $\sum_{k\ne i}E_k\ge(K-2)/\alpha$. If some $j\ne i$ has $E_j=\infty$, then it is easy to see $\{i,j\}\in\overline{\mathrm{eBH}}_\alpha^{\mathsf{m}}(\mathbf E)$, which implies that $\overline{\mathrm{ebh}}_\alpha^{\mathsf{m}}(\mathbf E)=\{i\}$ is impossible. Therefore, in the remaining case $E_j<\infty$ for all $j\ne i$. Hence there exists $j\ne i$ such that
$$E_j\ge\frac{K-2}{\alpha(K-1)}\ge\frac1{2\alpha}.$$
We show that $\{i,j\}\in\overline{\mathrm{eBH}}_\alpha^{\mathsf{m}}(\mathbf E)$. Let $A\subseteq \cK$ be nonempty. If $A\cap\{i,j\}=\varnothing$, there is nothing to prove. If $i\in A$, then $\{i\} \in \overline{\mathrm{eBH}}_\alpha^{\mathsf{m}}(\mathbf E)$ gives $\bar E_A \ge1/\alpha\ge |A\cap\{i,j\}|/(2\alpha)$. It remains to consider $j\in A$ and $i\notin A$. If $A=\{j\}$, then $E_j\ge1/(2\alpha)$. If $|A|\ge2$, then $\{i\} \in \overline{\mathrm{eBH}}_\alpha^{\mathsf{m}}(\mathbf E)$ applied to $A\cup\{i\}$ gives $E_i+\sum_{k\in A}E_k\ge(|A|+1)/\alpha$. Since $E_i\le2/\alpha$, we get
$$\bar E_A\ge\frac{|A|-1}{\alpha |A|}\ge\frac1{2\alpha}=\frac{|A\cap\{i,j\}|}{2\alpha}.$$
Thus $\{i,j\}\in\overline{\mathrm{eBH}}_\alpha^{\mathsf{m}}(\mathbf E)$, contradicting $\overline{\mathrm{ebh}}_\alpha^{\mathsf{m}}(\mathbf E)=\{i\}$.

(c) Suppose $E_i>2/\alpha$. Then the left-hand side is at most $1$.

Finally, write
$$p_0=\mathbb P\!\left(E_i<\frac1{2\alpha}\right),\ 
p_1=\mathbb P\!\left(\frac1{2\alpha}\le E_i\le\frac2\alpha\right),\text{ and }
p_2=\mathbb P\!\left(E_i>\frac2\alpha\right).
$$
Since $\mathbb E[E_i]\le1$, we have $p_1 / (2\alpha)+ 2p_2 / \alpha \le 1$. Taking expectations in \eqref{eq:point-mean-counter-singleton-bound}, and using $\alpha\ge1/2$, we obtain
$$
\mathbb E\left[
\frac{|\{i\}\cap\mathfrak D(\mathbf E)|}{|\mathfrak D(\mathbf E)|\vee1}
\right]
\le
\frac13p_0+\frac12p_1+p_2
=
\frac13+\frac16p_1+\frac23p_2
\le
\frac13+\frac{\alpha}{3}
\le
\alpha.
$$

\item[(iii)] The case $N=\varnothing$ is trivial.
\end{itemize}
\end{proof}

\subsection{Omitted proofs from  Section \ref{sec:complete-class}}

\begin{proof}[Proof of Theorem~\ref{thm:point-weighted-ebh-complete-class}]
Let $\Theta=\prod_{A\subseteq\cK}\Delta_{|A|+1}$. Each $\bm \theta\in\Theta$ specifies a weighted-mean e-collection; let $\overline{\mathrm{eBH}}_\alpha^{\bm \theta}$ denote the corresponding weighted-mean $\overline{\mathrm{eBH}}$ procedure. Define
$m_{\bm \theta}(\mathbf e)=\max\{|R|:\ \mathfrak D(\mathbf e)\subseteq R,\ R\in\overline{\mathrm{eBH}}_\alpha^{\bm \theta}(\mathbf e)\}$ with the convention $m_{\bm \theta}(\mathbf e)=-\infty$ if there is no $R\in\overline{\mathrm{eBH}}_\alpha^{\bm \theta}(\mathbf e)$ such that $\mathfrak D(\mathbf e)\subseteq R$.

Viewing $\mathfrak D$ as a singleton simultaneous procedure, Theorem~\ref{thm:admissible-weighted-ebh-dominator} gives some
$\bm \eta\in\Theta$ such that $\mathfrak D(\mathbf e)\in\overline{\mathrm{eBH}}_\alpha^{\bm \eta}(\mathbf e)$ for every $\mathbf e$. Hence $m_{\bm \eta}(\mathbf e)\ge|\mathfrak D(\mathbf e)|$ for every $\mathbf e$, and $m_{\bm \eta}(\mathbf e)$ never takes the value $-\infty$. Define
$$\mathfrak M=\left\{m_{\bm \theta}: \bm \theta\in\Theta,\ m_{\bm \eta}(\mathbf e)\le m_{\bm \theta}(\mathbf e)\text{ for all }\mathbf e\right\},$$
equipped with the order $m\le m'$ whenever $m(\mathbf e)\le m'(\mathbf e)$ for all $\mathbf e$. Note that $m_{\bm \eta}\in\mathfrak M$, so $\mathfrak M$ is nonempty. By a Zorn argument similar to Step 2 in the proof of Theorem~\ref{thm:admissible-weighted-ebh-dominator}, we can show that $\mathfrak M$ has a maximal element; denote it by
$m^*=m_{\bm \theta^*}$. Let $\overline{\mathrm{eBH}}_\alpha^{\bm \theta^*}$ be the $\overline{\mathrm{eBH}}$ procedure based on the weighted-mean e-collection $\{E_A^{\bm \theta^*}\}_{A\subseteq\cK}$ specified by $\bm \theta^*$.

Next, we construct an admissible point weighted-mean $\overline{\mathrm{eBH}}$ procedure that dominates $\mathfrak D$.
Fix a total order $\prec$ on $2^{\cK}$.\footnote{This total order can be chosen arbitrarily; it only needs to be independent of
$\mathbf e$.} For each $R\subseteq\cK$, define
$$A_R=\{\mathbf e:\ \mathfrak D(\mathbf e)\subseteq R\},$$
and
$$B_R=\bigcap_{A\subseteq\cK}\left\{\mathbf e:E_A^{\bm \theta^*}(\mathbf e)\ge\frac{|A\cap R|}{\alpha(|R|\vee1)}\right\}.$$
For each $\mathbf e$, define $\mathfrak D^*(\mathbf e)$ to be the $\prec$-smallest set $R\subseteq\cK$ such that $\mathbf e\in A_R\cap B_R$ and $|R|=m^*(\mathbf e)$. Such a set exists by the fact that $m^*\ge m_{\bm\eta}\ge|\mathfrak D|\ge0$. In the following, we verify that this $\mathfrak D^*$ is an admissible point weighted-mean $\overline{\mathrm{eBH}}$ procedure that dominates $\mathfrak D$.

\begin{itemize}
\item[(i)] We show that $\mathfrak D^*$ is measurable. It suffices to
show that, for each $R\subseteq\cK$, the set $\{\mathbf e:\ \mathfrak D^*(\mathbf e)=R\}$ is measurable. For each $t\in\{0,1,\dots,K\}$, by definition, we have 
$$\{\mathbf e: m^*(\mathbf e)\ge t\}=\bigcup_{|R|\ge t}(A_R\cap B_R).$$
Hence $m^*$ is measurable. Moreover,
$$\{\mathbf e: \mathfrak D^*(\mathbf e)=R\}=A_R\cap B_R\cap\{\mathbf e: m^*(\mathbf e)=|R|\} \cap \left(\bigcap_{R'\prec R, |R'|=|R|}(A_{R'}\cap B_{R'})^c\right).$$
Since $\mathfrak D$ is measurable, the set $A_R$ is measurable. Moreover, $B_R$ is measurable by the $\infty$ convention.
Thus all sets on the right-hand side are measurable. Therefore $\mathfrak D^*$ is measurable. Since
$\mathfrak D^*(\mathbf e)\in\overline{\mathrm{eBH}}_\alpha^{\bm \theta^*}(\mathbf e)$ for every $\mathbf e$, we have
$\mathfrak D^*\in\mathrm{PP}_\alpha$.

\item[(ii)] By construction, $\mathfrak D(\mathbf e)\subseteq\mathfrak D^*(\mathbf e)$ for every $\mathbf e$.

\item[(iii)] We show that $\mathfrak D^*$ rejects a maximal set in $\overline{\mathrm{eBH}}_\alpha^{\bm \theta^*}$. Suppose otherwise for contradiction. Then we can find $\mathbf e$ and $S\subseteq\cK$ such that $\mathfrak D^*(\mathbf e)\subsetneq S$ and $S\in\overline{\mathrm{eBH}}_\alpha^{\bm \theta^*}(\mathbf e)$. Then $\mathfrak D(\mathbf e)\subseteq S$ and
$|S|>|\mathfrak D^*(\mathbf e)|=m^*(\mathbf e)$, contradicting the definition of $m^*(\mathbf e)$. Thus $\mathfrak D^*$ is a point weighted-mean $\overline{\mathrm{eBH}}$ procedure.

\item[(iv)] We show the admissibility of $\mathfrak D^*$. Suppose otherwise for contradiction. Then there exists $\widehat{\mathfrak D}\in\mathrm{PP}_\alpha$ such that $\mathfrak D^*(\mathbf e)\subseteq\widehat{\mathfrak D}(\mathbf e)$ for all $\mathbf e$, and the set inclusion is strict for some $\mathbf e_1$. Similarly, Theorem~\ref{thm:admissible-weighted-ebh-dominator} gives some
$\widehat {\bm \theta}\in\Theta$ such that $\widehat{\mathfrak D}(\mathbf e)\in\overline{\mathrm{eBH}}_\alpha^{\widehat {\bm \theta}}(\mathbf e)$ for all $\mathbf e$. Since $\mathfrak D\subseteq\mathfrak D^*\subseteq \widehat{\mathfrak D}$, for all $\mathbf e$ it follows that
$$m_{\widehat {\bm \theta}}(\mathbf e)\ge|\widehat{\mathfrak D}(\mathbf e)|\ge|\mathfrak D^*(\mathbf e)|=m^*(\mathbf e)\ge m_{\bm\eta}(\mathbf e).$$
Thus $m_{\widehat {\bm \theta}}\in\mathfrak M$. At the point $\mathbf e_1$, we have
$$m_{\widehat {\bm \theta}}(\mathbf e_1)\ge|\widehat{\mathfrak D}(\mathbf e_1)|>|\mathfrak D^*(\mathbf e_1)|=m^*(\mathbf e_1),$$
which contradicts the maximality of $m^*$. Thus $\mathfrak D^*$ is admissible.
\end{itemize}

Therefore, there exists an admissible point weighted-mean $\overline{\mathrm{eBH}}$ procedure $\mathfrak D^*\in\mathrm{PP}_\alpha$ that dominates $\mathfrak D$. In particular, if $\mathfrak D$ is admissible, then $\mathfrak D\subseteq\mathfrak D^*$ implies $\mathfrak D=\mathfrak D^*$. Hence every admissible element of $\mathrm{PP}_\alpha$ admits a point weighted-mean $\overline{\mathrm{eBH}}$ representation. This completes the proof.
\end{proof}

\begin{proof}[Proof of Proposition~\ref{prop:closed-by-not-admissible}]

By construction, we know that $\mathfrak D$ dominates $\overline{\mathrm{BY}}_\alpha$. Moreover, $\overline{\mathrm{BY}}_\alpha(\mathbf p)\subsetneq\mathfrak D(\mathbf p)$ on $\mathcal A$, because if there exists $p_i>\alpha$, then $\cK\notin\mathcal C_\alpha^{\mathrm{BY}}(\mathbf p)$. Thus, $\mathfrak D$ strictly dominates $\overline{\mathrm{BY}}_\alpha$.

It suffices to show $\mathfrak D$ controls FDR. Define 
$$
E'_{\{i\}}(\mathbf p)
=
\frac{1}{\alpha}\mathds 1_{\{p_i\le \alpha/3\}}
+
\frac{1}{2\alpha}\mathds 1_{\{\alpha/3<p_i\le t_0\}},
$$
and $E'_A=E_A^{\mathrm{BY}}$ for all $|A|\ge2$. For a singleton $\{i\}$ and a null p-value $P_i$, taking expectations gives
$$
\mathbb E[E'_{\{i\}}(\mathbf P)]
\le
\frac{1}{2\alpha}\mathbb P(P_i\le \alpha/3)
+
\frac{1}{2\alpha}\mathbb P(P_i\le t_0)
\le
\frac16+\frac{t_0}{2\alpha}
\le1.
$$
Thus $(E'_A)_{A\subseteq\cK}$ is an e-collection. Let $\mathcal C'_\alpha$ be its corresponding candidate discovery sets. It suffices to show $\mathfrak D(\mathbf p)\in\mathcal C'_\alpha(\mathbf p)$ for every $\mathbf p$.

\textbf{Case 1:} $\mathbf p\notin\mathcal A$. Write $R=\overline{\mathrm{BY}}_\alpha(\mathbf p)$.

\begin{itemize}
\item[(i)] If $R=\varnothing$, there is nothing to prove.

\item[(ii)] Assume $|R|\ge2$. Let $A\subseteq\cK$ be nonempty.

(a) If $|A|\ge2$, then $E'_A=E_A^{\mathrm{BY}}$, so the required inequality directly follows.

(b) If $A=\{i\}$ and $i\notin R$, there is nothing to prove.

(c) If $A=\{i\}$ and $i\in R$, then necessarily $p_i\le\alpha$. Since $|R|\ge2$, we have
$E'_A(\mathbf p)\ge 1 / (2\alpha)\ge 1 / (\alpha |R|)$.

\item[(iii)] It remains to consider $R=\{i\}$. We claim that $p_i\le\alpha/3$. Suppose otherwise for contradiction. Fix $j\ne i$. We show that $\{i,j\}\in\mathcal C_\alpha^{\mathrm{BY}}(\mathbf p)$, contradicting that $R$ is largest. For $k\in\cK$ and $u\in[0,1]$, define
$$b_k(u)=\frac{\mathds 1_{\{h_k u\le \alpha\}}}{\lceil k h_k u/\alpha\rceil\vee 1}.$$
If $A\subseteq\cK$ is nonempty, then $E_A^{\mathrm{BY}}(\mathbf p)=\sum_{\ell\in A}b_{|A|}(p_\ell) / \alpha$ by definition.

(a) If $A\cap\{i,j\}=\varnothing$, there is nothing to prove. 

(b) If $i\in A$, from $\{i\}\in \mathcal C_\alpha^{\mathrm{BY}}(\mathbf p)$,
$$E_A^{\mathrm{BY}}(\mathbf p)\ge \frac1\alpha\ge \frac{|A\cap\{i,j\}|}{2\alpha}.$$

(c) If $j\in A$ and $i\notin A$. Set $T=A\cup\{i\}$. Also from $\{i\}\in \mathcal C_\alpha^{\mathrm{BY}}(\mathbf p)$,
$$\sum_{\ell\in A} b_{|T|}(p_\ell)+b_{|T|}(p_i)\ge1.$$
Since $|T|\ge2$ and $p_i>\alpha/3$, we have $|T|h_{|T|}p_i/\alpha>1$, and hence $b_{|T|}(p_i)\le1/2$. Therefore $\sum_{\ell\in A} b_{|T|}(p_\ell)\ge 1/2$. Since $b_k(u)$ is decreasing in $k$,
$$
E_A^{\mathrm{BY}}(\mathbf p)
=
\frac1\alpha\sum_{\ell\in A} b_{|A|}(p_\ell)
\ge
\frac1\alpha\sum_{\ell\in A} b_{|T|}(p_\ell)
\ge
\frac{1}{2\alpha}
=
\frac{|A\cap\{i,j\}|}{2\alpha}.
$$
Therefore $p_i\le\alpha/3$, and hence $E'_{\{i\}}(\mathbf p)=1/\alpha$. Thus $R\in\mathcal C'_\alpha(\mathbf p)$.
\end{itemize}

\textbf{Case 2:} $\mathbf p\in\mathcal A$. Then $\mathfrak D(\mathbf p)=\cK$. We show that $\cK\in\mathcal C'_\alpha(\mathbf p)$. Let $A\subseteq\cK$ be nonempty.

\begin{itemize}
\item[(i)] If $A=\{j\}$, then $p_j\le t_0$, so $E'_A(\mathbf p)\ge 1 / (2\alpha) \ge 1 / (K\alpha)$.

\item[(ii)] If $|A|\ge2$, then $E'_A=E_A^{\mathrm{BY}}$. All p-values except $p_{(K)}$ are at most $\alpha/(Kh_K)$, and each such p-value contributes $1/\alpha$ to $E_A^{\mathrm{BY}}(\mathbf p)$. Hence
$$E'_A(\mathbf p)=E_A^{\mathrm{BY}}(\mathbf p)\ge\frac{|A|-1}{\alpha}\ge\frac{|A|}{\alpha K}.$$
\end{itemize}

Therefore $\mathfrak D$ controls FDR and strictly dominates $\overline{\mathrm{BY}}_\alpha$.
\end{proof}

\subsection{Omitted proofs from  Section~\ref{sec:symmetric-ebh}}

\begin{proof}[Proof of Proposition~\ref{prop:symmetric-e-collection}]
\underline{The ``if'' direction} is immediate. 

\underline{The ``only if'' direction}: Suppose $\mathcal D$ is symmetric based on the e-collection $\mathbb M_{\bm{\lambda}^A}(\mathbf x^A)=\lambda_0^A+\sum_{i\in A}\lambda_i^A x_i$ for all finite $\mathbf x^A=(x_i)_{i\in A}$ and all $A\subseteq \cK$. We show that $\mathcal D$ can also be induced by a symmetric e-collection with constant terms $(\lambda_0^1,\dots,\lambda_0^K)\in [0,1]^K$ by induction on $|A|$. Suppose that, for some
$n\in\cK$, the e-collection $\mathbb M_{\bm\lambda^B}=\mathbb M_{\lambda_0^{|B|}}^{\mathrm{sym}}$ for every $B\subseteq\cK$ with $|B|<n$. We show that all weighted-mean e-collections indexed by the sets of size $n$ can also be replaced by symmetric ones without changing $\mathcal D$.

Fix $A\subseteq\cK$ with $|A|=n$. For every finite
$\mathbf x^A=(x_i)_{i\in A}$, define
$$E_A(\mathbf x^A)=\frac{1}{K!}\sum_{\sigma}\mathbb M_{\bm\lambda^{\sigma^{-1}(A)}}\left((x_{\sigma(i)})_{i\in\sigma^{-1}(A)}\right),$$
where the sum is over all permutations $\sigma$ of $\cK$. We set
$E_A(\mathbf x^A)=\infty$ whenever $\max(\mathbf x^A)=\infty$. Note that each $B\subseteq\cK$ with $|B|=n$ occurs equally often as $\sigma^{-1}(A)$, so the constant term of $E_A$ is
$\lambda_0^n:=\sum_{\substack{B\subseteq\cK, |B|=n}}\lambda_0^B / \binom{K}{n}$.
Moreover, the permutation average is invariant under every permutation
of the elements of $A$, and hence all coordinate coefficients of
$E_A$ are equal. Therefore, we obtain
$$E_A(\mathbf x^A)=\lambda_0^n+(1-\lambda_0^n)\bar x_A=\mathbb M_{\lambda_0^n}^{\mathrm{sym}}(\mathbf x^A).$$

We claim that simultaneously replacing $\mathbb M_{\bm\lambda^A}$ by $E_A$ for every $A$ with $|A|=n$ does
not change $\mathcal D$. Fix $\mathbf e\in[0,\infty]^K$ and $R\subseteq\cK$. If $\mathbb M_{\bm\lambda^B}(\mathbf e^B)< \FDP_B(R) / \alpha$ for some $B\subseteq\cK$ with $|B|<n$, then $R$ is excluded
both before and after the replacement. It remains to consider the case in
which $\mathbb M_{\bm\lambda^B}(\mathbf e^B)\ge \FDP_B(R) / \alpha$ for every $B\subseteq\cK$ with $|B|<n$.
If $e_i=\infty$ for some $i\in A$, then $\mathbb M_{\bm\lambda^A}(\mathbf e^A)=E_A(\mathbf e^A)=\infty$. Therefore, suppose that $\mathbf e^A$ is finite. Define $\mathbf y=(y_1,\dots,y_K)$ by
$$
y_i
=
\begin{cases}
e_i, & i\in A,\\
\infty, & i\notin A.
\end{cases}
$$
\begin{itemize}
\item[(i)] For every nonempty $B\subsetneq A$, $\mathbb M_{\bm\lambda^B}(\mathbf y^B)=\mathbb M_{\bm\lambda^B}(\mathbf e^B)\ge \FDP_B(R) / \alpha$. For every $B\setminus A\neq \varnothing$, $\mathbb M_{\bm\lambda^B}(\mathbf y^B)=\infty$. Thus, we have
\begin{equation}\label{eq:equivalence_1}
R\in\mathcal D(\mathbf y) \ \Longleftrightarrow \ \mathbb M_{\bm\lambda^A}(\mathbf e^A)\ge \frac{\FDP_A(R)}{\alpha}.
\end{equation}

\item[(ii)] Fix a permutation $\sigma$ of $\cK$. For every $B\subsetneq\sigma^{-1}(A)$, since $\sigma(B)\subsetneq A$, $\mathbb M_{\bm\lambda^B}(\mathbf y_\sigma^B)=\mathbb M_{\bm\lambda^{\sigma(B)}}(\mathbf y^{\sigma(B)}) \ge \FDP_{\sigma(B)}(R) / \alpha=\FDP_B(\sigma^{-1}(R)) / \alpha$.
For every $B\setminus \sigma^{-1}(A)\neq \varnothing$, $\mathbb M_{\bm\lambda^B}(\mathbf y_\sigma^B)=\infty$. By $\FDP_{\sigma^{-1}(A)}(\sigma^{-1}(R))=\FDP_A(R)$, we know
\begin{equation}\label{eq:equivalence_2}
\sigma^{-1}(R)\in\mathcal D(\mathbf y_\sigma) \ \Longleftrightarrow \ \mathbb M_{\bm\lambda^{\sigma^{-1}(A)}}\left((e_{\sigma(i)})_{i\in\sigma^{-1}(A)}\right)\ge \frac{\FDP_A(R)}{\alpha}.
\end{equation}
\end{itemize}

Since $\mathcal D$ is symmetric, $\mathcal D(\mathbf y_\sigma)=\sigma^{-1}(\mathcal D(\mathbf y))$.
Then for every permutation $\sigma$, the equivalences \eqref{eq:equivalence_1} and \eqref{eq:equivalence_2} give
$$\mathbb M_{\bm\lambda^A}(\mathbf e^A)\ge\frac{\FDP_A(R)}{\alpha}
\ \Longleftrightarrow\
\mathbb M_{\bm\lambda^{\sigma^{-1}(A)}}\left((e_{\sigma(i)})_{i\in\sigma^{-1}(A)}\right)\ge \frac{\FDP_A(R)}{\alpha}.$$
Therefore, we conclude that
$$\mathbb M_{\bm\lambda^A}(\mathbf e^A)\ge \frac{\FDP_A(R)}{\alpha}
\ \Longleftrightarrow\
E_A(\mathbf e^A)\ge \frac{\FDP_A(R)}{\alpha}.$$
This proves that simultaneously replacing $\mathbb M_{\bm\lambda^A}$ by $E_A$ for all $A$ with
$|A|=n$ leaves $\mathcal D$ unchanged. Finally, starting with $n=1$ and continuing to $n=K$ yields the desired symmetric weighted-mean e-collection representation of $\mathcal D$.
\end{proof}

\begin{proof}[Proof of Proposition~\ref{prop:symmetric-weighted-ebh-representation}]
Recall Step 1 of Theorem~\ref{thm:admissible-weighted-ebh-dominator}. When $\mathcal D\in\mathrm{SP}_\alpha^{\mathrm{sym}}$, we have $f_{\sigma^{-1}(A)}(\mathbf x_\sigma)=f_A(\mathbf x)$ and $F_{\sigma^{-1}(A)}(\mathbf x_\sigma)=F_A(\mathbf x)$ for each $A\subseteq \cK$ and every permutation $\sigma$ of $\cK$. By Theorem~\ref{thm:admissible-weighted-ebh-dominator}, there exists $(\lambda_0,(\lambda_i)_{i\in A})\in\Delta_{|A|+1}$ such that $F_{A}((x_i)_{i\in A})\le\lambda_0+\sum_{i\in A}\lambda_i x_i$ for all finite $\mathbf x$. Then averaging this inequality over all permutations of $A$ gives $F_{A}(\mathbf x)\le\lambda_0+(1-\lambda_0)\bar x_A$ for all finite $\mathbf x$. Here the constant term $\lambda_0$ can be chosen to depend only on $|A|$: carry out the construction for any fixed set with size $|A|$, and denote the resulting constant by $\lambda_0$; relabeling the coordinates then gives the same symmetric bound for every set of size $|A|$.
Hence, if $e_i<\infty$ for all $i\in A$, we have
$$f_A(\mathbf e)\le\lambda_0^{|A|}+(1-\lambda_0^{|A|})\bar e_A.$$
If $e_i=\infty$ for some $i\in A$, there is nothing to prove. Thus, $\mathcal D$ is strongly dominated by a symmetric weighted-mean $\overline{\mathrm{eBH}}$ procedure.
\end{proof}

 The proof of Proposition~\ref{prop:symmetric-point-representation} follows from Proposition~\ref{prop:symmetric-weighted-ebh-representation}, and it is omitted. 

\begin{proof}[Proof of Proposition~\ref{prop:symmetric-point-largest}]
(i). We first show that $\mathfrak D_\alpha$ controls FDR. If all hypotheses are null, then by Markov's inequality,
$$\mathbb P\left(\sum_{i=1}^K E_i\ge K/\alpha\right)\le\frac{\alpha}{K}\sum_{i=1}^K\mathbb E[E_i]\le\alpha.$$
If the true null set $N$ has size $m<K$, then on the event that $\mathfrak D_\alpha$ rejects, it rejects all, and hence
\[
\mathrm{FDP}_N(\mathcal K)\le\frac{K-1}{K}\le\alpha.
\]
The case $N=\varnothing$ is trivial. Thus $\mathfrak D_\alpha\in\mathrm{PP}_\alpha^{\mathrm{sym}}$.

Next, we show that $\mathfrak D_\alpha$ dominates every symmetric point procedure. Let $\mathfrak D\in\mathrm{PP}_\alpha^{\mathrm{sym}}$. We claim that $\mathfrak D(\mathbf e)\ne\varnothing$ implies $\sum_{i=1}^K e_i\ge K/\alpha$. Suppose otherwise for contradiction. Choose $p\in(\alpha,1]$ such that $p\sum_{i=1}^K e_i/K\le1$. Set $\cK$ as the true null set. Let $\Sigma$ be uniformly distributed over all permutations of $\mathcal K$, and define 
$$
\mathbf X
=
\begin{cases}
\mathbf e_\Sigma, & \text{with probability }p,\\
\mathbf 0, & \text{with probability }1-p.
\end{cases}
$$
For every $i\in \cK$, $\E [X_i]=p\sum_{j=1}^K e_j/K\le1$. By symmetry, $\mathfrak D(\mathbf e_\Sigma)=\Sigma^{-1}(\mathfrak D(\mathbf e))$. Therefore $\mathrm{FDR}_{\mathfrak D}(\mathbf X)\ge p>\alpha$. Thus every $\mathfrak D\in\mathrm{PP}_\alpha^{\mathrm{sym}}$ can reject only inside the region $\sum_i e_i\ge K/\alpha$. On this region $\mathfrak D_\alpha$ rejects $\mathcal K$, so $\mathfrak D_\alpha$ is largest.

(ii) For $K=2$, we have $1/K=(K-1)/K=1/2$.

\underline{The ``if'' direction}: Let $\mathfrak D\in \mathrm{PP}_\alpha^{\mathrm{sym}}$. Since $\alpha<1/K$, Theorem~\ref{thm:mean-ebh-symmetric-largest} gives $\mathfrak D(\mathbf e)\in \overline{\mathrm{eBH}}_\alpha^{\mathsf{m}}(\mathbf e)$ for all $\mathbf e$. Note that $\overline{\mathrm{eBH}}_\alpha^{\mathsf{m}}(\mathbf e)$ always contains a unique maximal set when $K=2$: if both $\{1\}$ and $\{2\}$ belong to $\overline{\mathrm{eBH}}_\alpha^{\mathsf{m}}(\mathbf e)$, then $\{1,2\}\in \overline{\mathrm{eBH}}_\alpha^{\mathsf{m}}(\mathbf e)$. Therefore $\mathfrak D \subseteq \overline{\mathrm{ebh}}_\alpha^{\mathsf{m}}$, and hence $\overline{\mathrm{ebh}}_\alpha^{\mathsf{m}}$ is the largest element. 

\underline{The ``only if'' direction}: Since $\alpha\ge (K-1)/K$, $\mathfrak D_\alpha$ introduced in (i) strictly dominates $\overline{\mathrm{ebh}}_\alpha^{\mathsf{m}}$.

(iii). Choose distinct numbers $y_1,\dots,y_{K-1}$ satisfying $(K-2) / (\alpha(K-1))<y_i<(K-1) / (\alpha K)$ for every $i=1,\dots,K-1$. Set $\mathbf y=(y_1,\dots,y_{K-1},\infty)$. Let $R_1=\mathcal K\setminus\{1\}$ and $R_2=\mathcal K\setminus\{2\}$. We claim that both $R_1$ and $R_2$ belong to $\overline{\mathrm{eBH}}_\alpha^{\mathsf{m}}(\mathbf y)$. Indeed, we only need to check $A$ with $K\notin A$. For $j=1,2$, we have $|A\cap R_j| / |R_j|\le (K-2) / (K-1)$, and hence
$$\frac1{|A|}\sum_{i\in A}y_i>\frac{K-2}{\alpha(K-1)}\ge\frac{|A\cap R_j|}{\alpha |R_j|}.$$
Thus $R_j\in \overline{\mathrm{eBH}}_\alpha^{\mathsf{m}}(\mathbf y)$ for $j=1,2$.

Define two point procedures $\mathfrak D_1$ and $\mathfrak D_2$ by $\mathfrak D_j(\mathbf y_\sigma)=\sigma^{-1}(R_j), j=1,2$ for every permutation $\sigma$ on $\cK$, and set $\mathfrak D_j(\mathbf e)=\varnothing$ outside this orbit. These two procedures are symmetric. Moreover, their outputs belong to the corresponding $\overline{\mathrm{eBH}}_\alpha^{\mathsf{m}}$ collection at every input. Therefore both $\mathfrak D_1$ and $\mathfrak D_2$ belong to $\mathrm{PP}_\alpha^{\mathrm{sym}}$. However, they are incomparable at $\mathbf y$, because $\mathfrak D_1(\mathbf y)=R_1$ and $\mathfrak D_2(\mathbf y)=R_2$. Suppose, for contradiction, that a largest element $\mathfrak G$ of $\mathrm{PP}_\alpha^{\mathrm{sym}}$ exists. Then $\mathfrak G$ must dominate both $\mathfrak D_1$ and $\mathfrak D_2$, and so $\mathfrak G(\mathbf y)=\cK$.

Pick the true null set $N=\{1,\dots,K-1\}$. Since $\max_{i\in N}y_i< (K-1) / (\alpha K)$, there exists $p\in(0,1]$ such that
$p y_i\le1$ for all $i\in N$ and $p (K-1) / K>\alpha$.
Define 
$$
\mathbf X
=
\begin{cases}
\mathbf y, & \text{with probability }p,\\
\mathbf 0, & \text{with probability }1-p.
\end{cases}
$$
For every $i\in N$, $\E [X_i]=p y_i\le1$. But
$$
\mathrm{FDR}_{\mathfrak G}(\mathbf X)
\ge
p\,\mathrm{FDP}_N(\mathcal K)
=
p\frac{K-1}{K}
>
\alpha.
$$
Therefore no largest element exists.
\end{proof}

Next, we prove  Theorem~\ref{thm:constant-term-admissibility-criterion}. We first need a technical lemma. 

\begin{lemma}
Suppose \eqref{eq:simple-constant-cap}--\eqref{eq:simple-constant-smooth}. For all $1\le b<n\le K$, we have 
\begin{equation}\label{eq:simple-pairwise-lambda}
\lambda_0^b<\lambda_0^n+\frac{n-b}{n}\frac{1-\lambda_0^n}{1-\alpha}.
\end{equation}
\end{lemma}

\begin{proof}
If $b\le\alpha n$, then \eqref{eq:simple-pairwise-lambda} directly follows, since the right-hand side of \eqref{eq:simple-pairwise-lambda} is at least $1$ and $\lambda_0^b<1$. It remains to consider $b>\alpha n$. Since $b<n$, we have $\alpha<b/n\le b/(b+1)$, and hence $b+1>1/(1-\alpha)$. Thus, for all $j=b+1,\dots,n$, we have $j(1-\alpha)>1$, and by \eqref{eq:simple-constant-smooth},
$$1-\lambda_0^{j-1}>(1-\lambda_0^j)\left(1-\frac{1}{j(1-\alpha)}\right).$$
Iterating gives
$$1-\lambda_0^b>(1-\lambda_0^n)\prod_{j=b+1}^n\left(1-\frac{1}{j(1-\alpha)}\right).
$$
Moreover, by induction we obtain
$$\prod_{j=b+1}^n\left(1-\frac{1}{j(1-\alpha)}\right)\ge1-\frac{n-b}{n(1-\alpha)}.$$
Therefore $1-\lambda_0^b>(1-\lambda_0^n)\{1-(n-b)/(n(1-\alpha))\}$, which is equivalent to the upper bound in \eqref{eq:simple-pairwise-lambda}.
\end{proof}

\begin{proof}[Proof of Theorem~\ref{thm:constant-term-admissibility-criterion}]
It suffices to show that if $\mathcal G\in\mathrm{SP}_\alpha$ satisfies $\overline{\mathrm{eBH}}_\alpha^{\bm{\lambda}}\preceq_{\mathrm w}\mathcal G$, then $\mathcal G\preceq_{\mathrm s}\overline{\mathrm{eBH}}_\alpha^{\bm{\lambda}}$. Suppose otherwise for contradiction. Then there exist $\mathbf e=(e_1,\dots,e_K)$, $R\in\mathcal G(\mathbf e)\setminus\overline{\mathrm{eBH}}_\alpha^{\bm{\lambda}}(\mathbf e)$ and a nonempty $A\subseteq\cK$ such that
$$\mathbb M_{\bm{\lambda}^A}(\mathbf e^{A})<\frac{|A\cap R|}{\alpha |R|}.$$
This implies $e_i<\infty$ for all $i\in A$. Write $n=|A|$. We split into two cases.

\textbf{Case 1}: $\lambda_0^n=0$. Choose $p>0$ so small that $p e_i<1$ for all $i\in A$ and, with $q=\alpha(1-p \bar e_A)$, $p+q<1$. Define $\mathbf y=(y_1,\dots,y_K)$ by
$$
y_i=
\begin{cases}
\dfrac{1-p e_i}{q}, & i\in A,\\[1ex]
\infty, & i\notin A.
\end{cases}
$$
We claim that $A\in\overline{\mathrm{eBH}}_\alpha^{\bm{\lambda}}(\mathbf y)$ and $\cK\in\overline{\mathrm{eBH}}_\alpha^{\bm{\lambda}}(\mathbf y)$ for sufficiently small $p$.

It suffices to check nonempty $B\subseteq A$, since $\mathbb M_{\bm{\lambda}^B}(\mathbf y^B)=\infty$ whenever $B\setminus A\ne\varnothing$. Write $b=|B|$. If $B=A$, then $\mathbb M_{\bm{\lambda}^B}(\mathbf y^B)=1/\alpha$. If $B\subsetneq A$, then $\mathbb M_{\bm{\lambda}^B}(\mathbf y^B)\to \lambda_0^b+(1-\lambda_0^b)/\alpha$ as $p\downarrow0$. By \eqref{eq:simple-pairwise-lambda} with $\lambda_0^n=0$,
$$\lambda_0^b+(1-\lambda_0^b)\frac1\alpha>\frac{b}{\alpha n}.$$
Hence, for sufficiently small $p$, $\mathbb M_{\bm{\lambda}^B}(\mathbf y^B)\ge b/(\alpha n)$ for all nonempty $B\subseteq A$. This gives $A\in\overline{\mathrm{eBH}}_\alpha^{\bm{\lambda}}(\mathbf y)$. The same inequalities also give $\cK\in\overline{\mathrm{eBH}}_\alpha^{\bm{\lambda}}(\mathbf y)$. Since $\cK\in\overline{\mathrm{eBH}}_\alpha^{\bm{\lambda}}(\mathbf y)$, $\overline{\mathrm{eBH}}_\alpha^{\bm{\lambda}}\preceq_{\mathrm w}\mathcal G$ forces $\overline{\mathrm{eBH}}_\alpha^{\bm{\lambda}}(\mathbf y)\subseteq\mathcal G(\mathbf y)$, and in particular $A\in\mathcal G(\mathbf y)$.

Set $A$ as the true null set and define
$$
\mathbf X=
\begin{cases}
\mathbf e, & \text{with probability }p,\\
\mathbf y, & \text{with probability }q,\\
\mathbf 0, & \text{with probability }1-p-q.
\end{cases}
$$
For every $i\in A$, $\E[X_i]=p e_i+q y_i=1$. Since $R\in\mathcal G(\mathbf e)$ and $A\in\mathcal G(\mathbf y)$,
$$\FDR_\mathcal G(\mathbf X)\ge p\frac{|A\cap R|}{|R|}+q=\alpha+\alpha p\left\{\frac{|A\cap R|}{\alpha |R|}-\mathbb M_{\bm{\lambda}^A}(\mathbf e^{A})\right\}>\alpha.$$

\textbf{Case 2}: $\lambda_0^n>0$. This implies $n\le K-\lfloor1/\alpha\rfloor$, so we may choose $L\subseteq\cK$ such that $|L|=\lfloor1/\alpha\rfloor+1$ and $|A\cap L|=1$. Let $a=(1/\alpha-\lambda_0^n) / (1-\lambda_0^n)$ and $h=(\tau-\lambda_0^n) / (1-\lambda_0^n)$, where $\tau=1/(\alpha (\lfloor1/\alpha\rfloor+1))$.
Then $a>1$ and $0\le h<1$. Choose $p>0$ sufficiently small and define
$$
q_+=\frac{1-p \bar e_A-(1-p)h}{a-h}
\quad\text{and}\quad
q_-=1-p-q_+.
$$
At $p=0$, we have $q_+=(1-h)/(a-h)>0$ and $q_-=(a-1)/(a-h)>0$. Hence $q_+>0$ and $q_->0$ for sufficiently small $p$. Define $\mathbf y^+=(y_1^+,\dots,y_K^+)$ and $\mathbf y^-=(y_1^-,\dots,y_K^-)$ by
$$
y_i^-=
\begin{cases}
h, & i\in A,\\[1ex]
\infty, & i\notin A,
\end{cases}
\quad\text{and}\quad
y_i^+=
\begin{cases}
\dfrac{1-p e_i-q_-h}{q_+}, & i\in A,\\[1ex]
\infty, & i\notin A.
\end{cases}
$$

We first claim that $A\in\overline{\mathrm{eBH}}_\alpha^{\bm{\lambda}}(\mathbf y^+)$ and $\cK\in\overline{\mathrm{eBH}}_\alpha^{\bm{\lambda}}(\mathbf y^+)$ for sufficiently small $p$. It suffices to check nonempty $B\subseteq A$. If $B=A$, then $\mathbb M_{\bm{\lambda}^B}((y_i^+)_{i\in B})=1/\alpha$. If $B\subsetneq A$ and $b=|B|$, then $\mathbb M_{\bm{\lambda}^B}((y_i^+)_{i\in B})\to \lambda_0^b+(1-\lambda_0^b)a$ as $p\downarrow0$. By the upper bound in \eqref{eq:simple-pairwise-lambda}, we get
$$\lambda_0^b+(1-\lambda_0^b)a>\frac{b}{\alpha n}.$$
Hence, for sufficiently small $p$, $\mathbb M_{\bm{\lambda}^B}((y_i^+)_{i\in B})\ge b/(\alpha n)$ for all nonempty $B\subseteq A$. This gives $A\in\overline{\mathrm{eBH}}_\alpha^{\bm{\lambda}}(\mathbf y^+)$. The same inequalities also give $\cK\in\overline{\mathrm{eBH}}_\alpha^{\bm{\lambda}}(\mathbf y^+)$. Then $\overline{\mathrm{eBH}}_\alpha^{\bm{\lambda}}\preceq_{\mathrm w}\mathcal G$ forces $\overline{\mathrm{eBH}}_\alpha^{\bm{\lambda}}(\mathbf y^+)\subseteq\mathcal G(\mathbf y^+)$. In particular, $A\in\mathcal G(\mathbf y^+)$.

Next we claim that, if a nonempty $S\subseteq\cK$ satisfies $|A\cap S|/ (\alpha |S|) \le\tau$, then $S\in\overline{\mathrm{eBH}}_\alpha^{\bm{\lambda}}(\mathbf y^-)$. It still suffices to check nonempty $B\subseteq A$. Write $b=|B|$. Since $\lambda_0^b\ge\lambda_0^n$ and $h<1$, we have $\mathbb M_{\bm{\lambda}^B}((y_i^-)_{i\in B})=\lambda_0^b+(1-\lambda_0^b)h\ge\lambda_0^n+(1-\lambda_0^n)h=\tau$. Therefore
$$
\mathbb M_{\bm{\lambda}^B}((y_i^-)_{i\in B})
\ge
\tau
\ge
\frac{|A\cap S|}{\alpha |S|}
\ge
\frac{|B\cap S|}{\alpha |S|}.
$$
This gives $S\in\overline{\mathrm{eBH}}_\alpha^{\bm{\lambda}}(\mathbf y^-)$. Applying this claim to $S=L$ gives $L\in\overline{\mathrm{eBH}}_\alpha^{\bm{\lambda}}(\mathbf y^-)$ and $\FDP_A(L)=1/(\lfloor1/\alpha\rfloor+1)$. By $\overline{\mathrm{eBH}}_\alpha^{\bm{\lambda}}\preceq_{\mathrm w}\mathcal G$, either $\overline{\mathrm{eBH}}_\alpha^{\bm{\lambda}}(\mathbf y^-)\subseteq\mathcal G(\mathbf y^-)$, in which case $L\in\mathcal G(\mathbf y^-)$, or $\overline{\mathrm{eBH}}_\alpha^{\bm{\lambda}}(\mathbf y^-)^\downarrow\subsetneq\mathcal G(\mathbf y^-)^\downarrow$. In the second case, choose $U\in\mathcal G(\mathbf y^-)^\downarrow\setminus\overline{\mathrm{eBH}}_\alpha^{\bm{\lambda}}(\mathbf y^-)^\downarrow$ and then choose $S\in\mathcal G(\mathbf y^-)$ such that $U\subseteq S$. Thus $S\notin\overline{\mathrm{eBH}}_\alpha^{\bm{\lambda}}(\mathbf y^-)^\downarrow$. The claim therefore forces $|A\cap S|/ (\alpha |S|)>\tau$, and hence $\FDP_A(S)>1/(\lfloor1/\alpha\rfloor+1)$. Thus, in all cases,
$$\max_{S\in\mathcal G(\mathbf y^-)}\FDP_A(S)\ge\frac1{\lfloor1/\alpha\rfloor+1}.$$

Set $A$ as the true null set and define
$$
\mathbf X=
\begin{cases}
\mathbf e, & \text{with probability }p,\\
\mathbf y^+, & \text{with probability }q_+,\\
\mathbf y^-, & \text{with probability }q_-.
\end{cases}
$$
For every $i\in A$, $\E[X_i]=p e_i+q_+y_i^++q_-y_i^-=1$. Moreover,
$$
\frac{q_+}{\alpha}+q_-\tau
=
1-p \mathbb M_{\bm{\lambda}^A}(\mathbf e^{A}).
$$
Thus we get
$$
\begin{aligned}
\FDR_\mathcal G(\mathbf X)
&\ge
p\frac{|A\cap R|}{|R|}
+
q_+
+
q_-\frac1{\lfloor1/\alpha\rfloor+1}  \\
&=
p\frac{|A\cap R|}{|R|}
+
\alpha (1-p \mathbb M_{\bm{\lambda}^A}(\mathbf e^{A})) \\
&=
\alpha+\alpha p\left(
\frac{|A\cap R|}{\alpha |R|}
-
\mathbb M_{\bm{\lambda}^A}(\mathbf e^{A})
\right)
>
\alpha.
\end{aligned}
$$
Thus we obtain the desired contradiction.
\end{proof}

Next, we prove Proposition~\ref{prop:beta-threshold-admissibility}.  
For $n\in\mathcal K$, set
$$\rho_n=\max\left\{\frac{\ell}{\alpha r}:1\le \ell\le n,\ \ell\le r\le K-n+\ell,\ \frac{\ell}{\alpha r}<1\right\},$$
with $\rho_n=0$ if the set is empty. 

\begin{lemma}\label{lem:beta-star-rho-equivalence}
For every \(\beta\in(0,1]\), $\beta\le\beta^*$ is equivalent to $\lambda_0^n(\beta)\le \rho_n$ for every $n\in\mathcal K$ such that $n<\beta K$.
\end{lemma}

\begin{proof}
We first show that 
\begin{equation}\label{E15}
    K\beta^*=\min_{1\le n\le K}\frac{n(1-\alpha\rho_n)}{1-\rho_n}.
\end{equation}
Denote $\gamma_n=\alpha\rho_n$, and define $h_\alpha(x)=(\alpha - \alpha x) / (\alpha-x),\ 0\le x<\alpha$. We have $n(1-\alpha\rho_n) / (1-\rho_n)=n h_\alpha(\gamma_n)$. Then we show $K\beta^*=\min_{1\le n\le K}n h_\alpha(\gamma_n)$.
\begin{itemize}
\item[(i)] Suppose $n\le K-\lfloor1/\alpha\rfloor$. Taking $\ell=1$ and $r=\lfloor1/\alpha\rfloor+1$ gives $\gamma_n\ge 1 / (\lfloor1/\alpha\rfloor+1)$. This lower bound is attained at $n=1$ over this region, because $\gamma_1=1/(\lfloor1/\alpha\rfloor+1)$. Since $h_\alpha$ is increasing, 
$$n h_\alpha(\gamma_n)\ge h_\alpha\left(\frac1{\lfloor1/\alpha\rfloor+1}\right)=\frac{\alpha \lfloor1/\alpha\rfloor}{\alpha(\lfloor1/\alpha\rfloor+1)-1}.$$

\item[(ii)] Suppose $n\ge K-\lfloor1/\alpha\rfloor+1$. Then $K-n+1\le \lfloor1/\alpha\rfloor$. For any 
$r\le K-n+\ell$, we have
$$\frac{\ell}{r}\ge\frac{\ell}{K-n+\ell}\ge\frac1{K-n+1}\ge\frac1{\lfloor1/\alpha\rfloor}\ge\alpha,$$
so no pair can satisfy $\ell/(\alpha r)<1$, and hence $\gamma_n=0$. From $n h_\alpha(0)=n$, the smallest value of $n h_\alpha(\gamma_n)$ over this region is $\max\{1,K-\lfloor1/\alpha\rfloor+1\}$.
\end{itemize}
Combining the two regions yields \eqref{E15}.

It remains to prove the equivalence. For $n<\beta K$, we have
$$\lambda_0^n(\beta)\le \rho_n\quad\Longleftrightarrow\quad\frac{\beta K-n}{\beta K-\alpha n}\le \rho_n\quad\Longleftrightarrow\quad\beta K\le\frac{n(1-\alpha\rho_n)}{1-\rho_n}.$$
Therefore, if $\beta\le\beta^*$, then $\lambda_0^n(\beta)\le \rho_n$ for every $n\in\mathcal K$ such that $n<\beta K$. Conversely, suppose $\lambda_0^n(\beta)\le \rho_n$ for every $n\in\mathcal K$ such that $n<\beta K$. If $n<\beta K$, then 
$n(1-\alpha\rho_n) / (1-\rho_n) \ge \beta K$. If $n\ge\beta K$, we also have 
$n(1-\alpha\rho_n) / (1-\rho_n) \ge n\ge \beta K$. Hence $\beta\le\beta^*$ follows from \eqref{E15}.
\end{proof}

\begin{proof}[Proof of Proposition~\ref{prop:beta-threshold-admissibility}]
\underline{The ``if'' direction}: Suppose $\beta\le\beta^*$. We verify that the constants $\lambda_0^n(\beta)$ satisfy
\eqref{eq:simple-constant-cap}, \eqref{eq:simple-constant-tail} and \eqref{eq:simple-constant-smooth} respectively. 

\begin{itemize}
\item[(i)] The sequence $\lambda_0^n(\beta)$ is decreasing in $n$. If $\beta K\le 1$, \eqref{eq:simple-constant-cap} directly follows. If $\beta K > 1$, \eqref{eq:simple-constant-cap} follows from $\beta\le \beta^*$ and the definition \eqref{eq:beta-star-explicit}.

\item[(ii)] Since $\beta K\le \max\{1,K-\lfloor1/\alpha\rfloor+1\}$, we know that $n>K-\lfloor1/\alpha\rfloor$ implies $n\ge \max\{1,K-\lfloor1/\alpha\rfloor+1\}\ge \beta K$. Thus \eqref{eq:simple-constant-tail} holds.

\item[(iii)] If $n<\beta K$, then
$$\lambda_0^{n-1}(\beta)-\lambda_0^n(\beta)=\frac{\beta K(1-\alpha)}{(\beta K-\alpha(n-1))(\beta K-\alpha n)}<\frac1{\beta K-\alpha n}=\frac{1-\lambda_0^n(\beta)}{n(1-\alpha)}.$$
If $n-1\ge \beta K$, there is nothing to prove. If $n-1<\beta K\le n$, then $\lambda_0^n(\beta)=0$, and writing $\beta K=n-\delta$ with $\delta\in[0,1)$ gives
$$\lambda_0^{n-1}(\beta)=\frac{1-\delta}{n(1-\alpha)+\alpha-\delta}<\frac1{n(1-\alpha)}.$$
Thus \eqref{eq:simple-constant-smooth} holds.
\end{itemize}

Therefore Theorem~\ref{thm:constant-term-admissibility-criterion} implies that $\overline{\mathrm{eBH}}_\alpha^{\bm{\lambda}(\beta)}$ is admissible.

\underline{The ``only if'' direction}: Assume $\beta>\beta^*$. By Lemma~\ref{lem:beta-star-rho-equivalence}, there exists $n\in \cK$ with $n< \beta K$ such that $\lambda_0^n(\beta)>\rho_n$. Define
$$
d_m=
\begin{cases}
\rho_n, & m=n,\\
\lambda_0^m(\beta), & m\ne n.
\end{cases}
$$
Let $\widetilde{\mathcal D}$ be the symmetric weighted-mean $\overline{\mathrm{eBH}}_\alpha$ procedure based on the e-collection
$$
\widetilde E_A(\mathbf e)
=
\begin{cases}
d_{|A|}+(1-d_{|A|})\bar e_A,
& e_i<\infty\text{ for all }i\in A,\\
\infty,
& e_i=\infty\text{ for some }i\in A.
\end{cases}
$$

We first show $\overline{\mathrm{eBH}}_\alpha^{\bm{\lambda}(\beta)}\preceq_{\mathrm s}\widetilde{\mathcal D}$. Take any $\mathbf e$ and any $R\in\overline{\mathrm{eBH}}_\alpha^{\bm{\lambda}(\beta)}(\mathbf e)$. Only the sets $A$ with $|A|=n$ need to be checked. If $e_i=\infty$ for some $i\in A$, then $\widetilde E_A(\mathbf e)=\infty$, and there is nothing to prove. Hence assume $e_i<\infty$ for all $i\in A$ and put $t=|A\cap R|/(\alpha(|R|\vee1))$.

\begin{itemize}
\item[(a)] If $t=0$, there is nothing to prove.

\item[(b)] If $0<t<1$, then we can write $t=\ell/(\alpha r)$, where $1\le \ell=|A\cap R|\le n$ and $\ell\le r=|R|\le K-n+\ell$. By definition $t\le\rho_n=d_n$, and so $\widetilde E_A(\mathbf e)\ge d_n\ge t$.

\item[(c)] If $t=1$, then $\mathbb M_{\bm{\lambda}^A(\beta)}(\mathbf e^{A})=\lambda_0^n(\beta)+(1-\lambda_0^n(\beta))\bar e_A\ge1$ implies $\bar e_A\ge1$, and therefore $\widetilde E_A(\mathbf e)=d_n+(1-d_n)\bar e_A\ge 1$ follows from $\lambda_0^n(\beta)>d_n$.

\item[(d)] If $t>1$, then from $\lambda_0^n(\beta)>d_n$, we get $(t-d_n) / (1-d_n) < (t-\lambda_0^n(\beta)) / (1-\lambda_0^n(\beta))$.
Therefore $\mathbb M_{\bm{\lambda}^A(\beta)}(\mathbf e^{A})=\lambda_0^n(\beta)+(1-\lambda_0^n(\beta))\bar e_A\ge t$ implies $\widetilde E_A(\mathbf e)=d_n+(1-d_n)\bar e_A\ge t$. 
\end{itemize}
Thus $R\in\widetilde{\mathcal D}(\mathbf e)$, and $\overline{\mathrm{eBH}}_\alpha^{\bm{\lambda}(\beta)}\preceq_{\mathrm s}\widetilde{\mathcal D}$.

We next show $\widetilde{\mathcal D}\not\preceq_{\mathrm s}\overline{\mathrm{eBH}}_\alpha^{\bm{\lambda}(\beta)}$. Choose $A\subseteq\cK$ with $|A|=n$. For every $B\subsetneq A$ with $|B|=b$, we have
$$\lambda_0^b(\beta)+(1-\lambda_0^b(\beta))\frac{\beta K}{\alpha n}>\frac{b}{\alpha n}.$$
From $d_n<\lambda_0^n(\beta)$ and $n<\beta K$, we get
$$\frac{1/\alpha-d_n}{1-d_n}<\frac{1/\alpha-\lambda_0^n(\beta)}{1-\lambda_0^n(\beta)}=\frac{\beta K}{\alpha n}.$$
Choose $q<\beta K/(\alpha n)$ sufficiently close to $\beta K/(\alpha n)$ such that $q\ge(1/\alpha-d_n)/(1-d_n)$ and for every $B\subsetneq A$,
\begin{equation}\label{E13}
    \lambda_0^b(\beta)+(1-\lambda_0^b(\beta))q\ge\frac{b}{\alpha n}.
\end{equation}
Define $\mathbf y=(y_1,\dots,y_K)$ by
\begin{equation*}
y_i=
\begin{cases}
q, & i\in A,\\[1ex]
\infty, & i\notin A.
\end{cases}.
\end{equation*}
We claim $A\in\widetilde{\mathcal D}(\mathbf y)$. Again, it suffices to check $B\subseteq A$. If $B=A$, $\widetilde E_A(\mathbf y)=d_n+(1-d_n)q\ge1/\alpha$; if $B\subsetneq A$, \eqref{E13} applies. However, $A\notin\overline{\mathrm{eBH}}_\alpha^{\bm{\lambda}(\beta)}(\mathbf y)$, because
$\mathbb M_{\bm{\lambda}^A(\beta)}(\mathbf y^A)=\lambda_0^n(\beta)+(1-\lambda_0^n(\beta))q<1 / \alpha$.
Therefore $\widetilde{\mathcal D}\not\preceq_{\mathrm s}\overline{\mathrm{eBH}}_\alpha^{\bm{\lambda}(\beta)}$, and so $\widetilde{\mathcal D}$ strictly dominates $\overline{\mathrm{eBH}}_\alpha^{\bm{\lambda}(\beta)}$ under strong dominance. Thus $\overline{\mathrm{eBH}}_\alpha^{\bm{\lambda}(\beta)}$ is not admissible when $\beta>\beta^*$.
\end{proof}

\section{Examples}\label{app:counterexamples}

We present five examples that   support  various statements made in the main paper.

\begin{example}\label{ex:same_procedure_different_lambda}
This counterexample shows that two distinct choices of the parameter $\bm \lambda$ may induce exactly the same
weighted-mean $\overline{\mathrm{eBH}}$ procedure.

Let $K=2$ and $\alpha=1/4$. Consider two $\overline{\mathrm{eBH}}$ procedures $\mathcal D$ and $\mathcal D'$ at level $\alpha$ with the corresponding weighted-mean e-collections given by $\{E_A\}_{A\subseteq\cK}$ and
$\{E'_A\}_{A\subseteq\cK}$ respectively. When the corresponding
coordinates are finite, define $E_{\{i\}}(e_i)=E'_{\{i\}}(e_i)=1$ for every $i\in \{1,2\}$, $E_{\{1,2\}}(e_1,e_2)=1$
and $E'_{\{1,2\}}(e_1,e_2)=e_1$. Set $E_A(\mathbf e)=E'_A(\mathbf e)=\infty$ whenever
$\max_{i\in A}e_i=\infty$ for every nonempty $A\subseteq \{1,2\}$. We claim that
$$\mathcal D(\mathbf e)=\mathcal D'(\mathbf e)=2^{\{i\in\{1,2\}:e_i=\infty\}}.$$
Indeed, suppose that $R\nsubseteq \{i\in\{1,2\}:e_i=\infty\}$. Choose
$j\in R\setminus \{i\in\{1,2\}:e_i=\infty\}$. Since $e_j<\infty$, take $A=\{j\}$ and then we have
$$E_{\{j\}}(e_j)=E'_{\{j\}}(e_j)=1<\frac{1}{\alpha|R|}=\frac{4}{|R|}.$$
Hence $R$ is rejected by neither procedure. Conversely, suppose that $R\subseteq \{i\in\{1,2\}:e_i=\infty\}$. If $A\cap R=\varnothing$,
then $\mathrm{FDP}_A(R)=0$. If $A\cap R\ne\varnothing$, then $A$ contains an
index $i\in R$ with $e_i=\infty$, and therefore
$E_A(\mathbf e)=E'_A(\mathbf e)=\infty$. Thus $R$ belongs to both procedures.
\end{example}

\begin{example}\label{ex:point-version-need-not-dominate-ebh}
This example shows that, even when the condition $\alpha\lambda_0^A+K\lambda_i^A\ge1$ holds for all $A\subseteq\cK$ and $i\in A$, not every
induced point weighted-mean $\overline{\mathrm{eBH}}$ procedure dominates the base eBH procedure. It also shows that, when $\alpha\ge1/2$, the condition is not necessary for the existence of an induced point procedure that dominates the base eBH procedure.

For the first claim, take $K=4$ and $\alpha=2/3$. Define a weighted-mean e-collection as follows. When all coordinates indexed by $A$ are finite, set $E_{\{4\}}(\mathbf e)=0.5+0.5 e_4$, $E_{\{2,4\}}(\mathbf e)=0.3+0.2 e_2+0.5 e_4$, and the ordinary arithmetic mean for all other $A$. By direct verification, the condition in Proposition~\ref{prop:esc-dominated-by-weighted-ebh} holds.
At the input $\mathbf e=(4,2,4,0)$, the base eBH procedure rejects $\{1,2,3\}$. The weighted-mean $\overline{\mathrm{eBH}}$ procedure rejects $\{\varnothing,\{1\},\{3\},\{1,3\},\{1,2,3\},\{1,3,4\}\}$. Therefore the maximal elements are $\{1,2,3\}$ and $\{1,3,4\}$. Thus the induced point procedure may reject $\{1,3,4\}$, which is incomparable with $\{1,2,3\}$.

For the second claim, take $K=2$ and $\alpha\in[1/2,1)$. Define a weighted-mean e-collection as follows. Choose $c\in[0,1]$ such that
$c\ge 1 / (2\alpha)$ and $c>1 / (2-\alpha)$. When all coordinates indexed by $A$ are finite, set $E_{\{i\}}(\mathbf e)=c+(1-c)e_i$ for $i\in\{1,2\}$, and $E_{\{1,2\}}$ as the ordinary arithmetic mean. For $A=\{i\}$, we have
$\alpha\lambda_0^A+K\lambda_i^A=\alpha c+2(1-c)=2-(2-\alpha)c<1$. Thus $\alpha\lambda_0^A+K\lambda_i^A\ge1$ fails.
If $e_1+e_2<2/\alpha$, the procedure $\overline{\mathrm{eBH}}_\alpha^{\bm{\lambda}}$ cannot reject any nonempty set. If $e_1+e_2\ge2/\alpha$, since $E_{\{1,2\}}(\mathbf e)\ge1/\alpha$ and $E_{\{i\}}(\mathbf e)\ge c\ge 1 / (2\alpha)$ for $i\in\{1,2\}$, we obtain $\{1,2\}\in \overline{\mathrm{eBH}}_\alpha^{\bm{\lambda}}(\mathbf e)$. Therefore, the unique maximal set is $\varnothing$ if $e_1+e_2<2/\alpha$ and $\{1,2\}$ if $e_1+e_2\ge2/\alpha$. Whenever the base eBH procedure rejects a nonempty set, either $\max(e_1,e_2)\ge2/\alpha$ or $\min(e_1,e_2)\ge1/\alpha$. In either case, $e_1+e_2\ge2/\alpha$. Moreover, at the input $\mathbf e=(2/\alpha,0)$, the base eBH procedure rejects only $\{1\}$, whereas every induced point procedure rejects $\{1,2\}$. Thus every induced point procedure strictly dominates $\mathrm{eBH}_\alpha$.
\end{example}

\begin{example}\label{ex:zero-weight-point-not-admissible}

This counterexample shows that the strict-positivity condition in Theorem~\ref{thm:positive-point-weighted-ebh-admissible} is essential. Removing it would cause inadmissibility.

Let $K=3$ and $\alpha\in[1/3,1/2)$. Define a constant-free weighted-mean e-collection as follows. When all coordinates indexed by $A$ are finite, set $E_{\{i\}}=e_i$ for every singleton $\{i\}\subseteq \cK$, and set $E_{\{1,2\}}(\mathbf e)=e_2, E_{\{1,3\}}(\mathbf e)=e_3, E_{\{2,3\}}(\mathbf e)=E_{\cK}(\mathbf e)=(e_2+e_3) / 2$. Let $\overline{\mathrm{eBH}}_\alpha^{\bm{\lambda}}$ be the corresponding weighted-mean $\overline{\mathrm{eBH}}$ procedure. Let $\overline{\mathrm{ebh}}_\alpha^{\bm \lambda}$ be any point weighted-mean $\overline{\mathrm{eBH}}$ procedure induced by $\overline{\mathrm{eBH}}_\alpha^{\bm{\lambda}}$ that rejects $\{2,3\}$ if $\cK\notin\overline{\mathrm{eBH}}_\alpha^{\bm{\lambda}}(\mathbf e)$ and $\{2,3\}\in\overline{\mathrm{eBH}}_\alpha^{\bm{\lambda}}(\mathbf e)$. We show that this point procedure $\overline{\mathrm{ebh}}_\alpha^{\bm \lambda}$ is not admissible. Define 
$$
\mathfrak G(\mathbf e)
=
\begin{cases}
\cK, & \text{if }\overline{\mathrm{ebh}}_\alpha^{\bm \lambda}(\mathbf e)=\{2,3\}
\text{ and } \min(e_2,e_3)\ge1/\alpha,\\[0.6ex]
\overline{\mathrm{ebh}}_\alpha^{\bm \lambda}(\mathbf e), & \text{otherwise.}
\end{cases}
$$
By construction, $\mathfrak G$ dominates $\overline{\mathrm{ebh}}_\alpha^{\bm \lambda}$. It suffices to verify that $\mathfrak G$ strictly dominates $\overline{\mathrm{ebh}}_\alpha^{\bm \lambda}$ and $\mathfrak G\in\mathrm{PP}_\alpha$.

(i) We show that $\mathfrak G$ strictly dominates $\overline{\mathrm{ebh}}_\alpha^{\bm \lambda}$. Take
$\mathbf e=(0, 1 / \alpha, 1 / \alpha)$. A direct check gives $\{2,3\}\in\overline{\mathrm{eBH}}_\alpha^{\bm{\lambda}}(\mathbf e)$ but $\cK\notin\overline{\mathrm{eBH}}_\alpha^{\bm{\lambda}}(\mathbf e)$. Therefore $\mathfrak G(\mathbf e)=\cK$, so $\mathfrak G$ strictly dominates $\overline{\mathrm{ebh}}_\alpha^{\bm \lambda}$.

(ii) It remains to show that $\mathfrak G\in\mathrm{PP}_\alpha$. We first provide a useful fact: 
\begin{equation}
\label{eq:zero-weight-counter-no-one}
e_1<\infty \text{ and }
1\in\overline{\mathrm{ebh}}_\alpha^{\bm \lambda}(\mathbf e)
\quad\Longrightarrow\quad
\overline{\mathrm{ebh}}_\alpha^{\bm \lambda}(\mathbf e)=\cK.
\end{equation}
Suppose otherwise for contradiction. If $\overline{\mathrm{ebh}}_\alpha^{\bm \lambda}(\mathbf e)=\{1\}$, then the constraints for $\{1,2\}$ and $\{1,3\}$ give
$e_2,e_3 \ge 1 / \alpha$. Hence $\{2,3\}\in\overline{\mathrm{eBH}}_\alpha^{\bm{\lambda}}(\mathbf e)$, which contradicts the definition of $\overline{\mathrm{ebh}}_\alpha^{\bm \lambda}$. If $\overline{\mathrm{ebh}}_\alpha^{\bm \lambda}(\mathbf e)=\{1,2\}$, then the constraints for $\{1,2\}$, $\{1,3\}$, and $\cK$ give $e_2\ge 1 / \alpha, e_3\ge 1 / (2\alpha)$ and $(e_2+e_3) / 2 \ge 1 / \alpha$. These inequalities imply $\{2,3\}\in\overline{\mathrm{eBH}}_\alpha^{\bm{\lambda}}(\mathbf e)$. The case
$\overline{\mathrm{ebh}}_\alpha^{\bm \lambda}(\mathbf e)=\{1,3\}$ is symmetric. Thus, if $\cK\notin\overline{\mathrm{eBH}}_\alpha^{\bm{\lambda}}(\mathbf e)$, the definition of $\overline{\mathrm{ebh}}_\alpha^{\bm \lambda}$ forces it to choose $\{2,3\}$ rather than a set containing $1$. This proves the claim. Consequently,
\begin{equation}
\label{eq:zero-weight-counter-G-no-one}
e_1<\infty \text{ and }
1\in\mathfrak G(\mathbf e)
\quad\Longrightarrow\quad
\mathfrak G(\mathbf e)=\cK .
\end{equation}

Now we prove $\mathfrak G\in\mathrm{PP}_\alpha$. Fix a null set $N\subseteq\cK$ and write $\mathbf E=(E_1,E_2,E_3)$. 

(a) If $1\notin N$, the only difference between $\mathfrak G$ and $\overline{\mathrm{ebh}}_\alpha^{\bm \lambda}$ is the promotion from $\{2,3\}$ to $\cK$, which cannot increase the FDP. Hence
$\FDP_N(\mathfrak G(\mathbf E))\le\FDP_N(\overline{\mathrm{ebh}}_\alpha^{\bm \lambda}(\mathbf E))$, which gives
$\FDR_{\mathfrak G}(\mathbf E)\le\alpha$.

(b) If $N=\{1\}$, then \eqref{eq:zero-weight-counter-G-no-one} gives
$\FDP_N(\mathfrak G(\mathbf E))\le 1 / 3 \le\alpha$. Thus $\FDR_{\mathfrak G}(\mathbf E)\le\alpha$.

(c) If $N=\{1,2\}$, we claim that $\FDP_N(\mathfrak G(\mathbf E))\le \alpha E_2$. If $\mathfrak G(\mathbf e)=\cK$, then either $\overline{\mathrm{ebh}}_\alpha^{\bm \lambda}(\mathbf e)=\cK$ or $\overline{\mathrm{ebh}}_\alpha^{\bm \lambda}(\mathbf e)=\{2,3\}, \min(e_2,e_3)\ge1/\alpha$. In the first case, the constraint for $N$ gives $e_2\ge 2 / (3\alpha)$. In the second case, $e_2\ge1/\alpha$. In both cases,
$$\FDP_N(\mathfrak G(\mathbf e))=\frac23\le \alpha e_2.$$
If $\mathfrak G(\mathbf e)\ne\cK$, then
$\mathfrak G(\mathbf e)=\overline{\mathrm{ebh}}_\alpha^{\bm \lambda}(\mathbf e)$ and
$1\notin\overline{\mathrm{ebh}}_\alpha^{\bm \lambda}(\mathbf e)$ by
\eqref{eq:zero-weight-counter-no-one}. Since
$\overline{\mathrm{ebh}}_\alpha^{\bm \lambda}(\mathbf e)\in\overline{\mathrm{eBH}}_\alpha^{\bm{\lambda}}(\mathbf e)$, 
$$
\FDP_N(\mathfrak G(\mathbf e))
=
\FDP_N(\overline{\mathrm{ebh}}_\alpha^{\bm \lambda}(\mathbf e))
\le
\alpha e_2.
$$
Taking expectations gives $\FDR_{\mathfrak G}(\mathbf E)\le\alpha$. 

(d) The case $N=\{1,3\}$ is symmetric to $N=\{1,2\}$.

(e) If $N=\cK$, we claim that
$\FDP_N(\mathfrak G(\mathbf E))\le\alpha (E_2+E_3) /2$.
If $\mathfrak G(\mathbf e)=\overline{\mathrm{ebh}}_\alpha^{\bm \lambda}(\mathbf e)$, this follows directly from $\overline{\mathrm{ebh}}_\alpha^{\bm \lambda}(\mathbf e)\in\overline{\mathrm{eBH}}_\alpha^{\bm{\lambda}}(\mathbf e)$. If
$\mathfrak G(\mathbf e)\ne\overline{\mathrm{ebh}}_\alpha^{\bm \lambda}(\mathbf e)$, then
$\overline{\mathrm{ebh}}_\alpha^{\bm \lambda}(\mathbf e)=\{2,3\}$ and
$\{2,3\}\in\overline{\mathrm{eBH}}_\alpha^{\bm{\lambda}}(\mathbf e)$. The constraint for $\cK$ gives
$E_\cK(\mathbf e)\ge 1 / \alpha$.
Hence
$$\FDP_\cK(\mathfrak G(\mathbf e))=1\le \alpha E_\cK(\mathbf e).$$
Taking expectations gives $\FDR_{\mathfrak G}(\mathbf E)\le\alpha$.

(f) The case $N=\varnothing$ is trivial.

\end{example}

\begin{example}\label{ex:infty-convention}

This counterexample shows that removing the $\infty$ convention in Definition~\ref{def:closed-eBH-procedure} or specifying $0\cdot\infty=0$ would break Theorem~\ref{thm:admissible-weighted-ebh-dominator}.

Let $K=2$ and $\alpha\in(0,1)$. Define 
$$
\mathcal D(e_1,e_2)
=
\{\varnothing\}
\cup
\{\{1\}: e_1\ge 1/\alpha\}
\cup
\{\{2\}: e_2=\infty\}.
$$
Then $\mathcal D\in\mathrm{SP}_\alpha$. Indeed, if $N=\{1,2\}$, then
$$
\mathbb E\left[\max_{R\in\mathcal D(\mathbf E)}\mathrm{FDP}_N(R)\right]
\le
\mathbb P(E_1\ge1/\alpha)+\mathbb P(E_2=\infty)
\le \alpha.
$$
If $N=\{1\}$, then the same Markov inequality applies, which gives FDR at most $\alpha$. If $N=\{2\}$, then $\mathbb P(E_2=\infty)=0$ gives FDR $0$. The case $N=\varnothing$ is trivial. 

We claim that no weighted-mean $\overline{\mathrm{eBH}}$ procedure can strongly dominate $\mathcal D$. Suppose, for contradiction, that $\mathcal D\preceq_{\mathrm s}\overline{\mathrm{eBH}}_\alpha^{\bm{\lambda}}$ with
$\mathbb M_{\bm{\lambda}^{\{1,2\}}}(e_1,e_2)=\lambda_0 +\lambda_1 e_1+\lambda_2 e_2$.
Since $\{1\}\in\mathcal D(1/\alpha,0)$, taking $A=\{1,2\}$ gives 
$\alpha\lambda_0+\lambda_1\ge 1$. Using $\lambda_1=1-\lambda_0-\lambda_2$, this gives
$-(1-\alpha)\lambda_0-\lambda_2\ge 0$. Thus $\lambda_0=\lambda_2=0$ and $\lambda_1=1$ necessarily hold. On the other hand, since $\{2\}\in\mathcal D(0,\infty)$, taking $A=\{1,2\}$ gives
$\lambda_0+\lambda_1\cdot0+\lambda_2\cdot\infty=0\ge 1 / \alpha$, a contradiction. Therefore no weighted-mean $\overline{\mathrm{eBH}}$ procedure can strongly dominate $\mathcal D$. Under the $\infty$ convention, this gap disappears because $\mathbb M_{\bm{\lambda}^{\{1,2\}}}(0,\infty)=\infty$.

\end{example}

\begin{example}\label{ex:dominator-need-not-symmetric}

For $\mathfrak D\in\mathrm{PP}_\alpha^{\mathrm{sym}}$, this counterexample shows that one may not be able to find a symmetric point weighted-mean $\overline{\mathrm{eBH}}$ procedure $\overline{\mathrm{ebh}}_\alpha^{\bm\lambda}$ that dominates $\mathfrak D$.

Let $K=3$ and $\alpha\in(0,2/3)$. Choose $x>0$ such that $1 / 2 < \alpha x < 2 / 3$, and then choose a small $\eta>0$ such that $x-\eta> 1 / (2\alpha)$ and $x+\eta< 2 / (3\alpha)$. Define a symmetric point procedure $\mathfrak D$ as follows. If $\mathbf e$ has a unique coordinate equal to $\infty$, say coordinate $k$, and the other two coordinates, say $e_i$ and $e_j$, both belong to $(x-\eta,x+\eta)$, set 
$$
\mathfrak D(\mathbf e)
=
\begin{cases}
\{i,k\}, & e_i>e_j,\\
\{j,k\}, & e_j>e_i,\\
\{k\}, & e_i=e_j.
\end{cases}
$$
Set $\mathfrak D(\mathbf e)=\varnothing$ for all other inputs. This procedure is symmetric. By the choice of $x$ and $\eta$, we can easily verify that $\mathfrak D(\mathbf e)\in \overline{\mathrm{eBH}}_\alpha^{\mathsf{m}}(\mathbf e)$ for every $\mathbf e$. Therefore $\mathfrak D\in\mathrm{PP}_\alpha^{\mathrm{sym}}$.

Suppose, for contradiction, that there exists a symmetric point weighted-mean $\overline{\mathrm{eBH}}$ procedure $\overline{\mathrm{ebh}}_\alpha^{\bm \lambda}$ dominating $\mathfrak D$. Let $\overline{\mathrm{eBH}}_\alpha^{\bm \lambda}$ be the corresponding weighted-mean $\overline{\mathrm{eBH}}$ procedure. Put $\mathbf u=(x,x,\infty)$ and, for every sufficiently large $n$, put
$\mathbf v_n=(x+1/n,x,\infty)$ and $\mathbf w_n=(x,x+1/n,\infty)$. Then we have $\mathfrak D(\mathbf u)=\{3\}$, $\mathfrak D(\mathbf v_n)=\{1,3\}$ and $\mathfrak D(\mathbf w_n)=\{2,3\}$. Since $\overline{\mathrm{ebh}}_\alpha^{\bm \lambda}$ dominates $\mathfrak D$, we have $\overline{\mathrm{ebh}}_\alpha^{\bm \lambda}(\mathbf v_n)\in\{\{1,3\},\mathcal K\}$ and $\overline{\mathrm{ebh}}_\alpha^{\bm \lambda}(\mathbf w_n)\in\{\{2,3\},\mathcal K\}$. For each fixed $R\subseteq\mathcal K$, the set $\{\mathbf e: R\in\overline{\mathrm{eBH}}_\alpha^{\bm \lambda}(\mathbf e)\}$ is closed. Hence, by letting $n\to\infty$, either $\mathcal K\in\overline{\mathrm{eBH}}_\alpha^{\bm \lambda}(\mathbf u)$, or both $\{1,3\}$ and $\{2,3\}$ belong to $\overline{\mathrm{eBH}}_\alpha^{\bm \lambda}(\mathbf u)$.

If $\mathcal K\notin\overline{\mathrm{eBH}}_\alpha^{\bm \lambda}(\mathbf u)$, then $\{1,3\},\{2,3\}\in\overline{\mathrm{eBH}}_\alpha^{\bm \lambda}(\mathbf u)$. Since $\overline{\mathrm{ebh}}_\alpha^{\bm \lambda}$ dominates $\mathfrak D$, the set $\overline{\mathrm{ebh}}_\alpha^{\bm \lambda}(\mathbf u)$ must contain $\{3\}$. Since $\overline{\mathrm{ebh}}_\alpha^{\bm \lambda}$ is symmetric and $\mathbf u$ is invariant under the permutation exchanging $1$ and $2$, the set $\overline{\mathrm{ebh}}_\alpha^{\bm \lambda}(\mathbf u)$ must also be invariant under this permutation. The only invariant supersets of $\{3\}$ are $\{3\}$ and $\mathcal K$. But $\{3\}$ is not maximal in $\overline{\mathrm{eBH}}_\alpha^{\bm \lambda}(\mathbf u)$. Hence necessarily $\mathcal K\in\overline{\mathrm{eBH}}_\alpha^{\bm \lambda}(\mathbf u)$, and therefore $\overline{\mathrm{ebh}}_\alpha^{\bm \lambda}(\mathbf u)=\mathcal K$. We claim that this violates FDR control. Since $x<2/(3\alpha)$ and $\alpha<2/3$, there exists $p\in(0,1]$ such that $px\le 1$ and $2p / 3>\alpha$. Take the true null set $N=\{1,2\}$, and let
\[
\mathbf X=
\begin{cases}
\mathbf u, & \text{with probability }p,\\
\mathbf 0, & \text{with probability }1-p.
\end{cases}
\]
Then $\mathbb E[X_i]=px\le1$ for every $i\in N$, so these are valid null e-values. But
$$\mathrm{FDR}_{\overline{\mathrm{ebh}}_\alpha^{\bm \lambda}}(\mathbf X)\ge p\,\mathrm{FDP}_{N}(\mathcal K)=\frac{2p}{3}>\alpha.$$
Therefore there does not exist a symmetric point weighted-mean $\overline{\mathrm{eBH}}$ procedure that dominates $\mathfrak D$.
\end{example}


\begin{thebibliography}{}

\bibitem[\protect\citeauthoryear{Barber and Ramdas}{2017}]{BR17}
Barber, R. F. and Ramdas, A. (2017). The p-filter: Multilayer false discovery rate control for grouped hypotheses.
\emph{Journal of the Royal Statistical Society Series B: Statistical Methodology}, \textbf{79}(4), 1247--1268.

\bibitem[\protect\citeauthoryear{Benjamini and Hochberg}{1995}]{BH95}
Benjamini, Y. and Hochberg, Y. (1995). Controlling the false discovery rate: A practical and powerful approach to multiple testing.
\emph{Journal of the Royal Statistical Society Series B}, \textbf{57}(1), 289--300.

\bibitem[\protect\citeauthoryear{Benjamini and Yekutieli}{2001}]{BY01}
Benjamini, Y. and Yekutieli, D. (2001). The control of the false discovery rate in multiple testing under dependency.
\emph{Annals of Statistics}, \textbf{29}(4), 1165--1188.

\bibitem[\protect\citeauthoryear{Blanchard and Roquain}{2008}]{BR08}
Blanchard, G. and Roquain, E. (2008). Two simple sufficient conditions for FDR
control.
\emph{Electronic Journal of Statistics}, \textbf{2}, 963--992.

\bibitem[\protect\citeauthoryear{Clerico}{2026}]{C26}
Clerico, E. (2026).
A simple geometric proof for the characterisation of e-merging functions.
\emph{Statistics \& Probability Letters}, \textbf{236}, 110750.

\bibitem[\protect\citeauthoryear{Goeman}{2026}]{Goe26}
Goeman, J. (2026).
A uniform improvement of the Benjamini-Hochberg procedure via e-Closure.
\emph{arXiv}:2606.01854.

\bibitem[\protect\citeauthoryear{Goeman et al.}{2021}]{GHS21}
Goeman, J.~J., Hemerik, J. and Solari, A. (2021). Only closed testing procedures
are admissible for controlling false discovery proportions.
\emph{Annals of Statistics}, \textbf{49}(2), 1218--1238.

\bibitem[\protect\citeauthoryear{Gr\"unwald et al.}{Gr\"unwald et al.}{2024}]{GDK24}
 Gr\"unwald,  P., de Heide, R. and  Koolen, W.~M. (2024). Safe testing.  \emph{Journal of the Royal Statistical Society Series B}, \textbf{86}(5), 1091--1128.

\bibitem[\protect\citeauthoryear{Ignatiadis et al.}{2024}]{IWR24}
Ignatiadis, N., Wang, R. and Ramdas, A. (2024). E-values as unnormalized weights in multiple testing.
\emph{Biometrika}, \textbf{111}(2), 417--439.

\bibitem[\protect\citeauthoryear{Ignatiadis et al.}{2026}]{IWR26}
Ignatiadis, N., Wang, R. and Ramdas, A. (2026).
Tiny but uniform improvements of adaptive BH procedures via compound
e-values.
\emph{arXiv}:2603.21424.

\bibitem[\protect\citeauthoryear{Lee and Ren}{2024}]{LR24}
Lee, J. and Ren, Z. (2024).
Boosting e-BH via conditional calibration.
\emph{arXiv}:2404.17562.

\bibitem[\protect\citeauthoryear{Ramdas and Wang}{2025}]{RW25}
Ramdas, A. and Wang, R. (2025). Hypothesis testing with e-values.
\emph{Foundations and Trends in Statistics}, \textbf{1}(1--2), 1--390.

\bibitem[\protect\citeauthoryear{Ren and Barber}{Ren and Barber}{2024}]{RB24}
Ren, Z. and Barber, R. F. (2024). Derandomised knockoffs: leveraging e-values for false discovery rate control. \emph{Journal of the Royal Statistical Society Series B}, \textbf{86}(1), 122--154.

\bibitem[\protect\citeauthoryear{Shafer}{2021}]{S21}
Shafer, G. (2021). Testing by betting: A strategy for statistical and scientific communication. \emph{Journal of the Royal Statistical Society, Series A}, \textbf{184}(2), 407--431.

\bibitem[\protect\citeauthoryear{Solari and Goeman}{2017}]{SG17}
Solari, A. and Goeman, J. J. (2017). Minimally adaptive BH: A tiny but uniform improvement of the procedure of Benjamini and Hochberg. \emph{Biometrical Journal}, \textbf{59}(4), 776--780.

\bibitem[\protect\citeauthoryear{Storey}{2002}]{S02}
Storey, J. D. (2002). A direct approach to false discovery rates.
\emph{Journal of the Royal Statistical Society Series B: Statistical Methodology}, \textbf{64}(3), 479--498.

\bibitem[\protect\citeauthoryear{Vovk et al.}{Vovk et al.}{2022}]{VWW22}
Vovk, V., Wang, B. and Wang, R. (2022). Admissible ways of merging p-values
  under arbitrary dependence.  \emph{Annals of Statistics}, \textbf{50}(1), 351--375.

\bibitem[\protect\citeauthoryear{Vovk and Wang}{2021}]{VW21}
Vovk, V. and Wang, R. (2021). E-values: Calibration, combination, and applications.
\emph{Annals of Statistics}, \textbf{49}(3), 1736--1754.


\bibitem[\protect\citeauthoryear{Vovk and Wang}{Vovk and Wang}{2024}]{VW24}
{Vovk, V. and  Wang, R.} (2024). True and false discoveries with independent and sequential e-values.
\emph{Canadian Journal of Statistics}, \textbf{52}(4), e11833.

\bibitem[\protect\citeauthoryear{Wang}{2025}]{W25}
Wang, R. (2025). The only admissible way of merging arbitrary e-values.
\emph{Biometrika}, \textbf{112}(2), asaf020.

\bibitem[\protect\citeauthoryear{Wang and Ramdas}{2022}]{WR22}
Wang, R. and Ramdas, A. (2022). False discovery rate control with e-values.
\emph{Journal of the Royal Statistical Society Series B}, \textbf{84}(3), 822--852.

\bibitem[\protect\citeauthoryear{Xu et al.}{2025}]{XSF25}
Xu, Z., Solari, A., Fischer, L., de Heide, R., Ramdas, A. and Goeman, J. (2025). Bringing closure to false discovery rate control: A general principle for multiple testing. \emph{arXiv}:2509.02517v3.


\end{thebibliography}
\end{document}